\documentclass[12pt]{article}
\usepackage[left=3cm,right=3cm,top=3cm,bottom=3cm]{geometry}

\usepackage{graphicx} 
\usepackage{amsmath}
\usepackage{amsthm}
\usepackage{amsfonts}
\usepackage{amssymb}
\usepackage{xcolor}
\usepackage{mathtools}
\usepackage{algpseudocode}
\usepackage{algorithm}
\usepackage{amsthm}
\usepackage{tikz} 
\usepackage{subcaption}
\usepackage{leftindex}
\usetikzlibrary{decorations.pathreplacing}
\usepackage{natbib}
\usepackage[colorlinks=true,
            linkcolor=black,
            citecolor=black,
            urlcolor=black,
            breaklinks=true]{hyperref}
\usepackage[section]{placeins}  

\newcommand{\bb}[1]{\boldsymbol{#1}}
\newcommand{\x}{\boldsymbol{x}}
\newcommand{\X}{\boldsymbol{X}}
\newcommand{\E}{\mathbb{E}}
\newcommand{\w}{\boldsymbol{w}}
\newcommand{\Z}{\boldsymbol{Z}}

\renewcommand{\r}{\tilde{r}}

\newcommand{\n}{\tilde{n}}

\newtheorem{theorem}{Theorem}
\newtheorem{corollary}{Corollary}

\colorlet{group1}{red!20}
\colorlet{group2}{blue!20}
\colorlet{group3}{green!20}
\colorlet{outs}{black!30}
\colorlet{miss}{white}

\begin{document}

  \title{\vspace{-1cm} \bf Outlier-Robust Multi-Group Gaussian Mixture Modeling with Flexible Group Reassignment}
  \author{Patricia Puchhammer\hspace{.2cm}\\
     Institute of Statistics and Mathematical Methods in Economics, \\ Vienna University of Technology, \\
Ines Wilms \hspace{.2cm}\\
     Department of Quantitative Economics, 
  Maastricht University \\
    and \\
    Peter Filzmoser\\
     Institute of Statistics and Mathematical Methods in Economics,\\ Vienna University of Technology}
\maketitle

\begin{abstract}
Do expert-defined or diagnostically-labeled data groups align with clusters inferred through statistical modeling? If not, where do discrepancies between predefined labels and model-based groupings occur and why? In this work, we introduce the multi-group Gaussian mixture model (MG-GMM), the first model developed to investigate these questions. It incorporates prior group information while allowing flexibility to reassign observations to alternative groups based on data-driven evidence. We achieve this by modeling the observations of each group as arising not from a single distribution, but from a Gaussian mixture comprising all group-specific distributions. Moreover, our model offers robustness against cellwise outliers that may obscure or distort the underlying group structure. We propose a novel penalized likelihood approach, called cellMG-GMM, to jointly estimate mixture probabilities, location and scale parameters of the MG-GMM, and detect outliers through a penalty term on the number of flagged cellwise outliers in the objective function. We show that our estimator has good breakdown properties in presence of cellwise outliers. We develop a computationally-efficient EM-based algorithm for cellMG-GMM, and demonstrate its strong performance in identifying and diagnosing observations at the intersection of multiple groups through simulations and diverse applications in medicine and oenology.
\end{abstract}


\newpage

\section{Introduction} 

In this paper, we study the problem of Gaussian mixture modeling for data pre-partitioned into groups, where the group assignment may be uncertain or imprecise and plagued by outliers. We show how the Gaussian mixture model (GMM) can be extended to a multi-group GMM that (i) exploits prior group information while allowing each observation to be reassigned to another group when supported by the data and (ii) stays reliable in presence of outliers that obscure or distort the group structure.

Data from heterogeneous populations are increasingly common across many  applications (for example, \citealp{Lyu2025} consider mixture models for binary variables, \citealp{Sugasawa2021} propose mixture models with an additional group structure within each cluster). We consider settings where observations can be pre-partitioned into groups using expert knowledge or contextual information; e.g., medical data divided into healthy individuals and patients, or spatial data where terrain type or country borders inform group structures. Such partitioning is often only preliminary because group assignments may be uncertain or imprecise. In medicine, for instance, progressive diseases involve patients transitioning from healthy states to more sever stages of a disease; a diabetes diagnosis, for instance, is based on blood sugar measurements that typically  vary smoothly across health conditions. 

Moreover, outliers are common and may distort the group structure; in the diabetes example, they may arise from  device malfunctions during blood sugar measurement. In complex multivariate settings like ours, such outliers are easily masked and may adversely effect the analysis if undetected.  Outlier detection in a multi-group setup is, however, challenging. An observation may be outlying in its original group yet fit better in another, suggesting a mismatch and the need to reconsider its assignment. Alternatively, an observation may be ``generally'' atypical, meaning  it cannot be appropriately assigned to any group. This atypicality may stem from unusual values across all variables, a few, or even a single variable. This calls for a cellwise outlier detection procedure that flags individual cells rather than entire observations as outlying (\citealp{Alqallaf2009, raymaekers2024challenges}). 

\begin{figure}[t]
     \centering
        \begin{subfigure}[b]{0.24\textwidth}
            \centering
            \includegraphics[width = \textwidth, trim = {0.15cm 1.3cm 0.15cm 1.9cm}, clip]{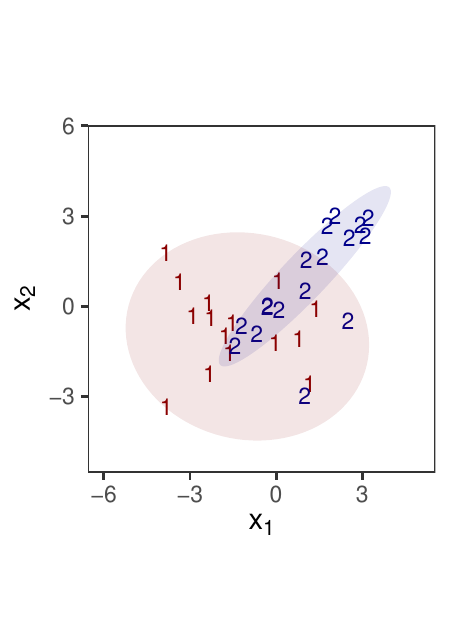}
            \caption{Data generation}
        \end{subfigure}
        \hfill
        \begin{subfigure}[b]{0.24\textwidth}
            \centering
            \includegraphics[width = \textwidth, trim = {0.15cm 0.25cm 0.15cm 0.8cm}, clip]{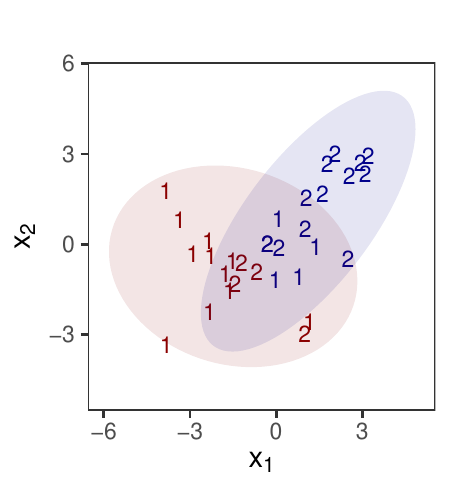}
            \caption{Classification}
        \end{subfigure}
        \hfill
        \begin{subfigure}[b]{0.24\textwidth}
            \centering
            \includegraphics[width = \textwidth, trim = {0.15cm 0.25cm 0.15cm 0.8cm}, clip]{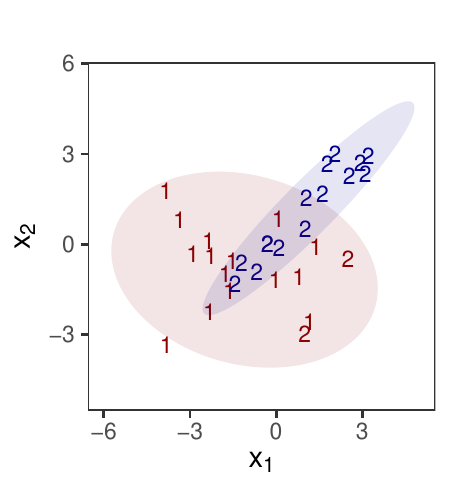}
            \caption{Multi-group GMM}
        \end{subfigure}
        \hfill
        \begin{subfigure}[b]{0.24\textwidth}
            \centering
            \includegraphics[width = \textwidth, trim = {0.15cm 0.25cm 0.15cm 0.8cm}, clip]{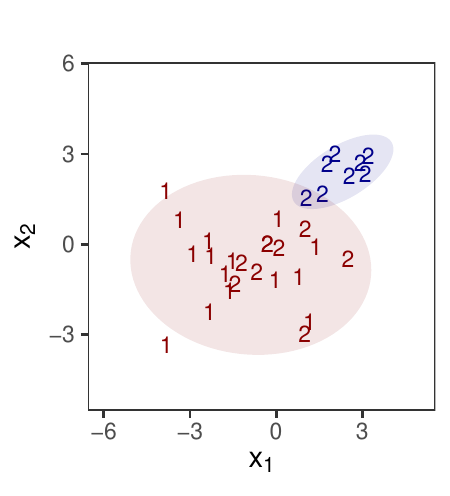}
            \caption{Clustering}
        \end{subfigure}
        \caption{Toy example with two predefined groups (labels 1/2). Colors show model-based assignments; shaded areas are groupwise tolerance ellipses.
        Left: True data generating process with two mislabeled observations.
Middle left: Quadratic discriminant analysis with fixed groups.
Middle right: Multi-group GMM with flexible reassignment.
Right: Standard GMM-based clustering.}
    \label{fig:toy}
\end{figure}

One way to study the distributional characteristics of the data is to treat the initial group partitioning as fixed and ignore potential outliers. However, this can lead to misleading inference on the group-level location and scale parameters and overlook important information on  interconnections among groups; for example, patients in transition and the factors driving that transition. Figure \ref{fig:toy} illustrates this with a toy example of two groups: panel (a) shows the real data generating process, panel (b) estimates, under normality, the mean and covariance of the two groups while treating the initial group assignments (i.e.\ label 1 or 2) as fixed, as is common in supervised classification.

Alternatively, one may  ignore the pre-assigned group structure (and the possible presence of outliers). This, however, risks overlooking important sources of variability when observations are assumed to be identically distributed, or discarding valuable expert or contextual information when using standard mixture models or clustering techniques. Figure \ref{fig:toy} panel (d) shows the result of applying an unsupervised method to the toy example (a classical GMM; \citealp[\texttt{mclust},][]{mclust_CRAN}) which does not exploit the expert-defined group labels. 

Although classification or clustering methods can be made robust to outliers (e.g., \citealp{hubert2024robust} for robust classification; \citealp{garcia2010review, coretto2016robust, coretto2017consistency} for robust clustering, and  \citealp{neykov2007, zaccaria2024} for robust GMMs),  a more flexible modeling approach is still needed -- one that incorporates expert or contextual knowledge while  allowing smooth connections among  predefined groups and flexible reassignment of observations based on data-driven evidence. We offer such a semi-supervised approach through the multi-group GMM, as visualized in Figure \ref{fig:toy} panel (c). 

We make three main contributions to the literature on Gaussian mixture modeling. 

\textbf{(i) A novel multi-group Gaussian mixture model.} We introduce a novel GMM (Section~\ref{subsec:model}); the multi-group GMM (MG-GMM) that allows for expert- or context-based initial group assignments. In contrast to standard GMMs, we do not assume each observation in the data set to be a random draw from one and the same GMM. Instead, we model each observation to have a main distribution, namely the initial group to which it is assigned, while being mixed with distributions of other groups. We hereby assume that a smooth process underlies the initial data partitioning.

\textbf{(ii) Cellwise robust estimation.} We robustify the MG-GMM to the possible presence of cellwise outliers. To this end, we propose the cellMG-GMM, a novel penalized likelihood-based estimator that adds a penalty on the flagged cellwise outliers to the objective function (Section~\ref{subsec:objective}). It jointly detects outliers and estimates the parameters of the MG-GMM. 

Recently, several proposals have successfully embedded cellwise outlier detection into penalized likelihood frameworks. Closely related to our work are  \cite{Raymaekers2023}, who introduce the cellwise minimum covariance determinant estimator, and \cite{zaccaria2024}, who propose a cellwise robust estimator for the standard GMM. Like these studies, we penalize the number of flagged cellwise outliers in the observed likelihood objective; but we do so for the newly introduced multi-group GMM. The latter offers a dual treatment of atypical observations: they may be reassigned to a better-fitting group, or labeled as outlying to all groups based on some or all variables. To our knowledge, we are the first to offer this dual treatment of cellwise outliers in GMMs. The outlier-robust MG-GMM is thus unique in revealing how observations transition between groups and which variables drive these transitions.

\textbf{(iii) Cellwise robust theory.} Our theoretical contribution consists of establishing the good finite-sample cellwise breakdown properties of cellMG-GMM. The definition of an adequate robustness measure for cluster analysis was introduced by \cite{Hennig2004} and extended to multivariate data by \cite{Cuesta2008}. Analyzing the breakdown point under well-clustered, idealized settings is crucial because in  mixture models even a single outlier can make  parameter estimates of at least one  mixture component break down. However, these earlier works address the classical rowwise contamination paradigm. To the best of our knowledge, we are the first to extend breakdown point theory for cluster and finite mixture models to the cellwise outlier paradigm; see Section~\ref{sec:theory}.

Beyond these main contributions, we provide an EM-based algorithm that integrates outlier detection -- treating outliers  as missing values that are unknown in advance --  with estimation of the mixture probabilities, and the location and scale parameters of the multi-group GMM (Section~\ref{sec:algorithm}). Our implementation is publicly available in the package \texttt{ssMRCD} \citep{ssMRCD_Cran} for the statistical computing environment \texttt{R} \citep{Rcoreteam}. Section~\ref{sec:simulations} evaluates our proposal through simulations and shows its robustness to adversarial contamination, while Section~\ref{sec:application} illustrates its value in two diverse applications. Replication files of all analyses are available at \url{https://github.com/patriciapuch/cellMG-GMM}.

\section{Outlier-Robust Multi-Group GMM}
\label{sec:methodology}
We introduce the multi-group Gaussian mixture model in Section~\ref{subsec:model}, and the corresponding penalized likelihood-based estimator called ``cellMG-GM'' in Section~\ref{subsec:objective}.

\subsection{Model and Notation}
\label{subsec:model}

Let~$\X_1, \X_2, \ldots, \X_N$ be data sets from $N$ pre-defined groups consisting of independent observations~$\X_g = ((\x_{g,1})', \ldots, (\x_{g, n_g})')' \in \mathbb{R}^{n_g \times p}$ per group~$g = 1, \ldots, N$ of the same~$p$ variables and $n = \sum_{g = 1}^N n_g$ total number of observations. We assume that observations~$\x_{g,i}$ from group~$g$, $i=1,\ldots ,n_g$, originate from a Gaussian mixture 
\begin{align}
\label{eq:model}
    \x_{g,i} \sim \mathcal{N} ( \bb{\mu}_k, \bb{\Sigma}_k) \text{ with probability } \pi_{g, k} \geq 0,
\end{align}
for~$k=1, \ldots , N$. 
In this novel multi-group GMM, or MG-GMM in short, observations of a particular group can thus originate from a Gaussian mixture of \textit{all} group distributions, this is in contrast to the standard GMM where each observation is a random draw from one and the same GMM. The mixture probabilities $\pi_{g, k}$ $(k=1,\ldots, N)$ for each group $g$ must sum to one. We do assume that each pre-specified group has a main distribution assigned to it. We thus enforce $\pi_{g,g} \geq \alpha \geq 0.5$, where the constant $\alpha$ regulates the model's strictness in terms of group reassignments. For $\alpha = 1$, reassignments are not allowed since all pre-assigned groups are then fixed (i.e.\ $\pi_{g,g}=1, \forall g$). In contrast, for $ 0.5 \leq \alpha < 1$, flexible reassignment is allowed with decreasing  $\alpha$ allowing for more and more flexibility. A more flexible MG-GMM can therefore identify observations that fall in the transition region between groups.

In the following, we introduce a penalized likelihood estimator for the MG-GMM  that is robust to the presence of cellwise outliers. Outliers will be treated as missing values in the likelihood framework such that they cannot influence the estimation process. However, unlike regular missing values, the positions of the outliers are not known in advance; the outliers need to be detected during estimation. In the remainder, we will use the following notation to denote the missingness pattern of the data. Observed and missing cells of $\x_{g,i}$ are denoted by a binary vector~$\w_{g,i} = (w_{g, i1}, \ldots, w_{g, ip})$, where a value of 1 indicates an observed data cell and 0 a missing/outlying data cell. The set of matrices $\bb{W} = (\bb{W}_g)_{g = 1}^N$ then collects all binary vectors $\w_{g,i}, i = 1, \ldots, n_g$, in the rows of each $\bb{W}_g$. These matrices are not given in advance but will be obtained during estimation. Furthermore, $\x_{g,i}^{(\w_{g,i})}$  denotes the vector with only the entries for which the variables are observed (i.e.\ $w_{g,ij}=1$ for variables $j=1,\ldots, p$), similarly for the mean $\bb{\mu}_k^{(\w_{g,i})}$. The matrix $\bb{\Sigma}_k^{(\w_{g,i})}$ denotes the submatrix of $\bb{\Sigma}_k$ containing only the rows and columns of the variables that are observed. For any binary vectors $\w$ and $\tilde{\w}$,  $\bb{\Sigma}_{k}^{(\w| \tilde{\w})}$ denotes the submatrix of $\bb{\Sigma}_{k}$ containing only the rows and columns of the observed variables indicated by $\w$ and $\tilde{\w}$ respectively. By convention, an observation consisting exclusively of missing cells (i.e.\ $\w_{g,i} = \bb 0$) has $\det (\bb{\Sigma}_k^{(\w_{g,i})}) = 1$, the squared Mahalanobis distance $(\x_{g,i}^{(\w_{g,i})} - \bb\mu_k^{(\w_{g,i})})' (\bb\Sigma_{k}^{(\w_{g,i})})^{-1} (\x_{g,i}^{(\w_{g,i})} - \bb\mu_k^{(\w_{g,i})}) = 0$, and $\varphi(\x_{g,i}^{(\w_{g,i})}; \bb{\mu}_k^{(\w_{g,i})}, \bb{\Sigma}_k^{(\w_{g,i})}) = 1$ where  $ \varphi(\x_{g,i}; \bb{\mu}_k, \bb{\Sigma}_k) $ denotes the multivariate normal density with mean $\bb{\mu}_k$ and covariance $\bb{\Sigma}_k$ of an observation $\x_{g,i}$. Finally, superscripts $(\mathbf{1}-\w)$ indicate missing cells instead of observed ones, $\{j: w_{g,ij} = 0,  j =1, \ldots, p\}$.

\subsection{cellMG-GMM: A Penalized Observed Likelihood Estimator}
\label{subsec:objective}

The parameters of the MG-GMM that need to be estimated are the mixture probabilities $\bb{\pi} = (\pi_{g, k} )_{g,k = 1}^N$, the sets of group-specific mean vectors $\bb{\mu} = ( \bb{\mu}_k)_{k = 1}^N$, and scale parameters $ \bb{\Sigma} = ( \bb{\Sigma}_k)_{k = 1}^N$. To simultaneously estimate these MG-GMM parameters and detect the outliers, hence estimate $\bb{W}$, we use a penalized observed likelihood approach.

We consider the \textit{observed likelihood} (\citealp{Dempster1977} and \citealp{Little2019} for the Gaussian model) which  removes the missing values from the likelihood estimation, in combination with a \textit{penalty term} on the number of flagged cellwise outliers; similar in spirit to \cite{Raymaekers2023} for cellwise robust covariance estimation and \cite{zaccaria2024} for cellwise robust (standard) GMMs. We propose the following \textit{observed penalized log-likelihood} $\operatorname{Obj}(\bb{\pi}, \bb{\mu}, \bb{\Sigma}, \bb{W})$ for the MG-GMM model, namely
\begin{align}
      \sum_{g = 1}^N \sum_{i = 1}^{n_g} \left[-2\ln \left(\sum_{k = 1}^N \pi_{g, k} \varphi\left(\x_{g,i}^{(\w_{g,i})}; \bb{\mu}_k^{(\w_{g,i})}, \bb{\Sigma}_{reg,k}^{(\w_{g,i})}\right)\right) + \sum_{j = 1}^{p} q_{g, ij} (1-w_{g, ij})  \right], \label{eq:oob} 
\end{align}
subject to the constraints
\begin{align}
    &\sum_{k = 1}^N \pi_{g, k}  = 1,  \quad \pi_{g,g} \geq \alpha \geq 0.5 & \forall g = 1, \ldots, N  \label{eq:ob_piinequality}\\
       &\sum_{i = 1}^{n_g} w_{g, ij} \geq h_g & \forall j= 1, \ldots, p, \forall g = 1, \ldots, N \label{eq:ob_Wconstraint}  \\ 
     &\bb{\Sigma}_{reg,k} = (1-\rho_k)\bb{\Sigma}_k + \rho_k \bb{T}_k & \forall k = 1, \ldots, N . \label{eq:ob_covreg}
\end{align}

Our estimator, the cellMG-GMM, is then obtained as the minimizer of $\operatorname{Obj}(\bb{\pi}, \bb{\mu}, \bb{\Sigma}, \bb{W})$. The first part of Objective~\eqref{eq:oob} is the observed likelihood of each observation $\x_{g,i}$ given the missingness pattern in $\w_{g,i}$. The second part contains the penalty term which discourages flagging too many cells as outlying. Flagging a cell of an observation $x_{g,ij}$ costs a value of $q_{g, ij}$ in the objective function. Intuitively, a cell $x_{g,ij}$ will be flagged as outlying iff its inclusion worsens the log-likelihood more than the cost of flagging it; this to reduce overflagging. We compute the constants $q_{g, ij}$  in Section~\ref{subsec:algorithm_parameters}; the idea is to flag a cell as outlying iff its squared standardized residual is atypically large, as measured by a $\chi^2$-quantile.

Regarding the constraints, Equation~\eqref{eq:ob_piinequality} originates from the MG-GMM introduced in Section~\ref{subsec:model}. Equation~\eqref{eq:ob_Wconstraint} limits the number of cells that can be flagged per variable $j$ and group $g$. Although one could require at least half the cells per group to be used ($h_g \geq \lceil 0.5n_g\rceil$), this may cause instabilities when estimating covariances if two variables have no overlapping observed cells \citep{Raymaekers2023}. We therefore impose $h_g = \lceil0.75n_g \rceil$ throughout, allowing at most $25\%$  flagged cells per variable $j$ and group $g$. Finally, Equation~\eqref{eq:ob_covreg} enforces regularization on all group-specific covariance matrices: each is a convex combination, with  factor $\rho_k > 0$,  of the group-specific covariance matrix  $\bb{\Sigma}_k$ and a diagonal matrix $\bb{T}_k$ containing  univariate robust scales for group $k$. This regularization is similar in spirit to the MRCD of \cite{Boudt2020}. The proposed values for $\rho_k$ and $\bb T_k$ are detailed in Section~\ref{subsec:algorithm_parameters}.

\section{Algorithm}
\label{sec:algorithm}

We propose a two-step algorithm to solve Problem~\eqref{eq:oob} and obtain the cellMG-GMM estimator. The W-step minimizes over $\bb{W}$ and the Expectation Minimization (Maximization) \citep[EM, ][]{Dempster1977, McLachlan2008} step minimizes over $(\bb{\pi}, \bb{\mu}, \bb{\Sigma})$. While our algorithmic implementation is, overall, similar to \citet{Raymaekers2023}, the EM-step requires careful adaptation to the 
MG-GMM model set-up. Given initial starting values (see Appendix~\ref{subsec:appendix_initialvalues}), we iteratively repeat the W-
and the EM-step until convergence.

\subsection{W-Step}
We update the matrix $\bb{W}$ in the $(\tau +1)$-th step while keeping the mixture parameters at their current values, namely $\hat{\bb{\pi}}^\tau = (\hat{\pi}_{g, k}^\tau )_{g,k = 1}^N$, $\hat{\bb{\mu}}^\tau = (\hat{\bb{\mu}}_k^\tau)_{k = 1}^N,$ and $  \hat{\bb{\Sigma}}^\tau = ( \hat{\bb{\Sigma}}_k^\tau )_{k = 1}^N$. To minimize the objective function Equation~\eqref{eq:oob} with respect to $\bb{W}$, denote the new pattern by $\tilde{\bb{W}}$ which we initialize at $\tilde{\bb{W}} = \hat{\bb{W}}^\tau$. We now modify $\tilde{\bb{W}}$ variable by variable. For a given variable $j$, we aim to obtain a new missingness pattern for the $j$th variable across all groups $g$ and observations $i$. To this end, we calculate the difference in the objective 
\begin{align*}
    \Delta_{g, ij} =& -2  \ln \left(\sum_{k = 1}^N \hat{\pi}_{g, k}^\tau \varphi\left(\x_{g,i}^{(\leftindex_1 {\tilde{\w}}_{g,i})}; {\bb{\hat{\mu}}_k^\tau}^{(\leftindex_1 {\tilde{\w}}_{g,i})}, {\bb{\hat{\Sigma}}_{reg,k}^{\tau}}^{(\leftindex_1 {\tilde{\w}}_{g,i})}\right)\right)  \\
    &+ 2\ln \left(\sum_{k = 1}^N \hat{\pi}_{g, k}^\tau \varphi\left(\x_{g,i}^{(\leftindex_0 {\tilde{\w}}_{g,i})}; {\bb{\hat{\mu}}_k^\tau}^{(\leftindex_0 {\tilde{\w}}_{g,i})}, {\bb{\hat{\Sigma}}_{reg,k}^\tau}^{(\leftindex_0 {\tilde{\w}}_{g,i})}\right)\right)  - q_{g, ij},
\end{align*}
between $\tilde{w}_{g,ij} = 1$, when including the cell in the estimation (denoted as $\leftindex_1 {\tilde{\w}}_{g,i}$), and   $\tilde{w}_{g,ij} = 0$, when flagging it as outlying (denotes as $\leftindex_0 {\tilde{\w}}_{g,i}$). 
If $\Delta_{g, ij} \leq  0$ for $h_g$ or more observations, the minimum is attained by setting the corresponding $\tilde{w}_{g,ij} $ to 1 and the others to 0; otherwise
the minimum is attained by setting $\tilde{w}_{g,ij}$ to 1 for those $h_g$ observations with the smallest $ \Delta_{g, ij}$ and the others to 0.
We update 
$\bb{W}$ by starting with variable $j=1$ and then consecutively cycling through the remaining variables, finally resulting in the updated $\hat{\bb{W}}^{\tau+1} = \tilde{\bb{W}}$.

\subsection{EM-Step}
Given the new missingness pattern $\hat{\bb{W}}^{\tau+1}$, we minimize Objective~\eqref{eq:oob} for incomplete data, hence we carry out an EM-step to update the parameters of the mixture model. To this end, we extend the EM-based algorithm for standard GMMs of
\citet{Eirola2014} to the multi-group GMM setting, thereby incorporating the additional Constraints~\eqref{eq:ob_piinequality} and~\eqref{eq:ob_covreg}. More details and derivations are provided in Appendix~\ref{subsec:appendix_EM}.

\subsection{Convergence of the Algorithm}
Pseudo-code for the algorithm is compactly presented in 
Algorithm~\ref{alg:main} of Appendix~\ref{subsec:app_pseudo}.
The algorithm iterates between the W-step and EM-step until the maximal absolute change in any entry of the covariance matrices, $\max_{k,j,j'} |{{\hat{\Sigma}}_{reg,k, jj'}^{\tau}} - {{\hat{\Sigma}}_{reg,k, jj'}^{\tau+1}}|$, is smaller than $\epsilon_{conv} = 10^{-4}$.

Since the regularization of the covariance matrices acts on the maximization step of the EM-algorithm, the same argumentation as in Proposition 6 from \citet{Raymaekers2023} can be applied to show that each W-step and EM-step reduce the objective function or leave it unchanged while fulfilling all constraints.
The algorithm thus converges; we verified that convergence was achieved in all simulations and applications.

\subsection{Choice of Hyperparameters}
\label{subsec:algorithm_parameters}
Objective function~(\ref{eq:oob}) depends on the hyperparameters $q_{g,ij}$, $\rho_k$, and $\bb T_k$. 

The penalty weights $q_{g,ij}$ 
need to be set for each group $g$, observation $i$ and variable $j$. To this end, we extend the choice of the penalty weights considered by \cite{Raymaekers2023} for cellwise-robust estimation of the MCD to the multi-group GMM setting.
Given initial estimates $\hat{\bb{\pi}}^0, \hat{\bb{\mu}}^0$, $\hat{\bb{\Sigma}}^0$ and $\hat{\bb{W}}^0$, we calculate the probabilities $\hat{t}_{g,i,k}^{0} $ 
(see Appendix A.2 for details)
and use a weighted penalty parameter for each observation, namely
   $q_{g, ij} = \chi_{1,0.99}^2 + \ln(2\pi) + \sum_{k = 1}^N \hat{t}_{g,i,k}^{0} \ln(C_{k,j}^0)$,
where $\chi_{1,0.99}^2$ denotes the $99$-th quantile of the chi-square distribution with one degree of freedom and $C_{k,j}^0 = 1/({\bb{\hat{\Sigma}}_{reg, k}^{0}})^{-1}_{jj}$. 

Regarding Constraint~\eqref{eq:ob_covreg}, we choose a diagonal matrix $\bb T_k$ consisting of robust univariate scale estimates 
for observations from group $k$, $\bb T_k = \operatorname{diag}(\hat{\sigma}_{k,1}, \ldots, \hat{\sigma}_{k,p})$. 
We use the univariate MCD estimator applied to each variable separately. 
To specify $\rho_k$, which regulates
the amount of regularization, 
we set it as small as possible and such that the condition  number fulfills $\rho_k \bb T_k + (1-\rho_k){\bb{\hat{\Sigma}}_{k}^{0}} \leq \kappa_k$ for an initial estimate ${\bb{\hat{\Sigma}}_{k}^{0}}$ and $\kappa_k = \max (1.1\operatorname{cond} \bb T_k, 100)$. We hereby opt for a condition number of $100$ for each covariance, but the factor $1.1$ allows for multivariate data input if the condition number of $\bb T_k$ is high.

\section{Theoretical Properties}
\label{sec:theory}
The study of theoretical properties such as the breakdown point (BP) in cluster and finite mixture model settings is complicated since the addition of a single outlying point can make the parameter estimation of at least one of the mixture components break down \citep{Hennig2004}.
It is therefore common to analyze the BP under a general assumption of well-clustered data in an idealized setting, as introduced in \cite{Hennig2004} for the rowwise contamination paradigm. In Section \ref{subsec:bdp_mixture}, we first extend this idealized setting to the cellwise contamination paradigm, which is of general interest in cluster and finite mixture settings.
In Section \ref{subsec:bdp_grouped}, we then specifically derive the BP of the cellMG-GMM estimator of the multi-group GMM model.

\subsection{Cellwise Breakdown in an Idealized Setting}
\label{subsec:bdp_mixture}

We consider the cellwise outlier paradigm \citep{Alqallaf2009} where data are assumed to be initially generated from a certain distributional model, after which some individual cells are contaminated.
To study cellwise outlyingness in mixture model settings, the idealized setting of well-clustered data in \cite{Hennig2004}, developed for the rowwise outlier paradigm, does not sufficiently separate the clusters under the cellwise outlier paradigm. Indeed, under cellwise contamination, 
the removal of a subset of variables could still lead to cluster overlap, see Figure~\ref{fig:bdp_a} for an intuitive illustration; the notation used in the figure is formalized below. The idealized setting should thus be adapted to cluster separation across all subsets, see Figure~\ref{fig:bdp_b}, which is equivalent to a separation in each variable.

\begin{figure}[t]
    \centering
        \begin{subfigure}[b]{0.46\textwidth}
            \begin{tikzpicture}
            
            \draw[->, thick] (-3,0) -- (2.5,0) node[anchor=west] {};
            \draw[->, thick] (0,-2.5) -- (0,2.5) node[anchor=south] {};
            
            \draw[thick, red] (-0.9,1) ellipse (0.8 and 0.5);
            
            \draw[thick, blue] (1,1) ellipse (0.7 and 0.95);
    
            \node at (-1,1) {\(A_m^1\)};
            \node at (1,1) {\(A_m^2\)};
            
            \draw[->, thick] (-2.1,1) -- (-3.1,1) node[anchor=south east] {};
            \node[rotate=0, anchor = south, scale=0.7] at (-2.5,1) {\(m \rightarrow \infty\)};
            
            \draw[->, thick] (1.85,1) -- (2.7,1);
            \node[rotate=0, anchor = south, scale=0.7] at (2.2,1) {\(m \rightarrow \infty\)};
            
            \filldraw [black] (0.5,-1) circle (2pt);
            \node at (0.5,-1) [below right] {$ \bb{y}_{2,m} 
            $}; 
            
            \filldraw [black] (-1.5,-0.75) circle (2pt);
            \node at (-1.5,-0.75) [above right] {$ \bb{y}_{1,m}$};
            
            \draw[->, thick] (-1.5,-1) -- (-1.5,-2) node[anchor=south east] {};
            \node[rotate = 90, anchor = south, scale=0.7] at (-1.5,-1.4) {\(m \rightarrow \infty\)};
            
            \end{tikzpicture}
        \caption{Non-ideal under cellwise paradigm: Clusters $A_m^1$, $A_m^2$ and $y_{2, m}$ are not separated along the second axis. The points $\bb{y}_{1, m}$ and $\bb{y}_{2, m}$ are outlying, but not separated along the first axis. Outlier $\bb{y}_{1, m}$ is infinitely far away from the clusters, but outlier $\bb{y}_{2, m}$ remains steady for $m \rightarrow \infty$.}
        \label{fig:bdp_a}
        \end{subfigure}
        \hfill
        \begin{subfigure}[b]{0.46\textwidth}
            \begin{tikzpicture}
            
            \draw[->, thick] (-3,0) -- (2.5,0) node[anchor=west] {};
            \draw[->, thick] (0,-2.5) -- (0,2.5) node[anchor=south] {};
            
            \draw[thick, red] (-1.5,-1) ellipse (1 and 0.5);
            
            \draw[thick, blue] (1,1.5) ellipse (0.7 and 1);
            
            \draw[->, thick] (-2.1,-1.6) -- (-2.8,-2.3) node[anchor=south east] {};
            \node[rotate=45, anchor = south, scale=0.7] at (-2.4,-1.9) {\(m \rightarrow \infty\)};
            
            \draw[->, thick] (1.7,2.2) -- (2.5,3);
            \node[rotate=45, anchor = south, scale=0.7] at (2.1,2.6) {\(m \rightarrow \infty\)};
            
            \node at (-1.5,-1) {\(A_m^1\)};
            \node at (1,1.5) {\(A_m^2\)};
            
            \filldraw [black] (0.5,-1) circle (2pt);
            \node at (0.5,-1) [below right] {$ \bb{y}_{2,m} 
            $}; 
            
            \filldraw [black] (-1.5,1) circle (2pt);
            \node at (-1.5,1) [below right] {$ \bb{y}_{1,m}$};
            
            \draw[->, thick] (-1.7,0.95) -- (-3, 0.46) node[anchor=south east] {};
            \node[rotate = 20, anchor = south, scale=0.7] at (-2.2,0.8) {\(m \rightarrow \infty\)};
    
            \draw[dashed] (-1.5,1) -- (-1.5,-0.5);
            
            \end{tikzpicture}
            \caption{Ideal under cellwise paradigm:
            Clusters $A_m^1$, $A_m^2$ are well-separated. Point $\bb{y}_{1,m} \in B_m^1$ is only outlying 
            along the second axis
            (i.e.\ $\w(\bb{y}_{1, m})= (1, 0)$) and its non-outlying part originates from the 1st cluster (indicated by the dashed line). Point $\bb{y}_{2,m}$ is steady and outlying in both directions (i.e.\ $\w(\bb{y}_{2, m}) = \bb{0}$).}    
            \label{fig:bdp_b}
        \end{subfigure}
    \caption{
    Non-ideal setting with overlapping clusters in panel (a) versus ideal setting with well-separated clusters under the
    cellwise outlier paradigm in panel (b). Arrows indicate the direction of each cluster or outlier sequence.}
    \label{fig:bdp}
\end{figure}

Formally and following the ideas of \cite{Hennig2004},
let $s \geq 2$ be the number of clusters, and $\n_1 < \n_2 < \ldots < \n_s = \n \in \mathbb{N}$. 
Consider a sequence of clusters $\left(\mathcal{X}_m\right)_{m \in \mathbb{N}}$. For each
$m$-th part of the sequence, the data~$\mathcal{X}_m$ are clustered into $s$ clusters $A_m^1 =  \{ \bb{x}_{1, m}, \ldots, \bb{x}_{\n_1, m} \}, \ldots, A_m^s =  \{ \bb{x}_{\n_{s-1} +1, m}, \ldots, \bb{x}_{\n_s, m} \}$, with 
 $\mathcal{X}_m = \bigcup_{l = 1}^s A_m^l$ and $\bb{x}_{i, m} = (x_{i1, m}, \ldots, x_{ip, m})$ for $i = 1, \ldots, \n$.
The sequence of well-separated clusters $\left(\mathcal{X}_m\right)_{m \in \mathbb{N}}$ is considered 
ideal when the distances between observations of the same cluster are bounded by a constant $b < \infty$,
\begin{align}
    \max_{1 \leq l \leq s} \max \{ |{x}_{i'j, m} - {x}_{ij, m} |: \bb{x}_{i', m}, \bb{x}_{i, m} \in A_m^l, j = 1, \ldots, p\} < b \quad \forall m \in \mathbb{N}, \label{eq:withincluster_cell}
\end{align}
and observations from different clusters are increasingly far away, thereby enforcing
\begin{align}
    \lim_{m \rightarrow \infty} \min \{ |x_{i'j, m} - x_{ij, m}|:  \bb{x}_{i', m} \in A_m^l, \bb{x}_{i, m} \in A_m^h, h \neq l, j = 1, \ldots, p \} = \infty. \label{eq:betweencluster_cell}
\end{align}
We now
add cellwise outliers $\mathcal{Y}_m =\{ \bb{y}_{1, m}, \dots,\bb{y}_{\r, m} \}$, such that $0 \leq \r_1 \leq \ldots \leq \r_s = \r $ and 
  $  B_m^1 = \{ \bb{y}_{1, m}, \ldots, \bb{y}_{\r_1, m} \}, \ldots, B_m^s = \{ \bb{y}_{\r_{s-1} +1, m}, \ldots, \bb{y}_{\r_s, m} \}.$ 
For each 
observation $\bb{y}_{i, m}$, there exists a $\w(\bb{y}_{i, m}) \in \{0, 1\}^p$ indicating 
outlying cells by $w(\bb{y}_{i, m})_{ j} = 0$ and non-outlying cells by $w(\bb{y}_{i, m})_{ j} = 1$. The non-outlying part 
should originate from one of the constructed clusters,
\begin{align*}
    \max_{1 \leq l \leq s} \max \{ |y_{i'j, m} - x_{ij, m} |:& \bb{x}_{i, m} \in A_m^l, \bb{y}_{i', m} \in B_m^l, \\& j=1, \ldots, p \text{ with }w(\bb{y}_{i', m})_{ j} = 1
    \} < b \quad \forall m \in \mathbb{N}, 
\end{align*}
and the outlying part
should be infinitely far away from all other outlying cells and clusters,
\begin{align}
    \lim_{m \rightarrow \infty} \min \{ |y_{i'j, m} - x_{ij, m}|:  \bb{x}_{i, m} \in \mathcal{X}_m, \bb{y}_{i', m} \in \mathcal{Y}_m, w(\bb{y}_{i', m})_{ j} = 0\} = \infty, \label{eq:outlier_clusters_cell} \\
    \lim_{m \rightarrow \infty} \min \{ |y_{i'j, m} - y_{ij, m}|:  \bb{y}_{i', m}, \bb{y}_{i, m} \in \mathcal{Y}_m, i \neq i',  w(\bb{y}_{i', m})_{ j} = 0\} = \infty. \label{eq:outlier_outlier_cell} 
\end{align}
The breakdown of an estimator $\hat{E}$ can then be defined in a relative fashion, thereby relating its behavior acting over $\mathcal{X}_m$  and over $\mathcal{X}_m \cup \mathcal{Y}_m$ for large values of $m$.
Location breakdown for a cluster $l$ occurs, if 
    $|| \bb{\hat{\mu}}_l(\mathcal{X}_m ) - \bb{\hat{\mu}}_{k}(\mathcal{X}_m \cup \mathcal{Y}_m) ||_2 \rightarrow \infty$ for all $k = 1, \ldots, N$,
where $|| \cdot ||_2$ denotes the Euclidean norm.
A covariance estimator of a cluster $l$ would implode (explode) if $\lambda_p(\bb{\hat{\Sigma}}_l(\mathcal{X}_m)) \rightarrow 0$ ($\lambda_1(\bb{\hat{\Sigma}}_l(\mathcal{X}_m)) \rightarrow \infty$) and $\lambda_p(\bb{\hat{\Sigma}}_l(\mathcal{X}_m \cup \mathcal{Y}_m)) \nrightarrow 0$ ($\lambda_1(\bb{\hat{\Sigma}}_l(\mathcal{X}_m \cup \mathcal{Y}_m)) \nrightarrow \infty$) or vice versa, where $\lambda_1$ and $\lambda_p$ denote the largest and  smallest eigenvalue, respectively.
The weight estimator $\hat{\pi}_l$ of a cluster $l$ breaks down if $\hat{\pi}_l \in \{0, 1\}$, i.e., whenever at least one cluster is empty. 
Finally, the cellwise additive BP is 
defined as
\begin{align*}
    \epsilon^* (\hat{E}) = \min \left\{\frac{ \max_{j = 1, \ldots, p} \sum_{i = 1}^{\r} (1-w(\bb{y}_{i, m})_{ j}) }{\n +\r }: \hat{E}\text{ breaks down} \right\},
\end{align*}
where $\sum_{i = 1}^{\r} (1-w(\bb{y}_{i, m})_{ j})$ denotes the number of contaminated cells for variable~$j$.

\subsection{Cellwise Breakdown of cellMG-GMM}
\label{subsec:bdp_grouped}

To obtain
the BP of the cellMG-GMM estimator of the multi-group GMM,
we assume $N$ 
well-separated underlying clusters and outliers constructed
as described in Section~\ref{subsec:bdp_mixture}.
All observations $\mathcal{X}_m  \cup \mathcal{Y}_m$, clean or contaminated, are partitioned into groups $\Z^1_m, \ldots, \Z^N_m$ of size $n_1 + r_1, \ldots, n_N + r_N$ (where $n_g$ is the number of clean and $r_g$ 
the number of added, contaminated observations of group $g$) by a function $\tilde{g}: \mathcal{X}_m \bigcup \mathcal{Y}_m \rightarrow \{1, \ldots, N \}$, thus $\mathcal{Z}_m = \bigcup_{g = 1}^N \Z^g_m = \mathcal{X}_m \bigcup \mathcal{Y}_m$. 
We assume that for each group $g$ a certain fraction $\tilde{\alpha}_g$ 
of its $n_g$ observations and $r_g$ added outliers are from cluster~$g$, 
\begin{align}
    \frac{|\{\x: \x \in A^g_m, \tilde{g}(\x) = g\}|}{n_g} \geq \tilde{\alpha}_g, \quad \frac{|\{\bb{y}: \bb{y} \in B^g_m, \tilde{g}(\bb{y}) = g\}|}{r_g} \geq \tilde{\alpha}_g, \nonumber
\end{align}
thus, reflecting the major distribution per group. An illustration 
for a fictitious ideal data set is shown in Figure~\ref{fig:illus_groupeddata}. Each column block corresponds to a group, each column within a block to a variable and each row  to an observation.  The first row block per group includes the clean observations, the second block the added and possibly contaminated observations. The cell color 
indicates either clean cells belonging to the ideal group (red, violet, green)  the observation originates from, or outlying cells in gray. For each group, the majority of both clean and contaminated observations comes from the main cluster. 
Cellwise contamination can affect single cells (e.g.\ group 2), all cells of certain variables (e.g.\  group 1, fully gray column for variables 2 and 4) and/or whole observations (e.g.\  group 3, fully contaminated first observation/row). Note that the latter observation is assigned to $B^3_m$, but it could stem from any other group too.

\begin{figure}[t]
    \centering
        \begin{tikzpicture}
        \node at (0, 3.2) {$Z^1_m $};
        \matrix[nodes={draw, minimum size=0.5cm, anchor=center},column sep=0mm, row sep=0mm] (m) {
            \node[fill=group1] {}; & \node[fill=group1] {}; & \node[fill=group1] {}; & \node[fill=group1] {}; & \node[fill=group1] {};\\
            \node[fill=group1] {}; & \node[fill=group1] {}; & \node[fill=group1] {}; & \node[fill=group1] {}; & \node[fill=group1] {};\\
            \node[fill=group1] {}; & \node[fill=group1] {}; & \node[fill=group1] {}; & \node[fill=group1] {}; & \node[fill=group1] {};\\
            \node[fill=group1] {}; & \node[fill=group1] {}; & \node[fill=group1] {}; & \node[fill=group1] {}; & \node[fill=group1] {};\\
            \node[fill=group1] {}; & \node[fill=group1] {}; & \node[fill=group1] {}; & \node[fill=group1] {}; & \node[fill=group1] {};\\
            \node[fill=group2] {}; & \node[fill=group2] {}; & \node[fill=group2] {}; & \node[fill=group2] {}; & \node[fill=group2] {};\\
            \node[fill=group3] {}; & \node[fill=group3] {}; & \node[fill=group3] {}; & \node[fill=group3] {}; & \node[fill=group3] {};\\
            \node[fill=group3] {}; & \node[fill=group3] {}; & \node[fill=group3] {}; & \node[fill =group3] {}; & \node[fill=group3] {};\\
            \\[2mm] 
            \node[fill=group1] {}; & \node[fill=outs] {}; & \node[fill=group1] {}; & \node[fill=outs] {}; & \node[fill=outs] {};\\
            \node[fill=group1] {}; & \node[fill=outs] {}; & \node[fill=outs] {}; & \node[fill=outs] {}; & \node[fill=group1] {};\\
            \node[fill=outs] {}; & \node[fill=outs] {}; & \node[fill=outs] {}; & \node[fill=outs] {}; & \node[fill=group2] {};\\
        };
        
        \draw[decorate,decoration={brace,mirror,amplitude=2mm,raise=1mm}]  (-1.3,2.9) -- (-1.3,0.4);
        \node at (-1.6, 1.65) [anchor=east] {$A^1_m $};
        
        \draw[decorate,decoration={brace,mirror,amplitude=2mm,raise=1mm}]  (-1.3,0.3) -- (-1.3,-0.1);
        \node at (-1.6, 0.1) [anchor=east] {$A^2_m $};
        
        \draw[decorate,decoration={brace,mirror,amplitude=2mm,raise=1mm}]  (-1.3,-0.2) -- (-1.3,-1.2);
        \node at (-1.6, -0.7) [anchor=east] {$A^3_m $};
        
        \draw[decorate,decoration={brace,mirror,amplitude=2mm,raise=1mm}]  (-1.3,-1.4) -- (-1.3,-2.4);
        \node at (-1.6, -1.9) [anchor=east] {$B^1_m $};
        
        \draw[decorate,decoration={brace,mirror,amplitude=2mm,raise=1mm}]  (-1.3,-2.5) -- (-1.3,-2.9);
        \node at (-1.6, -2.7) [anchor = east] {$B^2_m $};
        
        \end{tikzpicture}
        \hfill
        \begin{tikzpicture}
        \node at (0, 3.2) {$Z^2_m $};
        \matrix[nodes={draw, minimum size=0.5cm, anchor=center},column sep=0mm, row sep=0mm] (m) {
            \node[fill=group2] {}; & \node[fill=group2] {}; & \node[fill=group2] {}; & \node[fill=group2] {}; & \node[fill=group2] {};\\
            \node[fill=group2] {}; & \node[fill=group2] {}; & \node[fill=group2] {}; & \node[fill=group2] {}; & \node[fill=group2] {};\\
            \node[fill=group2] {}; & \node[fill=group2] {}; & \node[fill=group2] {}; & \node[fill=group2] {}; & \node[fill=group2] {};\\
            \node[fill=group2] {}; & \node[fill=group2] {}; & \node[fill=group2] {}; & \node[fill=group2] {}; & \node[fill=group2] {};\\
            \\[2mm] 
            \node[fill=outs] {}; & \node[fill=group2] {}; & \node[fill=group2] {}; & \node[fill=outs] {}; & \node[fill=group2] {};\\
            \node[fill=group3] {}; & \node[fill=outs] {}; & \node[fill=outs] {}; & \node[fill=group3] {}; & \node[fill=group3] {};\\
            \\[25mm];\\
        };

        \draw[decorate,decoration={brace,mirror,amplitude=2mm,raise=1mm}]  (-1.3,2.9) -- (-1.3,0.9);
        \node at (-1.6, 1.9) [anchor=east] {$A^2_m $};
        
        \draw[decorate,decoration={brace,mirror,amplitude=2mm,raise=1mm}]  (-1.3,0.6) -- (-1.3,0.2);
        \node at (-1.6, 0.4) [anchor=east] {$B^2_m $};
        
        \draw[decorate,decoration={brace,mirror,amplitude=2mm,raise=1mm}]  (-1.3,0.1) -- (-1.3,-0.35);
        \node at (-1.6, -0.2) [anchor=east] {$B^3_m $};       
        \end{tikzpicture}
        \hfill
        \begin{tikzpicture}
        \node at (0, 3.2) {$Z^3_m $};
        \matrix[nodes={draw, minimum size=0.5cm, anchor=center},column sep=0mm, row sep=0mm] (m) {
           \node[fill=group3] {}; & \node[fill=group3] {}; & \node[fill=group3] {}; & \node[fill=group3] {}; & \node[fill=group3] {};\\
            \node[fill=group3] {}; & \node[fill=group3] {}; & \node[fill=group3] {}; & \node[fill=group3] {}; & \node[fill=group3] {};\\
            \node[fill=group1] {}; & \node[fill=group1] {}; & \node[fill=group1] {}; & \node[fill=group1] {}; & \node[fill=group1] {};\\
            \\[2mm] 
            \node[fill=outs] {}; & \node[fill=outs] {}; & \node[fill=outs] {}; & \node[fill=outs] {}; & \node[fill=outs] {};\\
            \node[fill=group3] {}; & \node[fill=group3] {}; & \node[fill=group3] {}; & \node[fill =group3] {}; & \node[fill=group3] {};\\
            \node[fill=group3] {}; & \node[fill=outs] {}; & \node[fill=outs] {}; & \node[fill=group3] {}; & \node[fill=outs] {};\\
            \node[fill=group1] {}; & \node[fill=outs] {}; & \node[fill=outs] {}; & \node[fill=group1] {}; & \node[fill=group1] {};\\
            \node[fill=outs] {}; & \node[fill=group2] {}; & \node[fill=group2] {}; & \node[fill=outs] {}; & \node[fill=group2] {};\\
            \\[15mm];\\
        };

        \draw[decorate,decoration={brace,mirror,amplitude=2mm,raise=1mm}]  (-1.3,2.9) -- (-1.3,1.95);
        \node at (-1.6, 2.4) [anchor=east] {$A^3_m $};
        
        \draw[decorate,decoration={brace,mirror,amplitude=2mm,raise=1mm}]  (-1.3,1.85) -- (-1.3,1.4);
        \node at (-1.6, 1.65) [anchor=east] {$A^1_m $};
        
        \draw[decorate,decoration={brace,mirror,amplitude=2mm,raise=1mm}]  (-1.3,1.1) -- (-1.3,-0.3);
        \node at (-1.6, 0.5) [anchor=east] {$B^3_m $};
        
        \draw[decorate,decoration={brace,mirror,amplitude=2mm,raise=1mm}]  (-1.3,-0.4) -- (-1.3,-0.8);
        \node at (-1.6, -0.6) [anchor=east] {$B^1_m $};
        
        \draw[decorate,decoration={brace,mirror,amplitude=2mm,raise=1mm}]  (-1.3,-0.95) -- (-1.3,-1.35);
        \node at (-1.6, -1.2) [anchor = east] {$B^2_m $};
        \end{tikzpicture}
    \caption{Fictitious ideal data set with $N = 3$ groups (column blocks),  $p=5$ (variables in columns per block), and respectively 8, 4, 3 clean observations and  3, 2, 5 added and possibly contaminated observations in the rows, across groups 1-3. Cell colors (red-violet-green) indicate from which group each observation originates, or outlyingness (gray).}
    \label{fig:illus_groupeddata}
\end{figure}

For the ideal scenario, we assume that at least $\left\lceil \frac{n_g + r_g +1}{2} \right\rceil$ observations from group $g$ are from cluster $g$ and thus, $\tilde{\alpha}_g$ is restricted to 
$(n_g + r_g) \tilde{\alpha}_g \geq \left\lceil \frac{n_g +r_g+1}{2} \right\rceil $ for all $g = 1, \ldots, N$.
In terms of estimation, this implies that for any variable $j$ and group $g$ there always exists at least one observation in $\Z^g_m$ originating from cluster $g$ which is observed for variable~$j$.

Cellwise BP of the cellMG-GMM estimator of the multi-group GMM 
is defined as the minimal fraction of outlying cells for at least one variable in at least one group 
needed to lead to breakdown of one 
estimator~$\hat{E}$, 
\begin{align*}
    \epsilon_{MG-GMM}^* (\hat{E}) = \min_{g= 1, \ldots, N} \min \left\{
    \frac{ \max_{j = 1, \ldots, p} \sum_{\bb{y} \in \Z^g_m \cap \mathcal{Y}_m} (1-w(\bb{y})_{ j}) }
    {n_g + r_g}:
    \hat{E}\text{ breaks down} \right\}.
\end{align*}
Theorem \ref{theorem:breakdownpoint} presents the breakdown point results; all proofs are in Appendix~\ref{sec:appendix_theory}.

\begin{theorem}
\label{theorem:breakdownpoint}
Given the idealized setting (Section \ref{subsec:bdp_mixture} and extensions thereof in Section~\ref{subsec:bdp_grouped}) and fixed $\rho_k  > 0, \bb{T}_k \succ 0$, the following breakdown results hold under the cellwise contamination paradigm:
    \begin{itemize}
        \item[a.] Assuming that $h_g \geq \lceil 0.75 (n_g + 1)\rceil$ for all $g = 1, \ldots, N$, the location BP is at least $\min_g \{(n_g-h_g+1)/n_g\}$.
        \item[b.] For the covariance estimator, the implosion BP is $1$.
        \item[c.] For the covariance estimator, the explosion BP is at least $\min_g \{(n_g-h_g+1)/n_g\}$.
        \item[d.] For the covariance estimator, the explosion BP is exactly $\min_g \{(n_g-h_g+1)/n_g\}$, when the location estimator did not break down. 
        \item[e.] The weight BP is $1$.
    \end{itemize}
\end{theorem}

Theorem \ref{theorem:breakdownpoint} quantifies theoretical robustness guarantees of the location, covariance and weight estimators against a certain percentage of adversarial contamination. While the covariance estimator is robust against $(n_g-h_g+1)/n_g$ outliers per group $g$ for $h_g$ up to $0.5n_g$, in special cases the location estimator could break down immediately, if the additional restriction on $h_g$ is not fulfilled.

\section{Simulations}
\label{sec:simulations}

We assess the performance of cellMG-GMM in
five main scenarios: 
1) $N = 2$ balanced groups (our basic scenario), 
2) $N = 5$ balanced groups, 
3) $N = 2$ unbalanced groups,
4) $N = 2$ balanced groups  with increasing singularity issues, and 
5) high-dimensional $N = 2$ balanced groups. 
Scenarios 1) and 2) are described in detail in the main text, results for the remaining scenarios are available in Appendix~\ref{app:si:results} and summarized at the end of the results section.

In Section~\ref{subsec:simulations_datageneration}, we detail the generation of clean and contaminated data. 
Benchmark methods and evaluation criteria are summarized in Section~\ref{subsec:simulations_compete} and~\ref{subsec:simulations_eval} respectively. The results of the simulation study are discussed in Section~\ref{subsec:simulations_results}.

\subsection{Data Generation}
\label{subsec:simulations_datageneration}

{\it Clean data.} Data are generated according to the 
multi-group GMM 
in Equation~\eqref{eq:model}, for dimensions $p \in \{10, 20, 60\}$. 
For $N \in\{2, 5\}$ groups, we vary the mixture between the groups indicated by the parameter $\pi_{diag} \in\{0.75, 0.9\}$. The mixture probabilities are then given by $\pi_{gg} = \pi_{diag}$ and $\pi_{g,k} = \frac{1-\pi_{diag}}{N-1}$ for $g, k = 1, \ldots, N, g \neq k$. 
Each group $g$ consists of $n_g \in \{30, 40, 50, 100\}$ 
clean observations drawn with probability $\pi_{g,k}$ from $\mathcal{N}(\bb\mu_k, \bb\Sigma_k)$.

The covariance matrices of the mixture distribution are constructed
based on the approach of \cite{Agostinelli2015} (ALYZ) to obtain well-conditioned correlation matrices. We allow for more variation of the variances and stop the iterative procedure early, specifically when the trace of a covariance is bounded by $[p/2, 2p]$. 

We consider two different mean structures. First, we take
$\bb \mu_k = \bb 0$. 
Secondly, we consider a
more realistic scenario with different means, 
thereby applying the concept of c-separation \citep{Dasgupta1999} that gives a notion of how strongly the distributions overlap. We assume significant overlap ($0.5$-separated clusters) due to an underlying smooth variable and construct the means inductively, starting with $\bb\mu_1 = \bb 0_p$. Given $\bb\mu_1, \ldots, \bb\mu_{k-1}$ a new vector $\bb\mu_{tmp}$ is drawn from $\mathcal{N}(\bb 0_p, \bb I_p)$. To ensure a certain level of separation and overlap, we set the next distributional mean to $\bb\mu_k = t^* (\bb{\mu}_{tmp} - \frac{1}{k-1}\sum_{l = 1}^{k-1} \bb{\mu}_l) + \frac{1}{k-1}\sum_{l = 1}^{k-1} \bb{\mu}_l$, where $t^*$ is the minimal positive value that fulfills 
    $||\bb\mu_l - \bb\mu_k||_2 \geq  0.5 \sqrt{p \max(\lambda_1(\bb \Sigma_l), \lambda_1(\bb \Sigma_k))}$
for all $l = 1, \ldots, k-1$, with equality for at least one $l$. 

{\it Contamination.} For each group, a percentage $\epsilon_{cell} = 10\%$ of random cells per variable is contaminated as in \citet{Raymaekers2023}. Given an observation from group $g$ which is drawn from distribution $k$ and where a subset of variables indexed with $\mathcal{J}$ should be contaminated, cells indexed by $\mathcal{J}$ are replaced with 
\begin{align*}
    \bb \mu_{k,\mathcal{J}} + \bb v_{k ,\mathcal{J}} \frac{\gamma_{cell} \sqrt{|\mathcal{J}|}}{\sqrt{\bb v_{k ,\mathcal{J}}' \bb\Sigma_{k ,\mathcal{J}}^{-1}\bb v_{k ,\mathcal{J}}}}.
\end{align*}
Here, the subscript $\mathcal{J}$ denotes the part of the vectors/matrices corresponding to the indexed variables, and $\bb v_{k ,\mathcal{J}}$ denotes the eigenvector with the smallest eigenvalue of $ \bb\Sigma_{k ,\mathcal{J}}$. The parameter $\gamma_{cell} \in \{2, 6, 10\}$ controls the strength of the outlyingness of contaminated cells with respect to $\bb \mu_k$.
For $\gamma_{cell}=2$ the cellwise outliers are hard to distinguish from regular cells, while $\gamma_{cell}=10$ produces clear outliers which are easier to detect for robust methods, and very influential to non-robust procedures.

\subsection{Benchmarks}
\label{subsec:simulations_compete}

We compare  cellMG-GMM to six benchmarks (see Appendix \ref{app:benchmarks} for more details): 
\begin{description}
    
    \item [\textbf{sample}:] Sample covariance and mean per group 
    as non-robust, supervised benchmark where we use the word ``supervised'' to reflect knowledge of the group membership. 
    
    \item [\textbf{MRCD}:]  Rowwise robust, supervised covariance/location estimator of \cite{Boudt2020}, implemented in the R-package \texttt{rrcov} \citep{rrcov_CRAN},  
    and applied per group.

    \item [\textbf{cellMCD}:] Cellwise robust, supervised covariance/location estimator of \cite{Raymaekers2023}, implemented in the R-package \texttt{cellWise} \citep{cellWise_CRAN}, and applied per group.
    
    \item [\textbf{OC}:] Cellwise robust, supervised covariance estimator of \cite{Ollerer2015}, implemented in the R-package \texttt{pcaPP}~\citep{pcapp_CRAN}, and applied per group.
    
    \item [\textbf{ssMRCD}:] Rowwise robust, semi-supervised covariance/location estimator of \cite{Puchhammer2024}, implemented in the R-package \texttt{ssMRCD} \citep{ssMRCD_Cran}. 

     \item [\textbf{mclust}:] Non-robust, unsupervised basic finite GMM implemented 
    in the R-package \texttt{mclust} \citep{mclust_CRAN} 
    with the correct number of groups provided. 

    \item [\textbf{cellGMM}:]  Cellwise robust, unsupervised basic finite GMM of \cite{zaccaria2024} with R-code and suggested hyper-parameter setting available in their supplement.
\end{description}

\subsection{Evaluation Criteria}
\label{subsec:simulations_eval}
Given an estimated covariance $\hat{\bb \Sigma}_k$ by a particular method, the Kullback-Leibler divergence to the real covariance $\bb \Sigma_k$ is used as evaluation criterion to assess estimation accuracy,
\begin{align*}
    KL(\hat{\bb \Sigma}_k, \bb \Sigma_k) = \operatorname{tr}(\hat{\bb \Sigma}_k \bb \Sigma_k^{-1}) - p - \log\det (\hat{\bb \Sigma}_k \bb \Sigma_k^{-1}).
\end{align*}
For $N \geq 2$, the final performance metric is the average over all distributions, $KL = \frac{1}{N} \sum_{k = 1}^N KL(\hat{\bb \Sigma}_k, \bb \Sigma_k)$. 
The mean estimates $\hat{\bb \mu}_k$ and mixture probabilities $\hat{\bb \pi}$ are evaluated by the Mean Squared Error (MSE)
\begin{equation*}
    MSE(\hat{\bb \mu}_k, \bb \mu_k) = \frac{1}{p} \sum_{j = 1}^p (\mu_{kj} - \hat{\mu}_{kj})^2,  \ \ \ 
     MSE(\hat{\bb \pi}, \bb \pi) = \frac{1}{N^2} \sum_{g = 1}^N \sum_{k = 1}^N(\pi_{g, k} - \hat{\pi}_{g, k})^2,
\end{equation*}
and averaged over the groups for the mean, $MSE(\mu) = \frac{1}{N} \sum_{k = 1}^N MSE(\hat{\bb \mu}_k, \bb \mu_k)$.

Additionally, we measure the correctness of flagged cellwise outliers by the
standard recall, precision and F1-score. 
For outlier flagging, we only compare the cellMG-GMM to the 
cellMCD and the cellGMM, since these are the only 
benchmarks that flags cells as outlying.

\subsection{Results}
\label{subsec:simulations_results}

We focus on Scenarios 1 and 2 for $p=10$ and $n_g=100$ in the text, results for the other settings of $p$ and $n_g$ are available in Appendix~\ref{app:si:results}.
We summarize the main similarities and differences in the results across the other scenarios at the end of this section. 
Each simulation setting is repeated 100 times. 

\begin{figure}[t]
    \centering
    \includegraphics[width=0.9\textwidth, trim = {0cm 
    1cm
    0cm 0.6cm}, clip]{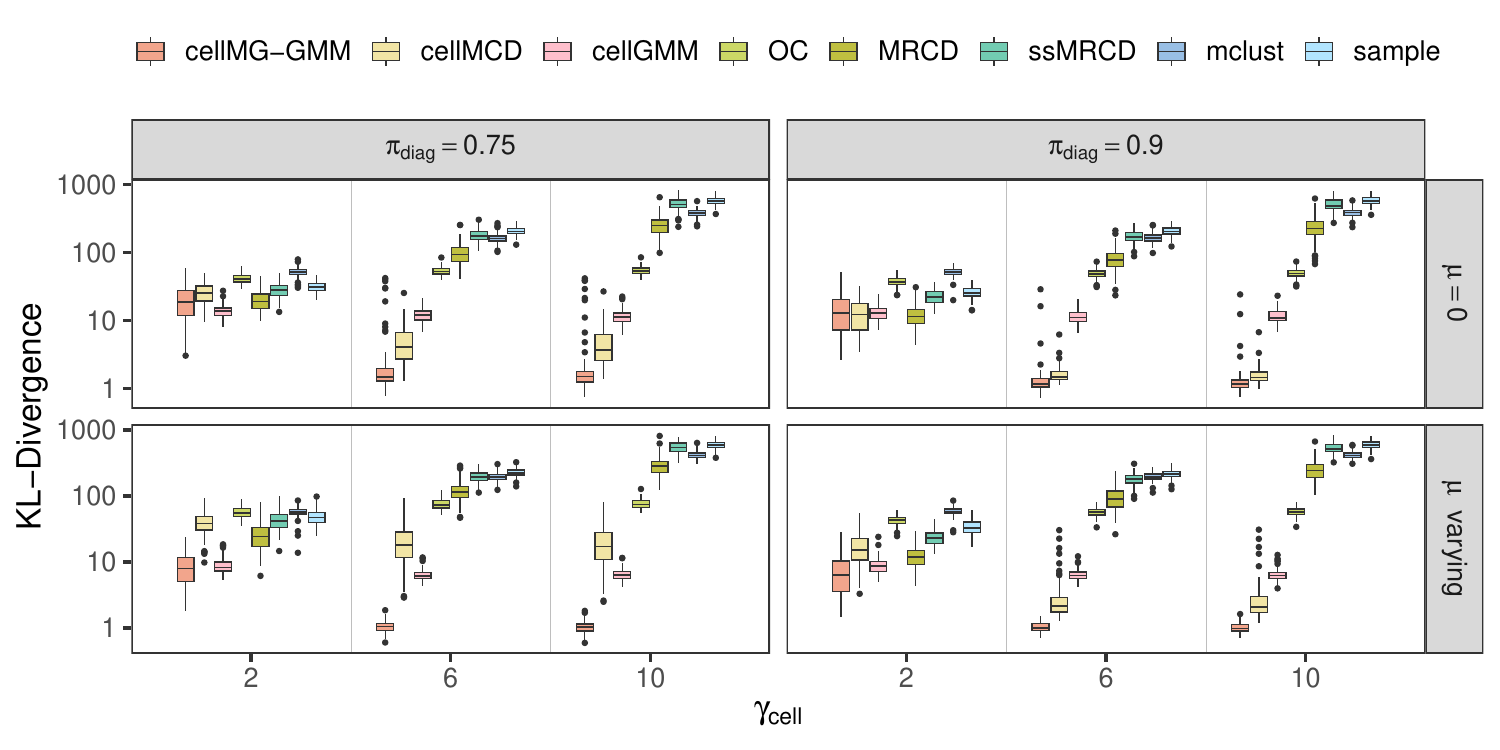}
    \includegraphics[width=0.9\textwidth, trim = {0cm 0cm 0cm 1.5cm}, clip]{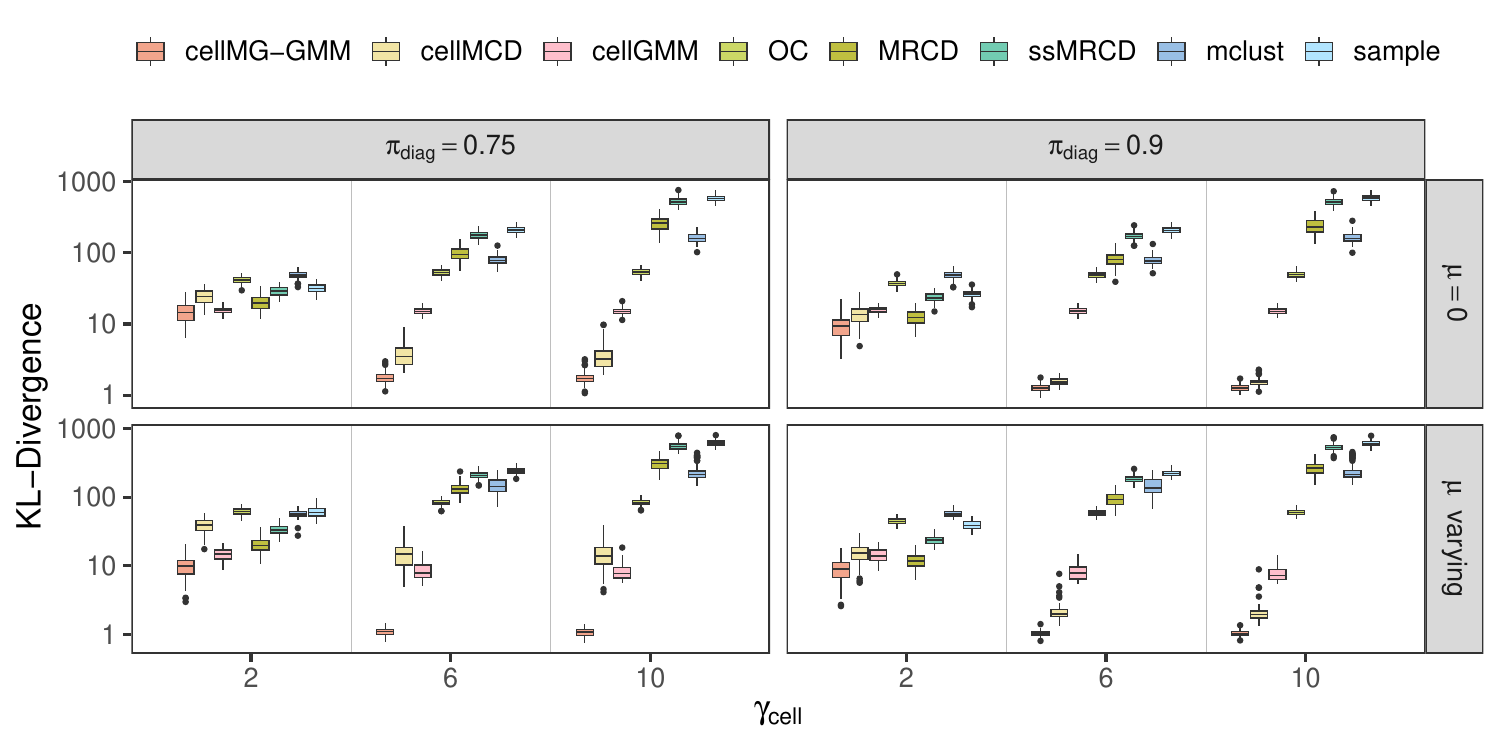}    
    \caption{KL-divergence for the basic balanced Scenario 1 with $N=2$ 
    (top) and Scenario 2 with $N=5$ (bottom), for varying strength of outlyingness $\gamma_{cell}$.}
    \label{fig:set1_A0X_sigma}
\end{figure}

We start with the basic balanced Scenario 1 with
$N = 2$ groups. 
Figure~\ref{fig:set1_A0X_sigma}, top panel,  
shows the KL-divergence for covariance estimation across all eight competing methods and a varying strength of outlyingness $\gamma_{cell}$.
Estimation accuracy results in terms of the group means are, qualitatively, similar and presented together with the results on the mixture probabilities of cellMG-GMM in Appendix~\ref{app:si:results}. The four subpanels differ regarding the coherency in the predefined groups. For example, observations of one group are very coherent for $\pi_{diag} = 0.9$ and $\mu = 0$ (top right panel) or less coherent for $\pi_{diag} = 0.75$ and varying $\mu$. 
Across all four coherency types, 
only the cellwise robust methods can manage outlying cells as $\gamma_{cell}$ increases, as expected. CellMG-GMM, cellMCD and cellGMM are the most reliable while OC 
is somewhat robust against an increase in the degree of cell outlyingness. 
When varying the group means (i.e.\ bottom row ``$\mu$ varying''),
especially cellMG-GMM maintains its good performance.
For cellMCD, non-coherency in the mean and covariance structures 
confuses the algorithm; its estimation accuracy and ability to correctly flag the outlying cells deteriorates, see Figure~\ref{fig:sim_balanced2_ALYZCOR_W} (top panel).
In comparison, the cellGMM benefits from less coherent groups due to more distinct clusters. However, cellGMM does not benefit from more clearly distinguished outliers, in contrast to cellMG-GMM and cellMCD.

\begin{figure}[t]
    \centering
    \includegraphics[width=0.9\textwidth, trim = {0cm 
    1cm
    0cm 0.6cm}, clip]{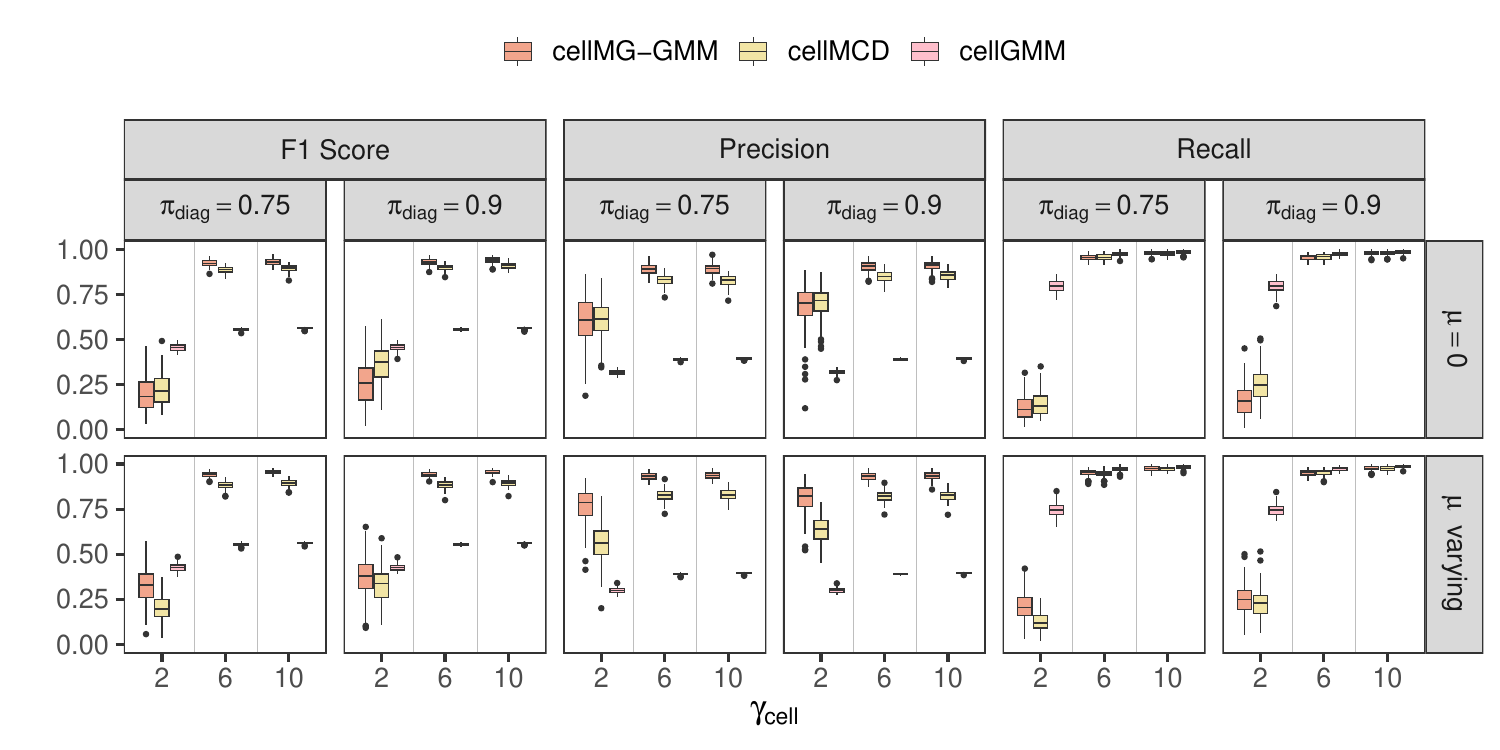}
    \includegraphics[width=0.9\textwidth, trim = {0cm 0.4cm 0cm 1.4cm}, clip]{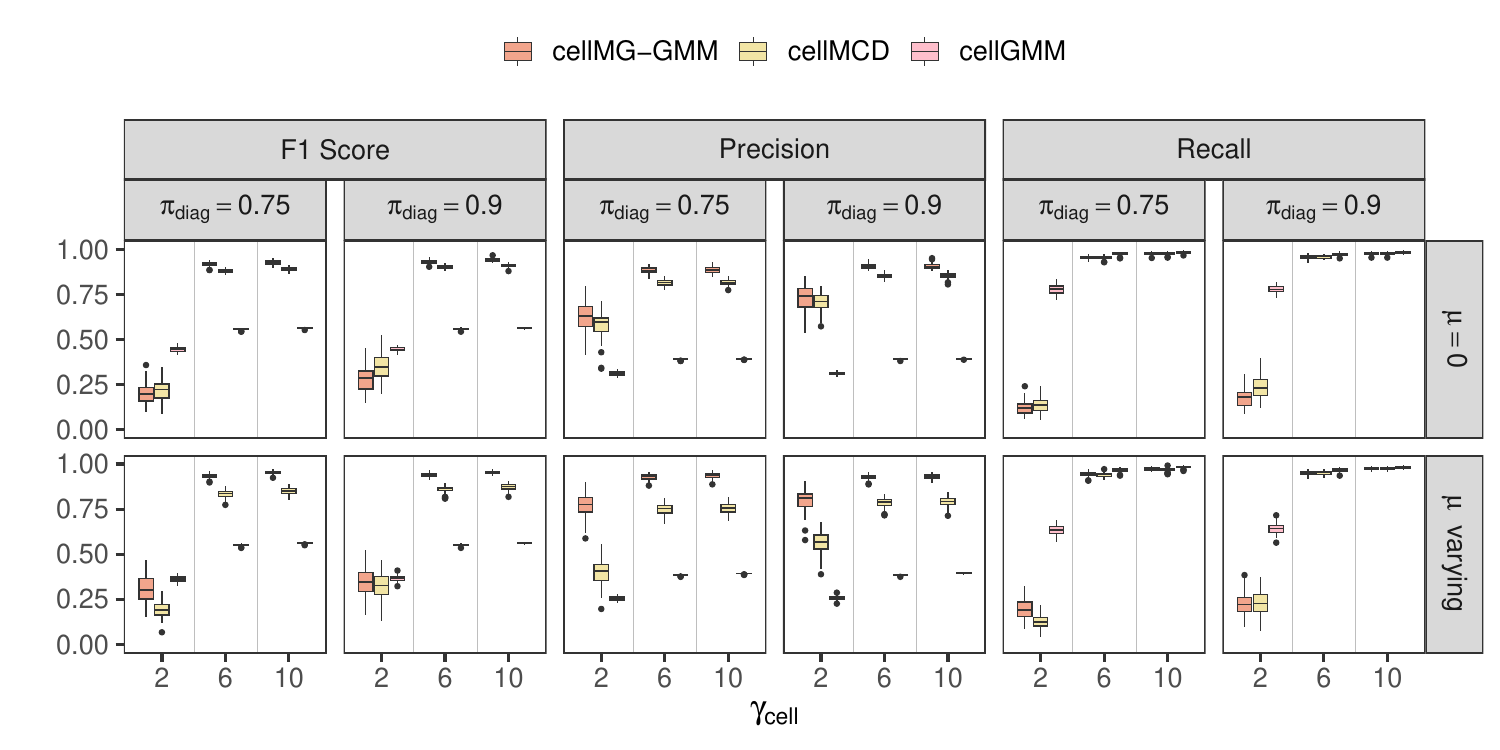}
    \caption{Performance of cellwise outlier detection evaluated by on precision, recall and F1-score for the basic balanced Scenario 1 with $N=2$ 
    (top) and 
    Scenario 2 with $N=5$ (bottom), for varying strength of outlyingness $\gamma_{cell}$.}
    \label{fig:sim_balanced2_ALYZCOR_W}
\end{figure}

In Scenario 2 with 
$N = 5$ groups (bottom panel in Figure~\ref{fig:set1_A0X_sigma}), 
we see similar but even more prominent patterns.
Methods that are not robust to cellwise outliers increasingly suffer with the degree of outlyingness. 
For varying $\mu$, the findings are similar to the basic setting, but we do see that cellMG-GMM performs better than cellMCD and cellGMM in the most and least coherent setting (top right and bottom left panel, respectively). 
The more groups are present among our considered scenarios, the better our proposal can leverage its strengths. 

With respect to the other 
scenarios (detailed results in Appendix~\ref{app:si:results}), the findings are, overall, qualitatively similar. 
The results 
in the unbalanced setting with $N = 2, p= 10, n_1 = 100$ and $ n_2 = 50$ (Scenario 3) are comparable to the balanced settings described above. When increasing the $p$-to-$n$-ratio ($N = 2$, $p = 20$, $n_1 = n_2 = 30$) in Scenario 4, we see that cellMCD and cellGMM struggle to flag cellwise outliers due to low 
estimation accuracy, thereby often delivering worse covariance estimates than the OC 
method. 
In the high dimensional Scenario 5 with $N = 2$, $ p = 60$, $ n_1 = n_2 = 40$, 
cellMG-GMM generally outperforms OC. 

In general, cellMG-GMM consistently performs well across all scenarios.
While it oftentimes performs comparable to cellMCD when $\mu = 0$, in realistic multi-group settings with varying group means,
it outperforms all  considered benchmarks.

\section{Applications}
\label{sec:application}
We demonstrate cellMG-GMM's practical advantages and versatility on diverse applications in
medicine (Section~\ref{subsec:darwin}), oenology (Section~\ref{subsec:wine}) and meteorology (Appendix \ref{app:weather}).

\subsection{Alzheimer Disease: Darwin Data}
\label{subsec:darwin}

Alzheimer is a non-curable neuro-degenerative disease which progresses over time, leading to cognitive impairment. To mitigate its
effects 
on affected patients and their loved ones, early diagnosis and treatment are essential. 
Previous research such as  \citet{Cilia2022} typically distinguishes between $N=2$ groups, namely healthy subjects and diagnosed Alzheimer patients, and train a classifier to discriminate between the groups. While the groups are established 
by an official diagnosis, some subjects can be on the verge to Alzheimer, thereby not yet being diagnosed or only recently. 
Then, a semi-supervised,
modeling approach, like MG-GMM, can  
analyze group intertwinings and highlight contributing factors.

We analyze the DARWIN (Diagnosis AlzheimeR WIth haNdwriting) data set \citep{Cilia2022}, available in the R-package \texttt{robustmatrix} \citep{Mayrhofer2024}, which contains handwriting samples from $n_1 = 85$ healthy subjects and $n_2 = 89$ patients with diagnosed Alzheimer disease (AD). Each subject was asked to execute 25 different handwriting tasks on a tablet from which 18 summary features where extracted: total time, air time, paper time, mean speed on paper, mean speed in air, mean acceleration on paper, mean acceleration on air, mean jerk on paper, mean jerk in air, mean of pressure, variance of pressure, generalization of the mean relative tremor (GMRT) on paper, GMRT in air, mean GMRT, number of pendowns, maximal x-extension, maximal y-extension and dispersion index; see 
\citet{Cilia2018} for more details. Similar to \citet{Mayrhofer2025}, we exclude the variables total time, mean GMRT and air time due to linear dependencies and unreliable measurements. The remaining variables are summarized over the  tasks by the median and median absolute deviation (mad), leading to 
$p = 30$ variables.

We apply the cellMG-GMM estimator (with $h_g = 0.75n_g$) and vary the parameter 
$\alpha \in \{1, 0.99, \ldots, 0.51 ,0.5\}$ to analyze how the 
groups become gradually more overlapping, since a decreasing $\alpha$ allows for more and more group re-assignments.
The left panel of Figure~\ref{fig:AD_residsd} presents the class probabilities for varying $\alpha$ for subjects whose probability $\hat{t}_{g, i, g}$ of being in their initial class $\hat{t}_{g, i, g}$ is lower than 50\% for at least one value of $\alpha$; hence, \textit{switching subjects}. 
A subset of 8 AD diagnosed patients and 2 healthy subjects (i.e.\ the bottom ones in each panel, as visible by the direct gray coloring as soon as $\alpha < 1$) move to the opposite group as soon as a switch is allowed, thereby indicating strong 
similarities 
with the opposite group. 

\begin{figure}
    \centering
    \includegraphics[width=0.95\textwidth, trim={0.6cm 0.7cm 0cm 0},clip]{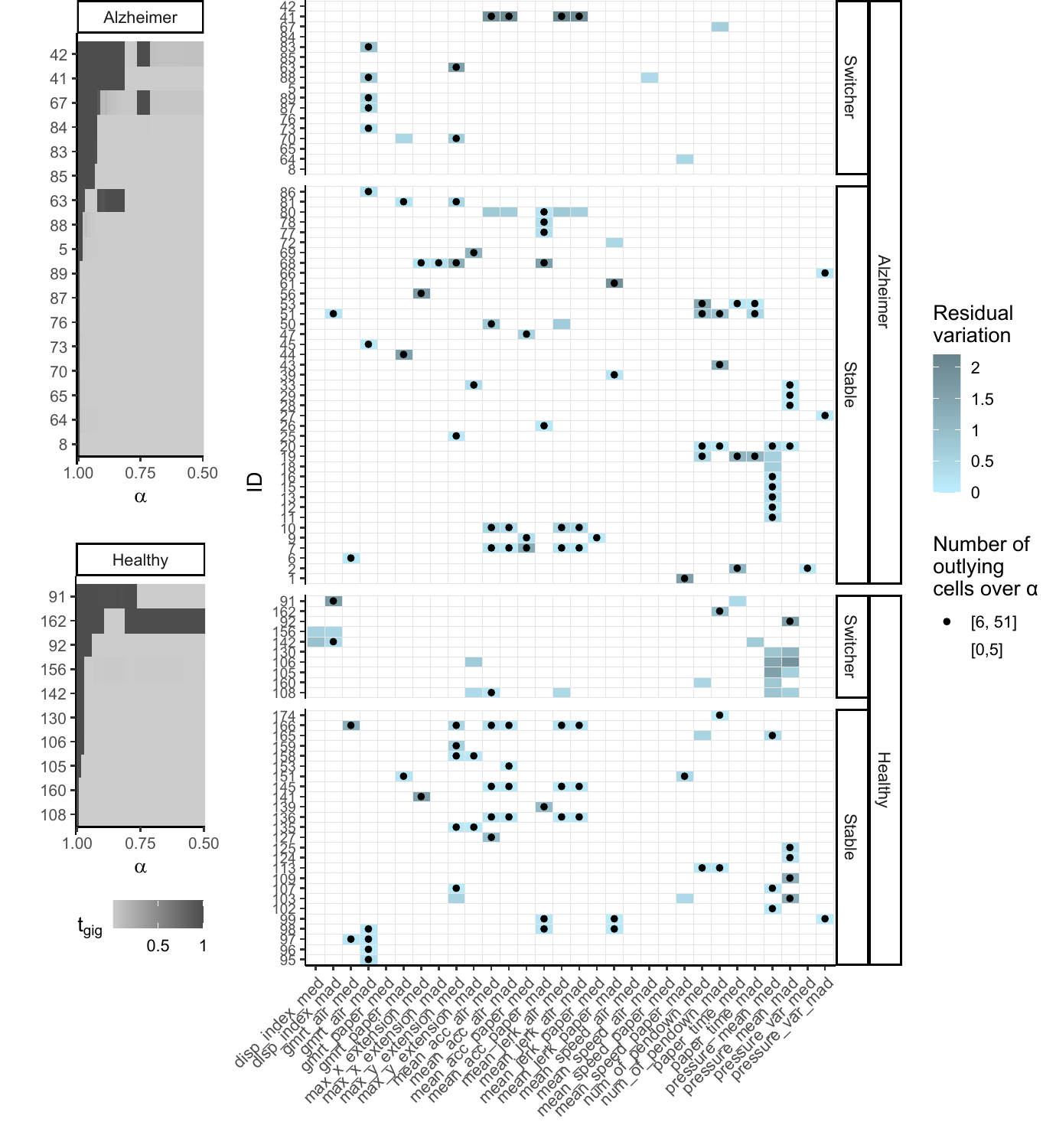}
    \caption{Left: Class probabilities $t_{g,i,g}$ for switching subjects per group (Alzheimer vs Healthy),  sorted by time of switching. Right: Data matrix with subjects (rows) and variables (columns), split by group  and sorted by switching and stable subjects within each group. Cell colors reflect the 
    standard deviation of residuals over $\alpha$. Dotted cells mark frequent outlyingness across different values of $\alpha$.}
    \label{fig:AD_residsd}
\end{figure}

The right panel of Figure~\ref{fig:AD_residsd} shows all cells of the $(n=84) \times (p=30)$ data matrix, including only the subjects for which at least one cell for one value of $\alpha$ is outlying. The subjects are split into Alzheimer patients and healthy subjects, and within each group the switching subjects are separated and sorted as in the left panel. 

Each cell of the data matrix is colored based on the variation of its standardized residuals,
\begin{align*}
    r_{g,ij} = \sum_{k = 1}^N \hat{t}_{g,i,k} \frac{x_{g,ij} - \hat{x}_{g,ij}^k}{\sqrt{
    \hat{\bb \Sigma}_{reg,k}^{(j|j)} - \hat{\bb \Sigma}_{reg,k}^{(j|\hat{\w}_{g,i})} \left(\hat{\bb \Sigma}_{reg,k}^{(\hat{\w}_{g,i}|\hat{\w}_{g,i})} \right)^{-1} \hat{\bb \Sigma}_{reg,k}^{(\hat{\w}_{g,i}|j)}
    }},
\end{align*}
over varying $\alpha$, 
where $\hat{x}_{g,ij}^k$ denotes the expected value of $x_{g,ij}$ 
assuming that it comes from distribution $k$ and using only unflagged cells $\hat{\w}_{g,i}$.
White cells indicate non-outlyingness across all $\alpha$.

Alzheimer patient 8 switches immediately to the healthy group without any change in residuals (i.e.\ no coloring). This patient is at the overlap of the two groups with respect to all variables, but it is relatively closer to the center of the healthy group.  Such a subject is likely to have an early diagnosis and low cognitive impairment.

Cells marked by a dot are outlying for several (i.e.\ 6 or more out of 51) values of $\alpha$, and  the cell color reflects the standard deviation of the residuals over varying $\alpha$.
Higher residual variability can occur for different reasons: (a) the subject switches to the other group, (b) the cell is identified as an outlier for particular values of $\alpha$, or both (a) and (b) occur.
The variables \texttt{pressure\_mean} (both median and mad) display many cells with high residual variability. Several of those cells are outlying (i.e.\ marked by a dot) as soon as the given diagnosis is no longer enforced, thereby revealing the inhomogeneity of these subjects with respect to the variables  \texttt{pressure\_mean}. 
There is, however, also a block of cells for the variables  \texttt{pressure\_mean} that is not outlying (i.e\ colored cells without dots). These subjects switch from the healthy to the AD group as the latter provides a better model fit. cellMG-GMM suggests that a closer inspection of the patients, possibly being in transition, and the variable \texttt{pressure\_mean} is needed since either unfavorable measurement conditions or other undiagnosed or progressive diseases affecting it could lead to this unusual behavior. 

\subsection{Wine Quality Data}
\label{subsec:wine}
We use a data set of \citet{Cortez2009}, available at the UCI Machine Learning Repository \citep{Wine2009} that consist of $p = 11$ physicochemical measurements, including fixed acidity, volatile acidity, citric acid, residual sugar, chlorides, free sulfur dioxide, total sulfur dioxide, density, pH-level, sulphates, and alcohol percentage, for $n = 4898$ samples of white \textit{vinho verde}, a popular Portuguese wine. Each wine was qualitatively graded from 0 (very bad) to 10 (excellent) by three different sensory assessors through blind tasting. The median of the three grades is reported as the variable quality.

Originally, \citet{Cortez2009} trained a Support Vector Machine classifier given the quality variable. 
We, in contrast, aim to leverage MG-GMM's flexibility to investigate how qualitative expert evaluations of wine are consistent with their quantitative chemical features.
We therefore partition the data into $N=3$ groups based on the quality assessment: the first group with low wine quality includes $n_1 = 1640$ wine samples with quality assessments 3 to 5, 
the second group with medium quality contains $n_2 = 2198$ samples with quality level 6, and the third group includes $n_3 = 1060$ good quality wine samples with quality assessments 7 to 10. 
Due to prominent skewness in multiple variables, we apply a robust transformation to each variable to achieve central normality \citep[see][]{Raymaekers2024}. We then apply the cellMG-GMM estimator with $h_g = 0.75n_g$ and $\alpha = 0.75$; taking $\alpha > 0.5$ stabilizes the estimation due to the low number of unbalanced groups and some incoherency within the groups.

\begin{figure}[t]
    \centering
    \includegraphics[width=1\textwidth]{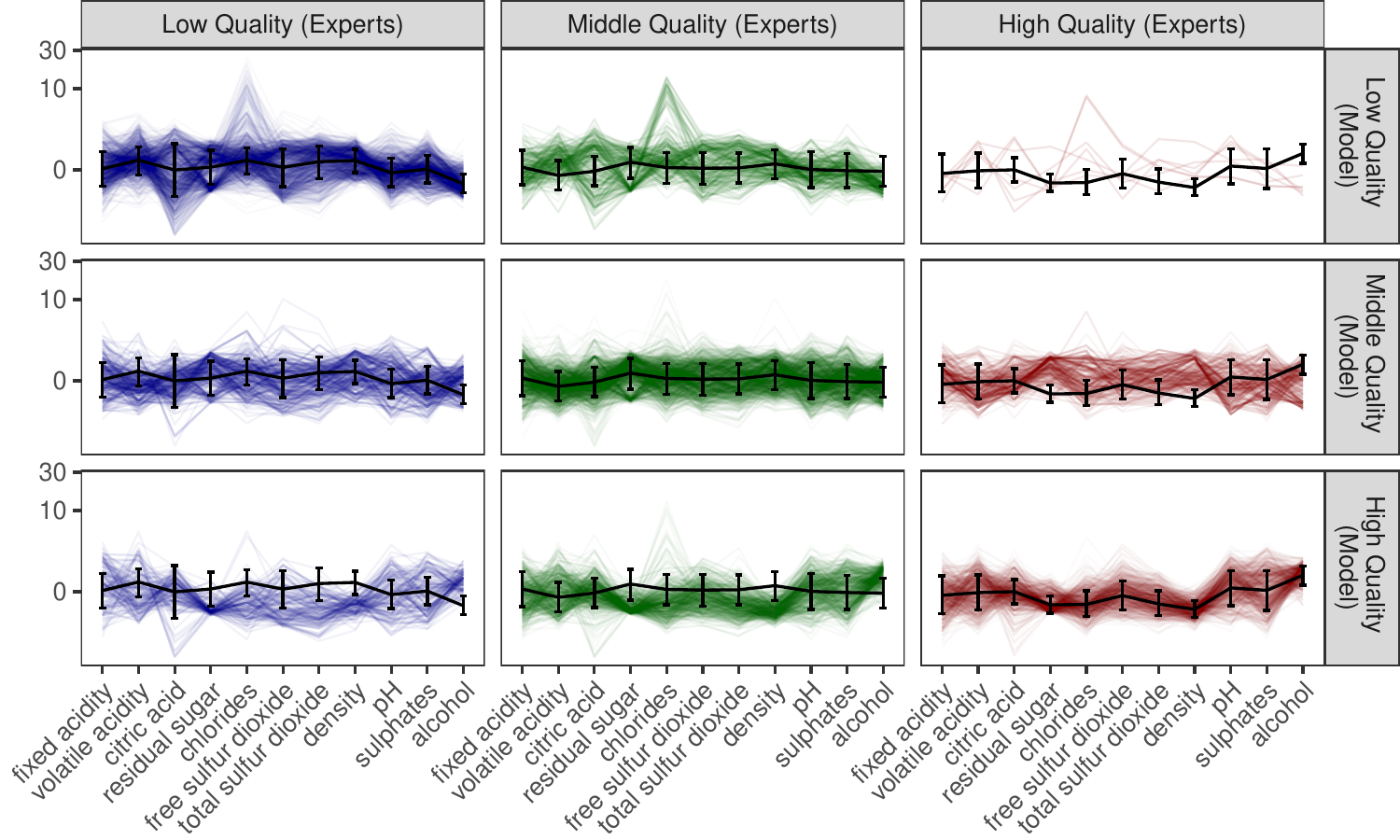}
    \caption{
    Discrepancies between the
    wine quality assessment of the experts (columns) and the model-based grouping (rows) based on the physicochemical features. Black lines show estimated location and standard deviation for expert groups, colored lines show wine measurements corresponding to each panel.}
    \label{fig:wine_white_pcp}
\end{figure}

The parallel coordinate plot in Figure~\ref{fig:wine_white_pcp}
highlights the discrepancies between the predefined expert labels (columns) and the model-based groupings (rows).
Diagonal panels highlight wine samples where both agree on their quality.
Panels below the main diagonal show wine samples that experts rate lower than their physicochemical measurements would suggest, and vice versa for the panels above the main diagonal.
Each panel includes the estimated location (solid black line) and standard deviation (black error bars) provided by the cellMG-GMM for the expert-proposed group; these are thus identical in each column. We notice a strong heterogeneity within each expert group. While the wine samples where experts and cellMG-GMM agree are quite coherent, clear structural differences are visible for the discrepant cases. 
The two bottom left panels show quantitatively good wines that are rated low by experts. They differ most prominently from less qualitative wines by low density and residual sugar while containing a relatively high amount of alcohol. On the contrary,  wines rated too high by experts (middle right panel) show adverse results for residual sugar, density and alcohol. 

Moreover, there are many cellwise outliers detected by cellMG-GMM that are also visible in the parallel coordinate plot. Especially the many outlying chloride values are noticeable, as well as the low citric acid values. Robustness against cellwise outliers, which cellMG-GMM provides, is thus key to get reliable estimates and to avoid clusters being dominated by one variable with many extreme values.

\section{Conclusion}
\label{sec:conclusion}
We propose a probabilistic multi-group Gaussian mixture model, MG-GMM, that accounts for expert or context-based group information and delivers 
(i) model-based groupings where observations may be flexibly reassigned to other groups based on data-driven evidence, and
(ii) outlier-robust moment estimates that can be one-to-one matched to the predefined groups. 
The combination of these features has not yet been offered by other methods.

To obtain the mixture parameter estimates and jointly identify cellwise outliers, we introduce cellMG-GMM, a penalized observed likelihood-based estimator for which we provide an EM-based algorithm that is carefully tailored towards the multi-group setting.
A key ingredient of cellMG-GMM is the parameter $\alpha$ that regulates the strictness of the initial group membership, or put alternatively the flexibility in terms of group reassignments. As 
$\alpha$ is varied, it can thus shed light on the transition dynamics of observations across groups. The parameter $\alpha$ hereby bridges the gap between separate group-specific parameter estimation with no reassignment ($\alpha = 1$) and a  (cellwise robust) yet standard GMM with a given number of clusters in the other extreme 
case ($\alpha = 0$); which we exclude since we assume that each pre-specified group has a main distribution assigned to it ($\alpha \geq 0.5$). The cellMG-GMM is therefore methodologically located between the cellMCD, applied to each group separately, and the cellGMM. 

A further theoretical contribution of our work is the introduction of an appropriate notion of breakdown in the multi-group setting with cellwise contamination. 
We describe  a novel idealized setting of well-clustered and cellwise-contaminated data for which the robustness properties can theoretically be evaluated and compared across  different methods. 
This idealized setting is of independent, general interest for cluster and finite mixture settings characterized by cellwise (instead of rowwise) contamination, and we directly extend  it to the multi-group GMM setting to  prove the breakdown properties of the cellMG-GMM.
The good robustness properties are confirmed in extensive simulation experiments.  

CellMG-GMM is applicable across many fields of research where assignments to pre-defined groups need to be viewed more flexibly, in a semi-supervised way. We  demonstrate the practical advantages of cellMG-GMM on versatile examples where the rich output produced by it allows for different interpretation angles. 
Future research might leverage the moment estimates delivered by cellMG-GMM for other prominent multivariate analyses like principal component analysis, discriminant analysis or graphical modeling.

\textbf{Acknowledgments:} We thank Jakob Raymaekers for comments provided on earlier versions of the paper. This work is
co-funded by the European Union (SEMACRET, Grant Agreement no. 101057741) and UKRI (UK Research and Innovation). Ines Wilms is supported by a grant from the Dutch Research Council (NWO), research
program Vidi under grant number VI.Vidi.211.032.

\section*{Disclosure statement}\label{disclosure-statement}

The authors have the no conflicts of interest to declare.

\section*{Data Availability Statement}\label{data-availability-statement}
The wine data are openly available in the UCI Machine Learning Repository at \linebreak \href{http://doi.org/10.24432/C56S3T}{http://doi.org/10.24432/C56S3T}. The DARWIN data set is available in the R-package \texttt{robustmatrix} \citep{Mayrhofer2024} and the weather data is available in the R-package \texttt{ssMRCD} \citep{ssMRCD_Cran}, both are openly available on CRAN.

\bibliographystyle{apalike}
\bibliography{biblio}

\newpage
 
\appendix

\renewcommand{\theequation}{A\arabic{equation}} 
\setcounter{equation}{0}
\setcounter{page}{1}

\begin{center}
{\LARGE Supplementary material to: ‘Outlier-Robust Multi-Group Gaussian Mixture Modeling with Flexible Group Reassignment'}
\end{center}

Section~\ref{sec:appendix_algorithm} contains details of the algorithm.  Section~\ref{sec:appendix_theory} contains the proof of Theorem~1. Section~\ref{app:simulations} contains additional details and results for the simulations and Section~\ref{app:weather} contains the weather example.

\section{Details of the EM-Algorithm}
\label{sec:appendix_algorithm}

In this appendix further details on the proposed EM-algorithm are provided.

\subsection{Initialization}
\label{subsec:appendix_initialvalues}
First, all data sets are standardized robustly on a global scale (thus ignoring the group structure). 
This leads to global scale and shift invariance and is helpful to stabilize the regularization of the covariance matrices. Note that the final estimates, obtained after convergence of the algorithm,  are rescaled to the original scale.

For a given $\alpha$, the initial estimate for $\hat{\bb \pi}^0$ has $\hat{\pi}^0_{g,g} = \alpha$ and $\hat{\pi}^0_{g,k} = (1-\alpha)/(N-1)$ for $g \neq k$.
We then use the DDCW algorithm of \cite{Raymaekers_Rousseeuw_2021}, applied separately for each group, to get initial estimates $\hat{\bb{\Sigma}}_{reg,k}^0 $ and $\hat{\bb{\mu}}_k^0$, in line with \citet{Raymaekers2023}. 
We hereby assume that each group has a main distribution as enforced by $\pi_{g,g} \geq \alpha \geq 0.5$ for all $g = 1, \ldots, N$. Thus, taking a robust estimate of the covariance and mean of the main bulk of the observations for each group separately is reasonable and a good initial estimate of the corresponding main distribution. To ensure regularity also in cases with low number of observations in a group $k$, each time a covariance is calculated by the DDCW-algorithm, it is regularized with regularization matrix $\bb T_k$ and an adaptive regularization factor $\rho_k$ ensuring a maximal condition number of $\kappa_k$, as detailed in Section~\ref{subsec:algorithm_parameters}.
Finally, the entries of the matrices $\bb{W}^0$ are all set to one, as in \citet{Raymaekers2023}.

\subsection{EM-Step}
\label{subsec:appendix_EM}

The Expectation-Maximization (EM) algorithm
is often used to find maximum likelihood estimates in settings where the data is incomplete.
In our setting, the missingness pattern is indicated by $\bb W$, which is not known in advance but is estimated in the W-step of the algorithm. Conditional on the current $\bb W$, the EM-step then updates the parameters of the mixture model. 

For each observation $\x_{g,i}$ a binary random variable $z_{g,i,k}$ indicates whether it was drawn from distribution $k$. The likelihood resulting from including the additional random variables $z_{g,i,k}$ is called the \textit{complete log-likelihood} and the resulting objective function, the \textit{complete objective function} $\operatorname{CObj}(\bb{\pi}, \bb{\mu}, \bb{\Sigma}, \bb{W}, \bb{Z})$, is $-2$ times
\begin{align*}
     \sum_{g = 1}^N \sum_{i = 1}^{n_g} \Bigg[\sum_{\substack{k = 1\\\pi_{g, k}\neq 0}  }^N  z_{g,i,k} \ln \left( \pi_{g, k}\varphi\left(\x_{g,i}^{(\w_{g,i})}; \bb{\mu}_k^{(\w_{g,i})}, \bb{\Sigma}_{reg,k}^{(\w_{g,i})}\right)\right) +\sum_{j = 1}^{p} q_{g, i, j} (1-w_{g, ij}) \Bigg],  
\end{align*}
where $\bb{Z}$ collects all $z_{g,i,k}$. Taking the conditional expectation of $z_{g,i,k}$ gives
\begin{align*}
    t_{g,i,k} = \E[z_{g,i,k}|\x_{g,i}^{(\w_{g,i})}, \bb{\pi}, \bb{\mu}, \bb{\Sigma}, \bb{W}] = \frac{\pi_{g, k} \varphi\left(\x_{g,i}^{(\w_{g,i})}; \bb{\mu}_k^{(\w_{g,i})}, \bb{\Sigma}_{reg,k}^{(\w_{g,i})}\right) }{\sum_{l = 1}^N \pi_{g, l}\varphi\left(\x_{g,i}^{(\w_{g,i})}; \bb{\mu}_l^{(\w_{g,i})}, \bb{\Sigma}_{reg,l}^{(\w_{g,i})}\right)}.
\end{align*}
The \textit{expected objective function} $\operatorname{EObj}(\bb{\pi}, \bb{\mu}, \bb{\Sigma}, \bb{W}) $, is then $-2$ times
\begin{align}
    \sum_{g = 1}^N \sum_{i = 1}^{n_g} \Bigg[ \sum_{\substack{k = 1\\\pi_{g, k}\neq 0}}^N  t_{g,i,k} \ln \left( \pi_{g, k}\varphi\left(\x_{g,i}^{(\w_{g,i})}; \bb{\mu}_k^{(\w_{g,i})}, \bb{\Sigma}_{reg,k}^{(\w_{g,i})}\right)\right) + \sum_{j = 1}^{p}  q_{g, i, j} (1-w_{g, ij}) \Bigg].  \label{eq:eob} 
\end{align}
The EM algorithm then leverages that we can iteratively take the expectation and maximize the Expected Objective~\eqref{eq:eob}. 

Extending the
maximization step regarding the parameters $\bb \mu$ and $\bb \Sigma$ for the GMM with missing values \citep{Eirola2014} to the multi-group GMM with missing values is straightforward since the group structure can be ignored once the conditional expectation of $z_{g,i,k}$ is calculated.
The only difference is the estimation of the mixture probabilities $\bb{\pi}$ due to the constraints $\sum_{k = 1}^N \pi_{g,k} = 1$  and $\pi_{g,g} \geq \alpha$ for all $g = 1, \ldots, N$. To find the optimal mixture probability, the Karush-Kuhn-Tucker theorem can be applied. Setting the derivative of the Expected Objective~\eqref{eq:eob} with respect to $\pi_{g,l}$ to zero, then the following conditions have to hold
\begin{equation*}
    \frac{\partial [EObj + \lambda(1-\sum_{k = 1}^N \pi_{g,k}) + \mu (\alpha - \pi_{g,g})]}{\partial \pi_{g,l}}  = \mu (\alpha - \pi_{g,g}) = 1-\sum_{k = 1}^N \pi_{g,k} = 0,
\end{equation*}
as well as $\mu \geq 0$.
Plugging in the 
formula from Equation~\eqref{eq:eob} leads to 
\begin{align*}
    \lambda &= \frac{-2 \sum_{i = 1}^{n_g} \sum_{l= 1, l\neq g}^N t_{g,i,l}}{(1-\pi_{g,g})} = \frac{-2 \sum_{i = 1}^{n_g} (1- t_{g,i,g})}{(1-\pi_{g,g})}.
\end{align*}
Plugging $\lambda$ in leads to
\begin{align*}
     \pi_{g, l} =& \frac{(1-\pi_{g,g}) \sum_{i = 1}^{n_g}  t_{g,i,l} }{ \sum_{i = 1}^{n_g} (1-t_{g,i,g})} = (1-\pi_{g,g})\frac{\frac{1}{n_g}\sum_{i = 1}^{n_g}  t_{g,i,l} }{ 1 -\frac{1}{n_g} \sum_{i = 1}^{n_g}t_{g,i,g}}.
\end{align*}
For the Lagrange parameter $\mu$, we finally have 
\begin{align*}
   \frac{- \frac{1}{n_g}\sum_{i = 1}^{n_g}  t_{g,i,g}}{\pi_{g, g}} + \frac{ (1-\frac{1}{n_g}\sum_{i = 1}^{n_g} t_{g,i,g})}{(1-\pi_{g,g})} =& \frac{\mu}{2n_g} \geq 0\\
    \frac{\pi_{g, g}}{(1-\pi_{g,g})} \geq& \frac{\frac{1}{n_g}\sum_{i = 1}^{n_g}  t_{g,i,g}}{(1-\frac{1}{n_g}\sum_{i = 1}^{n_g} t_{g,i,g})}.
\end{align*}
Since $f(x) = x/(1-x)$ is monotonously increasing, this 
holds if $\pi_{g, g} \geq \frac{1}{n_g}\sum_{i = 1}^{n_g}  t_{g,i,g}$. Thus, if the inequality is strict, $\mu > 0$ and $\pi_{g, g}  = \alpha$. Otherwise, $\pi_{g, g} = \frac{1}{n_g}\sum_{i = 1}^{n_g}  t_{g,i,g}$ is a feasible solution which is equal to the solution of the unconstrained minimization problem. Overall, we have
\begin{align*}
    \pi_{g, g} = \max\left\{\alpha, \frac{1}{n_g}\sum_{i = 1}^{n_g} t_{g,i,g} \right\}, \quad  \pi_{g, l}  = (1-\pi_{g,g})\frac{\frac{1}{n_g}\sum_{i = 1}^{n_g}  t_{g,i,l} }{ 1 -\frac{1}{n_g} \sum_{i = 1}^{n_g}t_{g,i,g}}.
\end{align*}
Note that the regularity condition \textit{linear independence constraint qualification} (LICQ) is fulfilled for all feasible $\bb \pi$.

Thus, the mixture probability estimates that fulfill the constraints $\sum_{k = 1}^N \pi_{g, k}  = 1$ and $\pi_{g,g} \geq \alpha \geq 0.5 $ for all $g = 1, \ldots, N$ in iteration step~$\tau +1$ are given by
\begin{align*}
    \hat{\pi}_{g, g}^{\tau+1} = \max\left\{\alpha, \frac{1}{n_g}\sum_{i = 1}^{n_g} \hat{t}_{g,i,g}^{\tau+1} \right\}, \quad  \hat{\pi}_{g, k}^{\tau+1}  = (1-\hat{\pi}_{g, g}^{\tau+1})\frac{\frac{1}{n_g}\sum_{i = 1}^{n_g}  \hat{t}_{g,i,k}^{\tau+1} }{ 1 -\frac{1}{n_g} \sum_{i = 1}^{n_g} \hat{t}_{g,i,g}^{\tau+1}},
\end{align*}
where $\hat{t}_{g,i,k}^{\tau+1}$ denotes the estimated expected probability that observation $\x_{g,i}$ is from distribution $k$ given the estimates (${\bb{\hat{\Sigma}}_{reg}^{\tau}}$, $\hat{\bb{\mu}}^\tau$, $\hat{\bb{\pi}}^\tau$, $\hat{\bb{W}}^{\tau+1}$) from the previous step,
\begin{align}
\label{eq:expected_t}
    \hat{t}_{g,i,k}^{\tau+1} = \frac{\hat{\pi}_{g, k}^\tau \varphi\left(\x_{g,i}^{(\hat{\w}_{g,i}^{\tau+1})}; {\bb{\hat{\mu}}_k^\tau}^{(\hat{\w}_{g,i}^{\tau+1})}, {\bb{\hat{\Sigma}}_{reg,k}^\tau}^{(\hat{\w}_{g,i}^{\tau+1})}\right) }{\sum_{l = 1}^N \hat{\pi}_{g, l}^\tau \varphi\left(\x_{g,i}^{(\hat{\w}_{g,i}^{\tau+1})}; {\bb{\hat{\mu}}_l^\tau}^{(\hat{\w}_{g,i}^{\tau+1})}, {\bb{\hat{\Sigma}}_{reg,l}^\tau}^{(\hat{\w}_{g,i}^{\tau+1})}\right)}.
\end{align}

The new estimates for the group-specific means are given by
\begin{align*}
    {\bb{\hat{\mu}}_k^{\tau+1}} = \frac{1}{\bar{t}_k} \sum_{g = 1}^N \sum_{i = 1}^{n_g} \hat{t}_{g,i,k}^{\tau+1} {\bb {\hat{x}}_{g,i,k}^{\tau+1}},
\end{align*}
with $ \bar{t}_k= \sum_{g = 1}^N \sum_{i = 1}^{n_g} \hat{t}_{g,i,k}^{\tau+1}$ and conditional expectations $\hat{\x}^{\tau+1}_{g,i,k}$ given by 
\begin{small}
\begin{align} 
    {\bb {\hat{x}}_{g,i,k}^{\tau+1}}^{(1-\hat{\w}^{\tau+1}_{g,i})} &= {\bb{\hat{\mu}}_k^\tau}^{(1-\hat{\w}^{\tau+1}_{g,i})} + {\bb{\hat{\Sigma}}_{reg,k}^\tau}^{(1-\hat{\w}^{\tau+1}_{g,i}| \hat{\w}^{\tau+1}_{g,i})}\left({\bb{\hat{\Sigma}}_{reg,k}^\tau}^{(\hat{\w}^{\tau+1}_{g,i}| \hat{\w}^{\tau+1}_{g,i})} \right)^{-1} \left( \x_{g,i}^{(\hat{\w}^{\tau+1}_{g,i})} - {\bb{\hat{\mu}}_k^\tau}^{(\hat{\w}^{\tau+1}_{g,i})}\right) \nonumber \\ 
    {\bb {\hat{x}}_{g,i,k}^{\tau+1}}^{(\hat{\w}^{\tau+1}_{g,i})} &=  \x_{g,i}^{(\hat{\w}^{\tau+1}_{g,i})}, \label{eq:imputed1}
\end{align}
\end{small} 
for an observation $\x_{g,i}$ with missingness pattern $\hat{\w}^{\tau+1}_{g,i}$, assuming that it comes from distribution $k$.

Finally, the new estimates of the regularized covariance matrices are
\begin{align*}
   {\bb{\hat{\Sigma}}_{reg,k}^{\tau+1}} = \rho_k {\bb T}_k + (1- \rho_k)\frac{1}{\bar{t}_k} \sum_{g = 1}^N \sum_{i = 1}^{n_g} \hat{t}_{g,i,k}^{\tau+1} \left[  ({\bb {\hat{x}}_{g,i,k}^{\tau+1}} - {\bb{\hat{\mu}}_k^{\tau+1}}) ({\bb {\hat{x}}_{g,i,k}^{\tau+1}} - {\bb{\hat{\mu}}_k^{\tau+1}})' +  \bb{\tilde{\Sigma}}_{reg,k}^\tau \right]
\end{align*}
with
\begin{align*}
    {\bb{\tilde{\Sigma}}_{reg,k}^\tau}^{(1-\hat{\w}^{\tau+1}_{g,i}| 1 -\hat{\w}^{\tau+1}_{g,i})}  &= {\bb{\hat{\Sigma}}_{reg,k}^\tau}^{(1-\hat{\w}^{\tau+1}_{g,i}| 1-\hat{\w}^{\tau+1}_{g,i})} - {\bb{\hat{\Sigma}}_{reg,k}^\tau}^{(1-\hat{\w}^{\tau+1}_{g,i}|\hat{\w}^{\tau+1}_{g,i})} \\
    & \quad \times \left( {\bb{\hat{\Sigma}}_{reg,k}^\tau}^{(\hat{\w}^{\tau+1}_{g,i}|\hat{\w}^{\tau+1}_{g,i})} \right)^{-1} {\bb{\hat{\Sigma}}_{reg,k}^\tau}^{(\hat{\w}^{\tau+1}_{g,i}| 1-\hat{\w}^{\tau+1}_{g,i})},
\end{align*}
for unobserved variables ($\hat{\w}^{\tau+1}_{g,i}$ equal to $0$), all other entries of ${\bb{\tilde{\Sigma}}_{reg,k}^\tau}$ are equal to zero.

\subsection{Algorithm Pseudo-code}
\label{subsec:app_pseudo}
Pseudo-code for the algorithm is compactly presented in Algorithm \ref{alg:main}.

\begin{algorithm}[H]
\caption{Cellwise-robust estimation of the multi-group GMM}
\label{alg:main}
\begin{algorithmic}[1] 
    \Require $\X_1, \X_2, \ldots, \X_N$; initial estimates ${\bb{\hat{\Sigma}}_{reg}^{0}}$, $\hat{\bb{\mu}}^0$, $\hat{\bb{\pi}}^0$, $\hat{\bb{W}}^0$; hyperparameters $q_{g, ij}$, $\bb T_k$, $\rho_k$, $\epsilon_{conv}$, $h_g$, $\alpha$
    \State $\bb{W} \leftarrow \hat{\bb{W}}^0$
    \State $(\bb{\Sigma}_{reg}, \bb{\mu}, \bb{\pi}) \leftarrow ({\bb{\hat{\Sigma}}_{reg}^{0}}, \hat{\bb{\mu}}^0, \hat{\bb{\pi}}^0) $
    \State $\texttt{crit} \leftarrow \infty$
    \While{\texttt{crit} $>$ $\epsilon_{conv}$}
        \State $\bb{\Sigma}_{reg}^{prev} \leftarrow \bb{\Sigma}_{reg}$
        \State $\bb{W} \leftarrow$ \texttt{wstep}($\X, \bb{\Sigma}_{reg}, \bb{\mu}, \bb{\pi}, \bb{W}, q_{g, ij}, h_g$)
        \State $(\bb{\Sigma}_{reg}, \bb{\mu}, \bb{\pi}) \leftarrow$ \texttt{emstep}($\X, \bb{\Sigma}_{reg}, \bb{\mu}, \bb{\pi}, \bb{W}, \bb T, \rho, \alpha$)
        \State \texttt{crit} $\leftarrow \max_{k,j,j'} |{\Sigma}_{reg,k, jj'}^{prev} - {\Sigma}_{reg,k, jj'}|$
    \EndWhile
    \State \Return $\bb{\Sigma}_{reg}, \bb{\mu}, \bb{\pi}, \bb W$
\end{algorithmic}
\end{algorithm} 

\renewcommand{\thecorollary}{B\arabic{corollary}} 
\setcounter{corollary}{0}
\renewcommand{\theequation}{B\arabic{equation}} 
\setcounter{equation}{0}
\section{Derivations of the Breakdown Point}
\label{sec:appendix_theory}

We start with a preliminary result (Section \ref{app:subsec:corollary}), and then present the proofs of Theorem~\ref{theorem:breakdownpoint} (Section \ref{app:subsec:theorem}). 
For ease of notation across all proofs, we drop the superscript $m$ for observations and the explicit dependence of the estimators on $\mathcal{Z}_m$ or $\mathcal{X}_m$ when possible. All limits correspond to $m \rightarrow \infty$. The notation $\w(\bb{y})$ marks the real outlying cells of $\bb{y}$ while the notation $\w_{\bb{y}}$ indicates the missingness pattern of $\bb{y}$ for a given $\bb W$ from the objective function if the indexation of $\bb{y}$ is irrelevant.

\subsection{Preliminary Result} \label{app:subsec:corollary}

\begin{corollary} \label{theorem:corollary}
Given the idealized setting (Section \ref{subsec:bdp_mixture} and \ref{subsec:bdp_grouped}) and fixed $\rho_k > 0,  \bb{T}_k \succ 0$ (positive definite), the following statements hold.
\begin{itemize}
    \item[a.] 
    For uncontaminated data $\mathcal{Z}_m = \mathcal{X}_m$ ($m \in \mathbb{N}$), there exist feasible estimates $\hat{\bb{\pi}}$, $\hat{\bb{\mu}}$, $\hat{\bb{\Sigma}}$ such that
    $\operatorname{Obj}(\hat{\bb{\pi}}$, $\hat{\bb{\mu}}$, $\hat{\bb{\Sigma}}, \bb{W})$
    is finite for any feasible set of $\bb{W}$.
    \item[b.] For contaminated data $\mathcal{Z}_m$ ($m \in \mathbb{N}$) and sets of estimates $\hat{\bb{\pi}}  (\mathcal{Z}_m)$, $\hat{\bb{\mu}}(\mathcal{Z}_m)$, $ \hat{\bb{\Sigma}}(\mathcal{Z}_m) $, $\hat{\bb{W}}(\mathcal{Z}_m)$: 
    \begin{itemize}
    \item[b1.] If there exists an $l$ such that $ \lambda_1(\hat{\bb{\Sigma}}_{reg, l}(\mathcal{Z}_m)) \rightarrow \infty$ for $m \rightarrow \infty$, then 
    $\operatorname{Obj}(\hat{\bb{\pi}}  (\mathcal{Z}_m), \linebreak \hat{\bb{\mu}}(\mathcal{Z}_m), \hat{\bb{\Sigma}}(\mathcal{Z}_m), \hat{\bb{W}}(\mathcal{Z}_m))$
    goes to infinity. 
    \item[b2.] If there exists a variable $j^*$, $l$, $k$ and a constant $\tilde{b}$ such that $|\hat{{\mu}}_{k,j^*}(\mathcal{Z}_m) - \hat{{\mu}}_{l,j^*}(\mathcal{Z}_m)| < \tilde{b}$ for $l \neq k$, 
    then 
    $\operatorname{Obj}(\hat{\bb{\pi}}  (\mathcal{Z}_m), \hat{\bb{\mu}}(\mathcal{Z}_m), \hat{\bb{\Sigma}}(\mathcal{Z}_m), \hat{\bb{W}}(\mathcal{Z}_m))$
    goes to infinity.
    \item[b3.] Given any feasible set of $\bb{W}$ with finite objective function, then, for all groups $g$ and  observations $\x_{g,i} \in (A^g \cup B^g) \cap \bb{Z}^g$ there exists exactly one estimate $\hat{\bb\mu}_k(\mathcal{Z}_m)$ with
        $||\x_{g,i}^{(\w_{g,i})} - \hat{\bb\mu}_k(\mathcal{Z}_m)^{(\w_{g,i})} ||  < \infty.$
\end{itemize}
\end{itemize}
\end{corollary}

\begin{proof}

First note that in the following, the penalty term can generally be left out since it is always bounded, $ |\sum_{g = 1}^N \sum_{i = 1}^{n_g} \sum_{j = 1}^{p} q_{g, ij} (1-w_{g, ij})| \leq pN \max_g n_g \max_{g, i, j} q_{g, ij} < \infty.$

\textbf{Proof of part a.}        
        Given a data matrix $\mathcal{X}$, we construct a set of estimators with finite objective function value.
        For all $k = 1, \ldots, N$ set $\hat{\Sigma}_{k ,jj} = 1$ and zero otherwise and $\hat{\bb \mu}_{k} = \frac{1}{| A^k_m|}\sum_{\x \in A^k_m} \x$, where $ |A^k_m|$ denotes the number of elements in $A^k_m$. Then, the regularized covariance matrices $\hat{\bb\Sigma}_{reg, k}$ have finite positive eigenvalues. Consider two cases for $\alpha$:
        
            First, assume $\alpha \neq 1$. Set $\hat{\pi}_{k,k} = \alpha \geq 0.5$, $\hat{\pi}_{k,l} = \frac{1-\alpha}{N-1} > 0$ for $k \neq l$. 
             For each observation $\x_{g,i}$ with $\w_{g,i}$ originating from any cluster~$l$ it holds that
            \begin{align*}
                \ln \Bigg(&\sum_{k = 1}^N  \hat{\pi}_{g, k}  \varphi\left(\x_{g,i}^{(\w_{g,i})}; \hat{\bb\mu}_k^{(\w_{g,i})}, \hat{\bb\Sigma}_{reg,k}^{(\w_{g,i})}\right)\Bigg) \\
                & \geq \ln \left( \frac{1-\alpha}{N-1}  \varphi\left(\x_{g,i}^{(\w_{g,i})}; \hat{\bb\mu}_l^{(\w_{g,i})}, \hat{\bb\Sigma}_{reg,l}^{(\w_{g,i})}\right)\right) \\
                & = \ln \frac{1-\alpha}{N-1}  + \ln   \frac{e^{-\frac{1}{2}(\x_{g,i}^{(\w_{g,i})} - \hat{\bb\mu}_l^{(\w_{g,i})})' (\hat{\bb\Sigma}_{reg,l}^{(\w_{g,i})})^{-1} (\x_{g,i}^{(\w_{g,i})} - \hat{\bb\mu}_l^{(\w_{g,i})})}}{\sqrt{(2\pi)^{\sum_j w_{g,ij}} \det \hat{\bb\Sigma}_{reg,l}^{(\w_{g,i})}}} \\
                 & \geq \ln \frac{1-\alpha}{N-1}  -\frac{1}{2} (\bb{b}^{(\w_{g,i})})' (\hat{\bb\Sigma}_{reg,l}^{(\w_{g,i})})^{-1} (\bb{b}^{(\w_{g,i})}) 
                 - \frac{1}{2} p \ln (2\pi)  -\frac{1}{2} \ln\det \hat{\bb\Sigma}_{reg,l}^{(\w_{g,i})},
            \end{align*}
            where $\bb b$ denotes the vector $\bb b = (b, \ldots, b) \in \mathbb{R}^p$, with $b$ corresponding to the upper bound of the within-cluster distances in the ideal scenario, and the last inequality follows from 
            \begin{align}
                \max_{1 \leq l \leq s} \max \{ ||\bb{x}_{i', m} - \bb{x}_{i, m} ||_2: \bb{x}_{i', m}, \bb{x}_{i, m} \in A_m^l\} < b \quad \forall m \in \mathbb{N}, \label{eq:withincluster}
            \end{align}
            where $||.||_2$ denotes the Euclidean norm.  
            Since all terms on the right hand side are bounded, the objective function is bounded from above. For the lower bound, 
            \begin{align*}
                \ln \Bigg(&\sum_{k = 1}^N \hat{\pi}_{g, k}  \varphi\left(\x_{g,i}^{(\w_{g,i})}; \hat{\bb\mu}_k^{(\w_{g,i})}, \hat{\bb\Sigma}_{reg,k}^{(\w_{g,i})}\right)\Bigg)\\
                & \leq  \max_k \ln\left(\varphi\left(\x_{g,i}^{(\w_{g,i})}; \hat{\bb\mu}_k^{(\w_{g,i})}, \hat{\bb\Sigma}_{reg,k}^{(\w_{g,i})}\right)\right)\\
                & \leq \max_k( \underbrace{-\frac{1}{2} (\x_{g,i}^{(\w_{g,i})} - \hat{\bb\mu}_k^{(\w_{g,i})})' (\hat{\bb\Sigma}_{reg,k}^{(\w_{g,i})})^{-1} (\x_{g,i}^{(\w_{g,i})} - \hat{\bb\mu}_k^{(\w_{g,i})}))}_{\leq 0} \\ &\quad \quad \underbrace{-\frac{1}{2} \sum_j w_{g,ij} \ln (2\pi)}_{\leq 0} + \max_k(-\frac{1}{2}  \ln\det \hat{\bb\Sigma}_{reg,k}^{(\w_{g,i})}) \\
                & \leq -\frac{1}{2}  \ln \min_k\det \hat{\bb\Sigma}_{reg,k}^{(\w_{g,i})}.
            \end{align*}
            Since the covariance estimates are finite, the objective function is finite for any feasible $\bb{W}$.
    
            Second, assume $\alpha = 1$.  Set $\hat{\pi}_{k,k} = 1$, $\hat{\pi}_{k,l} = 0$ for all $k \neq l$. 
            All observations from a group $g$ originate from cluster $g$, $\bb Z^g = A^g$. 
            Thus, for any $\x_{g,i}$ it holds that
            \begin{align*}
                \ln \Bigg(\sum_{k = 1}^N & \hat{\pi}_{g, k}  \varphi\left(\x_{g,i}^{(\w_{g,i})}; \hat{\bb\mu}_k^{(\w_{g,i})}, \hat{\bb\Sigma}_{reg,k}^{(\w_{g,i})}\right)\Bigg)\\
                & = -\frac{1}{2}(\x_{g,i}^{(\w_{g,i})} - \hat{\bb\mu}_g^{(\w_{g,i})})' (\hat{\bb\Sigma}_{reg,g}^{(\w_{g,i})})^{-1} (\x_{g,i}^{(\w_{g,i})} - \hat{\bb\mu}_g^{(\w_{g,i})}) \\
                & \quad \quad  -\frac{1}{2} \sum_j w_{g,ij} \ln (2\pi) -\frac{1}{2} \ln\det \hat{\bb\Sigma}_{reg,g}^{(\w_{g,i})} \\
                 & \geq -\frac{1}{2}\left(  (\bb{b}^{(\w_{g,i})})' (\hat{\bb\Sigma}_{reg,g}^{(\w_{g,i})})^{-1} (\bb{b}^{(\w_{g,i})}) + p \ln (2\pi)  + \ln\det \hat{\bb\Sigma}_{reg,g}^{(\w_{g,i})}\right) 
            \end{align*}
            and the objective function is bounded from above. For the lower bound, it follows
            \begin{align*}
                \ln \Bigg(\sum_{k = 1}^N \hat{\pi}_{g, k}  &\varphi\left(\x_{g,i}^{(\w_{g,i})}; \hat{\bb\mu}_k^{(\w_{g,i})}, \hat{\bb\Sigma}_{reg,k}^{(\w_{g,i})}\right)\Bigg)\\
                & = \underbrace{-\frac{1}{2} (\x_{g,i}^{(\w_{g,i})} - \hat{\bb\mu}_g^{(\w_{g,i})})' (\hat{\bb\Sigma}_{reg,g}^{(\w_{g,i})})^{-1} (\x_{g,i}^{(\w_{g,i})} - \hat{\bb\mu}_g^{(\w_{g,i})})}_{\leq 0} \\ &\quad \quad \underbrace{-\frac{1}{2} \sum_j w_{g,ij} \ln (2\pi)}_{\leq 0} -\frac{1}{2}  \ln\det \hat{\bb\Sigma}_{reg,g}^{(\w_{g,i})}\\
                & \leq -\frac{1}{2}  \ln\det \hat{\bb\Sigma}_{reg,g}^{(\w_{g,i})}.
            \end{align*}
            Thus, the objective function is bounded for any feasible $\bb{W}$.

   \textbf{Proof of part b1.}   
        Assume that under the given estimates the objective function is bounded. By construction, the estimated covariances $\hat{\bb\Sigma}_{reg,k}$ are regular and thus, the lowest eigenvalues $\lambda_p(\hat{\bb\Sigma}_{reg,k}) \geq \tilde{b}_k(\rho_k, \bb{T}_k)>0 $ are bounded away from zero. According to the proof of Proposition 2b) from \citet{Raymaekers2023} it holds for all $k$ and any feasible $\hat{\w}$ that
        \begin{align*}
            \ln \det \hat{\bb\Sigma}_{reg,k}^{(\hat{\w})} \geq \ln \max_{j = 1, \ldots, p} \hat{\Sigma}_{reg,k, jj}^{(\hat{\w})}  + (p-1)\ln\tilde{b}_k(\rho_k, \bb{T}_k), 
        \end{align*}
        where $\tilde{b}_k(\rho_k, \bb{T}_k)$ is a constant depending only on $\rho_k$ and $\bb{T}_k$.
 
        From part~a.~we know that for all $\x_{g,i}$ from group $g$ it holds that
        \begin{align}
            \ln \Bigg(&\sum_{k = 1}^N \hat{\pi}_{g, k}  \varphi\left(\x_{g,i}^{(\hat{\w}_{g,i})}; \hat{\bb\mu}_k^{(\hat{\w}_{g,i})}, \hat{\bb\Sigma}_{reg,k}^{(\hat{\w}_{g,i})}\right)\Bigg) \notag\\
            & \leq -\frac{1}{2} \min_k \Bigg( (\x_{g,i}^{(\hat{\w}_{g,i})} - \hat{\bb\mu}_k^{(\hat{\w}_{g,i})})' (\hat{\bb\Sigma}_{reg,k}^{(\hat{\w}_{g,i})})^{-1} (\x_{g,i}^{(\hat{\w}_{g,i})} - \hat{\bb\mu}_k^{(\hat{\w}_{g,i})})
            + \ln\det \hat{\bb\Sigma}_{reg,k}^{(\hat{\w}_{g,i})}\Bigg) \notag\\
            &\leq  -\frac{1}{2} \min_k \Bigg( (\x_{g,i}^{(\hat{\w}_{g,i})} - \hat{\bb\mu}_k^{(\hat{\w}_{g,i})})' (\hat{\bb\Sigma}_{reg,k}^{(\hat{\w}_{g,i})})^{-1} (\x_{g,i}^{(\hat{\w}_{g,i})} - \hat{\bb\mu}_k^{(\hat{\w}_{g,i})}) \notag\\
            & \quad \quad + \ln \max_{j = 1, \ldots, p} \hat{\Sigma}_{reg,k, jj}^{(\hat{\w}_{g,i})}  + (p-1)\ln\tilde{b}_k(\rho_k, \bb{T}_k) \Bigg)\notag\\
            &\leq  -\frac{1}{2}\min_k (p-1)\ln\tilde{b}_k(\rho_k, \bb{T}_k) -\frac{1}{2} \min_k \Bigg( (\x_{g,i}^{(\hat{\w}_{g,i})} - \hat{\bb\mu}_k^{(\hat{\w}_{g,i})})'(\hat{\bb\Sigma}_{reg,k}^{(\hat{\w}_{g,i})})^{-1} (\x_{g,i}^{(\hat{\w}_{g,i})} - \hat{\bb\mu}_k^{(\hat{\w}_{g,i})})  \notag\\
            &  \quad\quad 
            + \ln \max_{j = 1, \ldots, p} \hat{\Sigma}_{reg,k, jj}^{(\hat{\w}_{g,i})} \Bigg). \label{eq:proof_argmin}
        \end{align}

        Let $j^*(l) = \max_{j = 1, \ldots, p} \hat{\Sigma}_{reg,l,jj}$ for the distribution where $\lambda_1(\hat{\bb{\Sigma}}_{reg, l}) \rightarrow \infty$.
        For each group $g$ there exists at least one observation $\x_{g,i^*(g)}$ from cluster $g$ for which variable $j^*(l)$ is observed, $w_{g,i^*(g)j^*(l)} = 1$. For these observations, we have
        \begin{align*}
            (\x_{g,i^*(g)}^{(\hat{\w}_{g,i^*(g)})} - \hat{\bb\mu}_l^{(\hat{\w}_{g,i^*(g)})})' (\hat{\bb\Sigma}_{reg,l}^{(\hat{\w}_{g,i^*(g)})})^{-1}(\x_{g,i^*(g)}^{(\hat{\w}_{g,i^*(g)})} - \hat{\bb\mu}_l^{(\hat{\w}_{g,i^*(g)})}) \quad\quad\\ 
            + \ln \max_{j = 1, \ldots, p} \hat{\Sigma}_{reg,l, jj}^{(\hat{\w}_{g,i^*(g)})} &\geq \ln \max_{j = 1, \ldots, p} \hat{\Sigma}_{reg,l, jj}\\
            & = \ln \max_{j, j' = 1, \ldots, p} |\hat{\Sigma}_{reg,l, jj'} | \\
            & \geq \ln \frac{\lambda_1(\hat{\bb{\Sigma}}_{reg, l})}{p} \rightarrow \infty.
        \end{align*}
        Thus, for all $\x_{g,i^*(g)}, g = 1, \ldots, N$, the argument $l$ cannot be the minimizer. 
        
        Without loss of generality, assume that all other covariance matrices are bounded, $\lambda_1(\hat{\bb{\Sigma}}_{reg, k}) < \infty$ if $ k \neq l$. 
        Within the ideal scenario,
        it holds that $|x_{g,i^*(g)j^*(l)} -x_{h,i^*(h)j^*(l)}| \rightarrow \infty$ if $g \neq h$. Also,
        \begin{align*}
            (\x_{g,i^*(g)}^{(\hat{\w}_{g,i^*(g)})} - \hat{\bb\mu}_k^{(\hat{\w}_{g,i^*(g)})})' (\hat{\bb\Sigma}_{reg,k}^{(\hat{\w}_{g,i^*(g)})})^{-1} (\x_{g,i^*(g)}^{(\hat{\w}_{g,i^*(g)})} - \hat{\bb\mu}_k^{(\hat{\w}_{g,i^*(g)})}) \\
            \geq (x_{g,i^*(g)j^*(l)}- \hat{\mu}_{k, j^*(l)})^2 \lambda_p\left((\hat{\bb\Sigma}_{reg,k}^{(\hat{\w}_{g,i^*(g)})})^{-1}\right).
        \end{align*}
        The smallest eigenvalue going to zero, $\lambda_p((\hat{\bb\Sigma}_{reg,k}^{(\hat{\w}_{g,i^*(g)})})^{-1}) \rightarrow 0$ 
        implies $\lambda_1(\hat{\bb\Sigma}_{reg,k}^{(\hat{\w}_{g,i^*(g)})}) \rightarrow \infty$ as well as $\lambda_1(\hat{\bb\Sigma}_{reg,k}) \rightarrow \infty$, which contradicts that the other covariances are bounded in the first eigenvalue. Thus, $\lambda_p\left((\hat{\bb\Sigma}_{reg,k}^{(\hat{\w}_{g,i^*(g)})})^{-1}\right)$ is bounded away from zero.  

        Since all observations are increasingly far away, there exists at least one $\x_{g',i^*(g')}$ such that $(x_{g',i^*(g')j^*(l)}- \hat{\mu}_{k, j^*(l)})^2 \rightarrow \infty$ for all $k = 1, \ldots, N, k \neq l$ and for which the minimum from Equation~\eqref{eq:proof_argmin} goes to infinity. Moreover, all parts are bounded from above, 
        \begin{align*}
            \ln \Bigg(\sum_{k = 1}^N \hat{\pi}_{g, k}  \varphi\left(\x_{g,i}^{(\hat{\w}_{g,i})}; \hat{\bb\mu}_k^{(\hat{\w}_{g,i})}, \hat{\bb\Sigma}_{reg,k}^{(\hat{\w}_{g,i})}\right)\Bigg) \leq -\frac{p}{2} \min_k \ln\tilde{b}_k(\rho_k, \bb{T}_k).
        \end{align*}
        Thus, the objective function has to explode.

 \textbf{Proof of part b2.}
        Assume that the objective function of the estimators $\hat{\bb{\pi}}$, $\hat{\bb{\mu}}$, $ \hat{\bb{\Sigma}}$, $\hat{\bb{W}}$ is finite. Then $ \hat{\bb{\Sigma}}_{reg, k}$ are regular and not exploding due to part b1.
        For all groups $g$ there exists at least one observation $\x_{g,i^*(g)} \in (A^g \cup B^g) \cap \bb{Z}^g$ such that $\hat{w}_{g,i^*(g)j^*} = 1$. Let $C_1 =\min_{k, \hat{\w},j}  \hat{\bb\Sigma}_{reg,k, jj}^{(\hat{\w})} > 0$ and $C_2 =\min_{k, \hat{\w}, j} (\hat{\bb\Sigma}_{reg,k}^{(\hat{\w})})^{-1}_{jj} > 0$ (see part b1), then it holds
        \begin{align*}
            \ln \Bigg(\sum_{k = 1}^N & \hat{\pi}_{g, k}   \varphi\left(\x_{g,i^*(g)}^{(\hat{\w}_{g,i^*(g)})}; \hat{\bb\mu}_k^{(\hat{\w}_{g,i^*(g)})}, \hat{\bb\Sigma}_{reg,k}^{(\hat{\w}_{g,i^*(g)})}\right)\Bigg) \notag\\
            & \leq   -\frac{1}{2}\min_k (p-1)\ln\tilde{b}_k(\rho_k, \bb{T}_k)  -\frac{1}{2} \min_k \ln \max_{j = 1, \ldots, p} \hat{\bb\Sigma}_{reg,k, jj}^{(\hat{\w}_{g,i^*(g)})} \\
            & \quad \quad -\frac{1}{2} \min_k  \Bigg((\x_{g,i^*(g)}^{(\hat{\w}_{g,i^*(g)})} - \hat{\bb\mu}_k^{(\hat{\w}_{g,i^*(g)})})' (\hat{\bb\Sigma}_{reg,k}^{(\hat{\w}_{g,i^*(g)})})^{-1} 
            (\x_{g,i^*(g)}^{(\hat{\w}_{g,i^*(g)})} - \hat{\bb\mu}_k^{(\hat{\w}_{g,i^*(g)})})\Bigg)  \\
            & \leq  -\frac{1}{2}\min_k (p-1)\ln\tilde{b}_k(\rho_k, \bb{T}_k) -\frac{1}{2} \ln C_1 \notag \\
            & \quad \quad -\frac{1}{2} C_2 \min_k \Bigg( (x_{g,i^*(g)j^*}- \hat{\mu}_{k, j^*})^2\Bigg).
        \end{align*}
        There are $N$ many observations observed in $j^*$ that move increasingly far away from each other in variable $j^*$. Since there exists $l',l$ such that to $|\hat{{\mu}}_{l',j^*} - \hat{{\mu}}_{l,j^*}| < \tilde{b}$ there are only $N-1$ location estimates that move infinitely far away from each other. It follows that $\max_g\min_k (x_{g,i^*(g)j^*}- \hat{\mu}_{k, j^*})^2 \rightarrow \infty$ and thus, there is one term in the objective function that explodes, while the others are bounded (see part b1).

   \textbf{Proof of part b3.} 
        From the proof of part b2
        together with 
        \begin{align*}
            (\x_{g,i}^{(\w_{g,i})} - \hat{\bb\mu}_k^{(\w_{g,i})})' 
            (\hat{\bb\Sigma}_{reg,k}^{(\w_{g,i})})^{-1} (\x_{g,i}^{(\w_{g,i})} - \hat{\bb\mu}_k^{(\w_{g,i})}) 
            \geq \lambda_p((\hat{\bb\Sigma}_{reg,k}^{(\w_{g,i})})^{-1}) \sum_{\substack{j: w_{g,ij} = 1}} | x_{g,ij} - \hat{\mu}_{k,j}|^2
        \end{align*}
        for all $k = 1, \ldots, N$, we know that
        \begin{align*}
            \ln \Bigg(\sum_{k = 1}^N & \hat{\pi}_{g, k}   \varphi\left(\x_{g,i}^{(\w_{g,i})}; \hat{\bb\mu}_k^{(\w_{g,i})}, \hat{\bb\Sigma}_{reg,k}^{(\w_{g,i})}\right)\Bigg) \notag\\
            & \leq\underbrace{ -\frac{1}{2}\min_k (p-1)\ln\tilde{b}_k(\rho_k, \bb{T}_k)  -\frac{1}{2} \min_k \ln \max_{j = 1, \ldots, p} \hat{\bb\Sigma}_{reg,k, jj}^{(\w_{g,i})}}_{\text{bounded by part b1 and finite objective function}} \\
            & \quad \quad -\frac{1}{2} \min_k  \Bigg( \lambda_p((\hat{\bb\Sigma}_{reg,k}^{(\w_{g,i})})^{-1}) \sum_{j : w_{g,ij} = 1} | x_{g,ij} - \hat{\mu}_{k,j}|^2\Bigg) \\
            & \leq\underbrace{-\frac{1}{2}\min_k (p-1)\ln\tilde{b}_k(\rho_k, \bb{T}_k)  -\frac{1}{2} \min_k \ln \max_{j = 1, \ldots, p} \hat{\bb\Sigma}_{reg,k, jj}^{(\w_{g,i})}-\frac{1}{2} \min_k \lambda_p((\hat{\bb\Sigma}_{reg,k}^{(\w_{g,i})})^{-1})}_{\text{bounded by part b1 and finite objective function}} \\
            & \quad \quad -\frac{1}{2} \min_k  \Bigg(  \sum_{\substack{j : w_{g,ij} = 1}} | x_{g,ij} - \hat{\mu}_{k,j}|^2\Bigg).
        \end{align*}
        Thus, for all $\x_{g,i} \in (A^g \cup B^g) \cap \bb{Z}^g$ the term        
        $$\min_k ( \sum_{\substack{j : w_{g,ij} = 1}} | x_{g,ij} - \hat{\mu}_{k,j}|^2 )$$ needs to stay bounded, otherwise the objective function would explode. It follows that for each $\x_{g,i}$ there exists a $k^*$ such that $||\x_{g,i}^{(\w_{g,i})} - \hat{\bb\mu}_{k^*}(\mathcal{Z}_m)^{(\w_{g,i})} ||  < \infty$. Due to Corollary~\ref{theorem:corollary} part b2 
        and the finite objective function, the corresponding $k^*$ is unique.
        
\end{proof}

\subsection{Proof of Theorem \ref{theorem:breakdownpoint}} \label{app:subsec:theorem}

\begin{proof}
    We first discuss the parts for which the proofs are more compact, then the parts with lengthier proofs.  
  
    \textbf{Proof of part b.}
    Clear, since the lowest eigenvalues are always bound away from zero \citep[see also proof of Theorem 2c in][]{Puchhammer2024}.
    
    \textbf{Proof of part e.}
    Since the constraint
    $\pi_{g,g} \geq \alpha \geq 0.5$ for all $g = 1, \ldots, N$ restricts the estimates $\hat{\bb{\pi}}(\mathcal{Z}_m)$ such that $ \hat{{\pi}}(\mathcal{Z}_m)_{g, g} \geq \alpha > 0.5$ for all $g$, the weight of each cluster $k$ is $\frac{1}{N}\sum_{g = 1}^N\hat{{\pi}}(\mathcal{Z}_m)_{g, k} \geq \frac{\alpha}{N} > 0$. Thus, all clusters have non-zero weight.

    \textbf{Proof of part a.}
    Assume, that there are up to $n_g-h_g$ cellwise outliers. By flagging all the cellwise outliers with $\hat{\bb{W}}$, there exists a solution with finite objective function according to Corollary \ref{theorem:corollary} part a and the optimal estimates have a finite objective function value. Denote the optimal estimates with $\hat{\bb{\pi}}  (\mathcal{Z})$, $\hat{\bb{\mu}}(\mathcal{Z})$, $ \hat{\bb{\Sigma}}(\mathcal{Z}) $, $\hat{\bb{W}}(\mathcal{Z})$ for the contaminated data and $\hat{\bb{\pi}}  (\mathcal{X})$, $\hat{\bb{\mu}}(\mathcal{X})$, $ \hat{\bb{\Sigma}}(\mathcal{X}) $, $\hat{\bb{W}}(\mathcal{X})$ for the uncontaminated data. 
    
    Based on the constraint for $h_g$, for each group $g$ and any pair of variables $j_1$ and $j_2$ there exist at least two uncontaminated observation $\x_{g,i},\x_{g,i'} \in A^g \cap \bb{Z}^g$ such that $\hat{\bb{w}}_{g,ij_1}(\mathcal{X}) = \hat{\bb{w}}_{g,ij_2}(\mathcal{X}) = 1$ and $\hat{\bb{w}}_{g,i'j_1}(\mathcal{Z}) = \hat{\bb{w}}_{g,i'j_2}(\mathcal{Z}) = 1$, respectively. Since the objective function is finite, it follows from Corollary \ref{theorem:corollary} 
    part b3 that for each $\x_{g,i}$ and $\x_{g,i'}$ there exists a unique $k_*$ and $k_*'$ such that $||{\x}_{g,i}^{(\hat{\bb{w}}_{g,i}(\mathcal{X}))} - \hat{\bb\mu}_{k_*}^{(\hat{\bb{w}}_{g,i}(\mathcal{X}))}(\mathcal{X}) || < \infty$ and $||{\x}_{g,i'}^{(\hat{\bb{w}}_{g,i'}(\mathcal{Z}))} - \hat{\bb\mu}_{k_*'}^{(\hat{\bb{w}}_{g,i'}(\mathcal{Z}))}(\mathcal{Z}) || < \infty$, respectively. 
    
    We show, that $k_*$  is the same over all pairs of variables. Let $j_1 = 1$ and $j_2 = 2$ and $\x_{g,i_1}$ be the corresponding observation where both variables are observed. There exists a unique $k_{1*}$ such that $||x_{g,i_11} - \hat{\mu}_{k_{1*}, 1}(\mathcal{X}) || < \infty$ and $||x_{g,i_1 2} - \hat{\mu}_{k_{1*}, 2}(\mathcal{X}) || < \infty$. For $j_1 = 2$ and $j_2 = 3$ there exists an observation $\x_{g,i_2}$ and a unique $k_{2*}$ such that $||x_{g,i_22} - \hat{\mu}_{k_{2*}, 2}(\mathcal{X}) || < \infty$ and $||x_{g,i_23} - \hat{\mu}_{k_{2*}, 3}(\mathcal{X}) || < \infty$. Since $|x_{g,i_12}- x_{g,i_22}| < \infty$
    in the ideal scenario,
    it follows from Corollary~\ref{theorem:corollary} part b2 
    that $k_{1*} = k_{2*}$. By induction it follows, that $k_*$ is the same for all variables.
    The same applies to $k_*'$.
    
    Since the distance between observations from $A_g$ are bounded in 
    the ideal scenario, also the distance between $\hat{\bb\mu}_{k_*}(\mathcal{X})$ and $\hat{\bb\mu}_{k_*'}(\mathcal{Z})$ is bounded in each variable and thus,
        $|| \hat{\bb\mu}_{k_*}(\mathcal{X}) - \hat{\bb\mu}_{k_*'}(\mathcal{Z}) ||^2 
        = \sum_{j = 1}^p   |\hat{\mu}_{{k_*},j}(\mathcal{X}) - \hat{\mu}_{{k_*'},j}(\mathcal{Z})|^2 < \infty$
    holds for the choice of $k_*$ and $k_*'$ based on a given group.
    
    Based on increasing distance of clusters to each other in the ideal scenario and
    Corollary ~\ref{theorem:corollary}
    part b2 and b3, for any given $k$ there exists exactly one group $g(k)$ such that the distance of $\hat{\bb\mu}_{k}(\mathcal{X})$ to observations from $A_{g(k)} \cap \Z_{g(k)}$ is bounded. 
    Following from above, for all $\hat{\bb\mu}_{k}(\mathcal{X})$ there exists $\hat{\bb\mu}_{k'(g(k))}(\mathcal{X})$ with
      $  || \hat{\bb\mu}_{k}(\mathcal{X}) - \hat{\bb\mu}_{k'(g(k))}(\mathcal{Z}) || < \infty$
    and no breakdown occurs.
    
     \textbf{Proof of part c.}
    From Corollary~\ref{theorem:corollary} part a, we know for uncontaminated 
    data $\mathcal{X}_m$ that the objective function is finite for the minimizers, and from Corollary~\ref{theorem:corollary} part b1, we know that the covariance matrix estimates are not exploding. Thus, a breakdown occurs only when there exists an $l$ such that $\lambda_1(\hat{\bb{\Sigma}}_{reg, l}(\mathcal{Z}_m)) \rightarrow \infty $. 
    
    Assume that for each group $g$ only up to $n_g-h_g$ cells per column are contaminated and outlying in the idealized scenario. It is possible to set $\hat{\bb W}$ such that $w_{\bb{y},j} = 0$ for all cells of added outliers $\bb{y}$ exactly when $w(\bb{y})_j = 0$. Thus, there exists a copy of an uncontaminated ideal scenario $\tilde{\mathcal{X}}_m$, that has the same values if cells are observed as indicated by $\hat{\bb W}$ and non-outlying values if $w_{\bb{y},j} = 0$. From Corollary~\ref{theorem:corollary} part a, for the given $\hat{\bb W}$ it follows that there exist candidate estimates with finite objective function for $\tilde{\mathcal{X}}_m$ and the value of the objective function on $\mathcal{X}_m \cup \mathcal{Y}_m$ is the same (and finite). From Corollary~\ref{theorem:corollary} part b1, it follows that if a covariance matrix explodes, the objective function explodes as well and the estimates cannot be minimizers of the objective function because there exist candidate estimates with a lower objective function. Thus, the breakdown point is at least $\min_g \{(n_g-h_g+1)/n_g\}$.

    \textbf{Proof of part d.}
    We construct a counter example that shows that the covariance needs to explode if the location estimator
    did not break down in 
    the idealized scenario.

    Given an uncontaminated sample $\mathcal{X}$ and one variable $j^*$, we assume that all cells from variable $j^*$ of the uncontaminated data are positive. The uncontaminated data $\mathcal{X}$ is partitioned into groups $\Z^1, \ldots, \Z^N$ and only one group $g'$ is contaminated with $n_{g'} - h_{g'} + 1$ many cellwise outliers $\mathcal{Y}$, outlying only in variable $j^*$ with negative values. Thus, for any $\bb W_{g'}$ there is always at least one outlying cell in variable $j^*$, that is observed. The data used in the contaminated case is then $\mathcal{Z} = \bigcup_{g = 1}^N \Z^g$. For an estimator $\hat{\bb W} (\mathcal{Z})$ let $\tilde{\bb y}$ be an outlier for which variable $j^*$ is observed, $w(\tilde{\bb y})_{j^*} = 0$ and $\hat{w}_{\tilde{\bb y},j^*} = 1$. 

    Let $\hat{t}_k(\bb z)$ denote the probability of an observation $\bb z\in \Z_g$ that it belongs to distribution $k$ given the estimates $\hat{\bb{\pi}}(\mathcal{Z})$, $\hat{\bb{\mu}}(\mathcal{Z})$, $ \hat{\bb{\Sigma}}(\mathcal{Z}) $ and $\hat{\bb{W}}(\mathcal{Z})$,
    \begin{align*}
            \hat{t}_k(\bb z) = \frac{\hat{\pi}_{g, k} \varphi\left(\bb z^{(\hat{\bb w}_{\bb z})}; \hat{\bb{\mu}}_{k}^{(\hat{\bb w}_{\bb z})}, \hat{\bb{\Sigma}}_{reg,k}^{(\hat{\bb w}_{\bb z})}\right) }{\sum_{l = 1}^N \hat{\pi}_{g, l} \varphi\left(\bb z^{(\hat{\bb w}_{\bb z})}; \hat{\bb{\mu}}_l^{(\hat{\bb w}_{\bb z})}, \hat{\bb{\Sigma}}_{reg,l}^{(\hat{\bb w}_{\bb z})}\right)}.
    \end{align*}
    Note that due to the regularity of the covariance estimates the density goes to zero, $\varphi\left(\bb z^{(\hat{\bb w}_{\bb z})}; \hat{\bb{\mu}}_{k}^{(\hat{\bb w}_{\bb z})}, \hat{\bb{\Sigma}}_{reg,k}^{(\hat{\bb w}_{\bb z})}\right) \rightarrow 0 $, if $||\bb z^{(\hat{\bb w}_{\bb z})} -  \hat{\bb{\mu}}_{k}^{(\hat{\bb w}_{\bb z})}||_2 \rightarrow \infty$ and thus $ \hat{t}_k(\bb z) \rightarrow 0$. 
    Since there are $N$ many possible distributions, for $\tilde{\bb y}$ there exists a distribution $k^*$ with $\hat{t}_{k^*}(\tilde{\bb y}) \geq \frac{1}{N} > 0$. 

    Upon convergence of the EM-algorithm the location estimate of the $j^*$-th variable of distribution $k^*$ is 
    \begin{align*}
        \hat{{\mu}}_{k^*j^*}(\mathcal{Z}) & =\frac{1}{\bar{t}_{k^*}}  \sum_{g = 1}^N \sum_{\bb{z} \in \Z_g} \hat{t}_{k^*}(\bb z) \hat{ z}_{j^*},
    \end{align*}
    with $\bar{t}_{k^*} = \sum_{g = 1}^N \sum_{\bb{z} \in \Z_g} \hat{t}_{k^*}(\bb z)$ and $\hat{ z}_{j^*}$ being the imputed value of $\bb z$ for variable $j^*$. 
    For $\hat{w}_{\bb z,j^*} = 1$ it is equal to $z_{j^*}$  and for $\hat{w}_{\bb z,j^*} = 0$ it is equal to 
    \begin{align*}
        \hat{\mu}_{k^* j^*} + \hat{\bb{\Sigma}}_{reg,k^*}^{(j^*| \hat{\bb{w}}_{\bb z})} \left(\hat{\bb{\Sigma}}_{reg,k^*}^{(\hat{\bb{w}}_{\bb z}| \hat{\bb{w}}_{\bb z})} \right)^{-1} \left( \bb z^{(\hat{\bb{w}}_{\bb z})} - \hat{\bb{\mu}}_{k^*}^{(\hat{\bb{w}}_{\bb z})}\right),
    \end{align*}
    where $\hat{\bb{\Sigma}}_{reg,k^*}^{(j^*| \hat{\bb{w}}_{\bb z})}$ indicates the submatrix $\hat{\bb{\Sigma}}_{reg,k^*}$ consisting of the $j^*$-th row and the observed variables of $\bb z$ as columns.

    Denote the set of observations of $\mathcal{Z}$ where variable $j^*$ is observed as $\mathcal{O}_{j^*} = \{ \bb z \in \mathcal{Z}: \hat{{w}}_{\bb z, j^*} = 1\}$, and let $\mathcal{O}_{j^*}^c$ denote its complement. We can then separate the sum term into 
    \begin{align*}
        \hat{{\mu}}_{k^*j^*}(\mathcal{Z}) 
        & =\frac{1}{\bar{t}_{k^*}}  \sum_{g = 1}^N \sum_{\bb{z} \in \Z_g} \hat{t}_{k^*}(\bb z) \hat{ z}_{j^*} \\
        & =\frac{1}{\bar{t}_{k^*}}  \sum_{g \neq g'} \sum_{\bb{x} \in \Z_g} \hat{t}_{k^*}(\bb x) \hat{ x}_{j^*} + \frac{1}{\bar{t}_{k^*}} \sum_{\bb{z} \in \Z_{g'}} \hat{t}_{k^*}(\bb z) \hat{z}_{j^*} \\
         & =\frac{1}{\bar{t}_{k^*}}  \sum_{g=1}^N \sum_{\bb{x} \in \Z_g \cap \mathcal{X}} \hat{t}_{k^*}(\bb x) \hat{ x}_{j^*} + \frac{1}{\bar{t}_{k^*}} \sum_{\bb{y} \in \Z_{g'} \cap \mathcal{Y}} \hat{t}_{k^*}(\bb y) \hat{y}_{j^*} 
     \end{align*}
     and further into
     \begin{align*}
     \hat{{\mu}}_{k^*j^*}(\mathcal{Z})
         & =\frac{1}{\bar{t}_{k^*}}  \sum_{g=1 }^N \sum_{\bb{x} \in \Z_g \cap \mathcal{X}\cap \mathcal{O}_{j^*}} \hat{t}_{k^*}(\bb x) \hat{ x}_{j^*} + \frac{1}{\bar{t}_{k^*}}  \sum_{g=1 }^N \sum_{\bb{x} \in \Z_g \cap \mathcal{X}\cap \mathcal{O}_{j^*}^c} \hat{t}_{k^*}(\bb x) \hat{ x}_{j^*} \\
         &\quad \quad  + \frac{1}{\bar{t}_{k^*}} \sum_{\bb{y} \in \Z_{g'} \cap \mathcal{Y}\cap \mathcal{O}_{j^*}} \hat{t}_{k^*}(\bb y) \hat{ y}_{j^*} + \frac{1}{\bar{t}_{k^*}} \sum_{\bb{y} \in \Z_{g'} \cap \mathcal{Y}\cap \mathcal{O}_{j^*}^c} \hat{t}_{k^*}(\bb y) \hat{y}_{j^*}.
    \end{align*}
    Together with the expressions for the imputed values $\hat{z}_{j^*}$, we get
    \begin{align*}
         \hat{{\mu}}_{k^*j^*}(\mathcal{Z}) & = \frac{1}{\bar{t}_{k^*}}  \sum_{g= 1}^N \sum_{\bb{x} \in \Z_g \cap \mathcal{X}\cap \mathcal{O}_{j^*}} \hat{t}_{k^*}(\bb x) x_{j^*} + \frac{1}{\bar{t}_{k^*}} \sum_{\bb{y} \in \Z_{g'} \cap \mathcal{Y}\cap \mathcal{O}_{j^*}} \hat{t}_{k^*}(\bb y) y_{j^*}\\ 
         & \quad \quad + \frac{1}{\bar{t}_{k^*}}  \sum_{g= 1}^N \sum_{\bb{x} \in \Z_g \cap \mathcal{X}\cap \mathcal{O}_{j^*}^c} \hat{t}_{k^*}(\bb x)\Big[ \hat{\mu}_{k^* j^*}(\mathcal{Z}) + \hat{\bb{\Sigma}}_{reg,k^*}^{(j^*| \hat{\bb{w}}_{\bb x})} \left(\hat{\bb{\Sigma}}_{reg,k^*}^{(\hat{\bb{w}}_{\bb x}| \hat{\bb{w}}_{\bb x})} \right)^{-1} \\
         & \quad \quad \times\left( \bb x^{(\hat{\bb{w}}_{\bb x})} - \hat{\bb{\mu}}_{k^*}(\mathcal{Z})^{(\hat{\bb{w}}_{\bb x})}\right) \Big] + \frac{1}{\bar{t}_{k^*}} \sum_{\bb{y} \in \Z_{g'} \cap \mathcal{Y}\cap \mathcal{O}_{j^*}^c} \hat{t}_{k^*}(\bb y) \Big[ \hat{\mu}_{k^* j^*}(\mathcal{Z}) \\
         &\quad \quad + \hat{\bb{\Sigma}}_{reg,k^*}^{(j^*| \hat{\bb{w}}_{\bb y})} \left(\hat{\bb{\Sigma}}_{reg,k^*}^{(\hat{\bb{w}}_{\bb y}| \hat{\bb{w}}_{\bb y})} \right)^{-1} \left( \bb y^{(\hat{\bb{w}}_{\bb y})} - \hat{\bb{\mu}}_{k^*}(\mathcal{Z})^{(\hat{\bb{w}}_{\bb y})}\right) \Big].
    \end{align*}
    
    Subtracting the estimated location on the uncontaminated sample $\hat{{\mu}}_{k^*j^*}(\mathcal{X})$ and using that the location estimator
    did not break down, we further get
    \begin{align*}
        &\underbrace{\hat{{\mu}}_{k^*j^*}(\mathcal{Z}) -\hat{{\mu}}_{k^*j^*}(\mathcal{X})}_{\text{bounded}} =\\
        &\quad \quad =\frac{1}{\bar{t}_{k^*}}  \sum_{g= 1}^N \sum_{\bb{x} \in \Z_g \cap \mathcal{X}\cap \mathcal{O}_{j^*}} \hat{t}_{k^*}(\bb x) \underbrace{(x_{j^*}-\hat{{\mu}}_{k^*j^*}(\mathcal{X}))}_{*} \\
        &\quad \quad  + \frac{1}{\bar{t}_{k^*}} \sum_{\bb{y} \in \Z_{g'} \cap \mathcal{Y}\cap \mathcal{O}_{j^*}} \hat{t}_{k^*}(\bb y) \underbrace{(y_{j^*} -\hat{{\mu}}_{k^*j^*}(\mathcal{X}))}_{\rightarrow -\infty}\\ 
         & \quad \quad + \frac{1}{\bar{t}_{k^*}}  \sum_{g= 1}^N \sum_{\bb{x} \in \Z_g \cap \mathcal{X}\cap \mathcal{O}_{j^*}^c} \hat{t}_{k^*}(\bb x)\\
         & \quad \quad \quad \quad \quad \left[ \underbrace{\hat{\mu}_{k^* j^*}(\mathcal{Z})-\hat{{\mu}}_{k^*j^*}(\mathcal{X}) }_{\text{bounded}}+ \hat{\bb{\Sigma}}_{reg,k^*}^{(j^*| \hat{\bb{w}}_{\bb x})} \left(\hat{\bb{\Sigma}}_{reg,k^*}^{(\hat{\bb{w}}_{\bb x}| \hat{\bb{w}}_{\bb x})} \right)^{-1} \underbrace{\left( \bb x^{(\hat{\bb{w}}_{\bb x})} - \hat{\bb{\mu}}_{k^*}(\mathcal{Z})^{(\hat{\bb{w}}_{\bb x})}\right)}_{*} \right] \\
         & \quad \quad + \frac{1}{\bar{t}_{k^*}} \sum_{\bb{y} \in \Z_{g'} \cap \mathcal{Y}\cap \mathcal{O}_{j^*}^c} \hat{t}_{k^*}(\bb y) \\
         & \quad \quad \quad \quad \quad \left[ \underbrace{\hat{\mu}_{k^* j^*}(\mathcal{Z}) -\hat{{\mu}}_{k^*j^*}(\mathcal{X})}_{\text{bounded}} + \hat{\bb{\Sigma}}_{reg,k^*}^{(j^*| \hat{\bb{w}}_{\bb y})} \left(\hat{\bb{\Sigma}}_{reg,k^*}^{(\hat{\bb{w}}_{\bb y}| \hat{\bb{w}}_{\bb y})} \right)^{-1} \left( \bb y^{(\hat{\bb{w}}_{\bb y})} - \hat{\bb{\mu}}_{k^*}(\mathcal{Z})^{(\hat{\bb{w}}_{\bb y})}\right) \right]. 
    \end{align*} 
   Due to Corollary~\ref{theorem:corollary} part a, the objective function of the uncontaminated sample is finite and due to Theorem~\ref{theorem:breakdownpoint}, part b.~and c., the estimated covariances on the uncontaminated sample are bounded and regular. Since we assume that the location estimator did not break down, 
   variables cannot be separated. 
    Thus, for all $\x \in \mathcal{X}$ there exists $k$ such that $|\x^{(\w)}-\hat{\bb{\mu}}_{k}^{(\w)}(\mathcal{X})|$ bounded for all feasible $\bb w$ -- otherwise the objective function would explode -- and thus, if $|\x^{(\w)}-\hat{\bb{\mu}}_{l}^{(\w)}(\mathcal{X})| \rightarrow \infty$ for $l \neq k$ it follows that $\hat{t}_{l}(\x) \rightarrow 0 $ and $t_{l}(\x)(\x^{(\w)}-\hat{\bb{\mu}}_{l}^{(\w)}(\mathcal{X})) \rightarrow 0$. 
    Thus, all subtraction parts marked with $*$ are bounded. The last term $\hat{t}_{k^*}(\bb y)\left(\bb y^{(\hat{\bb{w}}_{\bb y})} - \hat{\bb{\mu}}_{k^*}(\mathcal{Z})^{(\hat{\bb{w}}_{\bb y})}\right)$ is also bounded, since outliers are only outlying in variable $j^*$ and otherwise they are part of one cluster. Thus, with the same argument as for uncontaminated data, the term is bounded.
    
    Since $\hat{t}_{k^*}(\tilde{\bb y}) \geq 1/N$ and $\tilde{\bb y} \in \Z_{g'} \cap \mathcal{Y}\cap \mathcal{O}_{j^*}$ the whole sum of $\in \Z_{g'} \cap \mathcal{Y}\cap \mathcal{O}_{j^*}$ goes to minus infinity. To enable the equality of both sides, at least one of the covariances needs to explode (in variable $j^*$) to counteract the exploding sum.        
\end{proof}

\section{Simulation Study: Additional Details}\label{app:simulations}

\subsection{Benchmarks: Additional Details} \label{app:benchmarks}
We provide additional implementation details on the benchmark methods.

\textbf{cellMCD.} The calculation of the cellMCD is stopped at the initialization stage if too many marginal outliers are present or if $p$ is larger than $n_i$, in which case the runs are not included for this estimator.
The calculation stopped at the initialization stage for at most 21\% of the simulation runs across all considered simulation scenarios 1 to 4 with the ALYZ covariance structure, and for scenario 5 no runs are completed. Note that we also ran experiments with a Toeplitz covariance structure \citep[similar to][]{Raymaekers2023}. In those settings, cellMCD was oftentimes more competitive to cellMG-GMM but the problem of failed simulation runs was more pronounced. Results are available upon request.  

\textbf{OC.} \cite{Ollerer2015} only provide a cellwise robust, supervised covariance estimator, but no location estimate. Furthermore, cellwise outliers are not flagged as part of the estimation process, hence no results on outlier detection are included for this benchmark.

\textbf{cellGMM.} The cellGMM encounters internal errors during the computation (for scenario 4 often and for scenario 5 always), that are likely linked to increased singularity issues, in which case the runs are not included for this estimator.

Finally, for \textbf{mclust} and \textbf{cellGMM} there is no clear attribution of an estimated cluster to a group. Thus, both will only be calculated for the two-group settings and clusters will be assigned to groups in the most favorable way. In particular, the assignment of groups to clusters is such that it minimizes the 
KL-divergence. It is possible that the performance of estimating locations might suffer for the considered performance criteria.

\subsection{Additional Simulation Results for ALYZ Structure} \label{app:si:results}

\begin{figure}[H]
    \centering
    \includegraphics[width=0.9\textwidth, trim = {0cm 0.3cm 0cm 0.25cm}, clip]{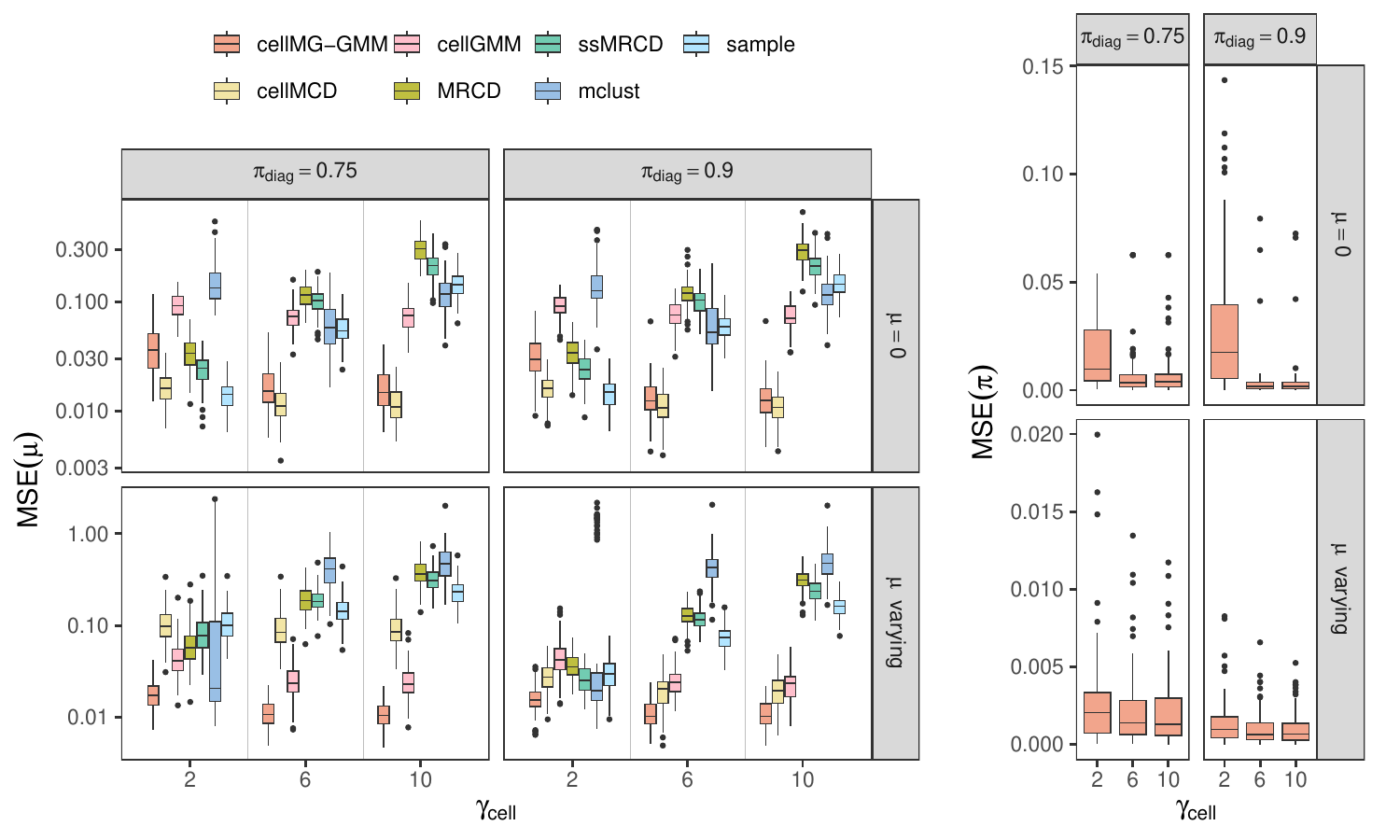}
    \caption{Parameter estimates for Scenario 1
    ($N = 2, p = 10, n_1 = n_2 = 100$).
     In the left panel MSE of the means $\bb \mu_k$, in the right of the mixture probabilities $\bb \pi$.}
    \label{fig:sim_balanced2_ALYZCOR_param}
\end{figure}

\begin{figure}[H]
    \centering
    \includegraphics[width=0.9\textwidth, trim = {0cm 0.3cm 0cm 0.25cm}, clip]{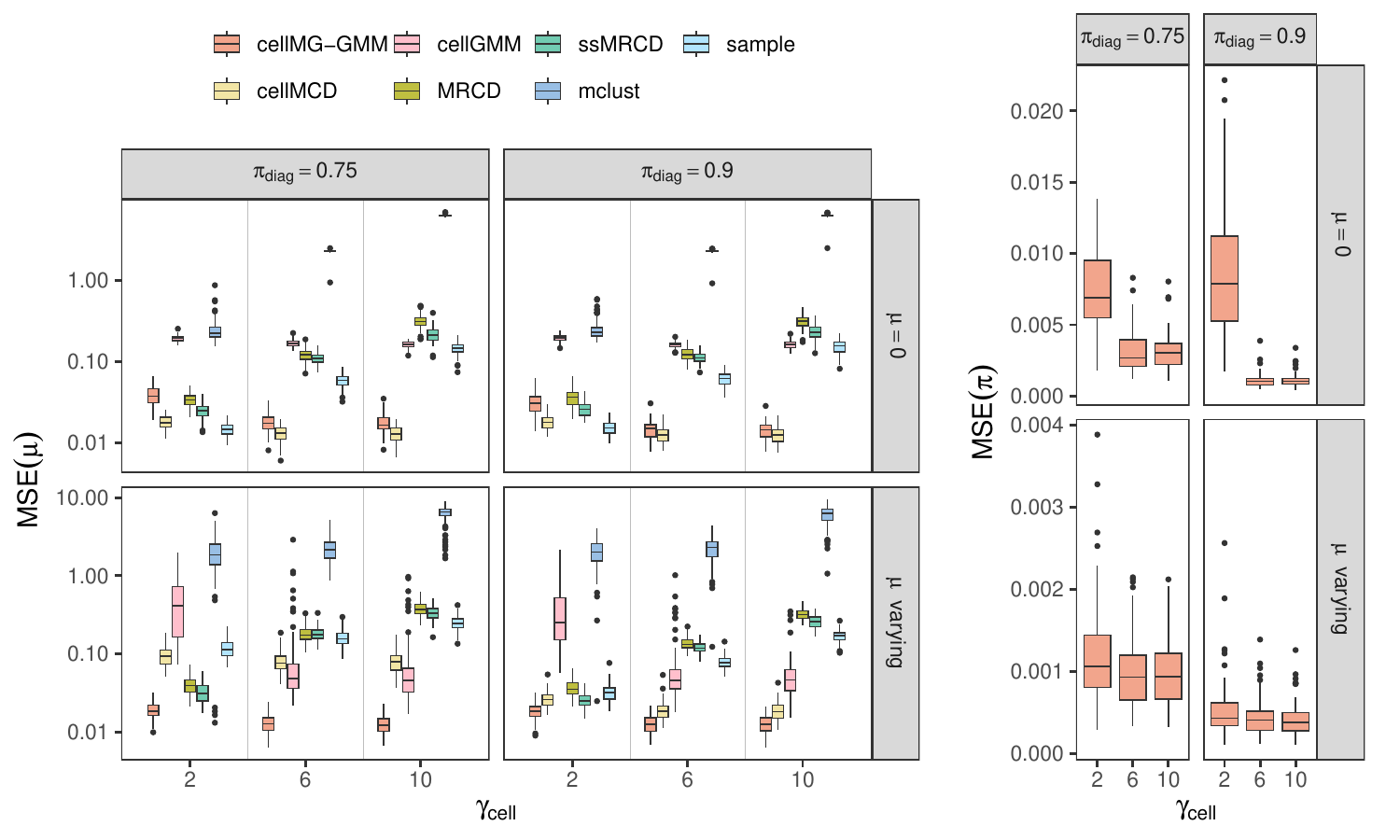}
    \caption{Parameter estimates for Scenario 2 
    ($N = 5, p = 10, n_i = 100$).
    In the left panel MSE of the means $\bb \mu_k$, in the right of the mixture probabilities $\bb \pi$.}
    \label{fig:sim_balanced5_ALYZCOR_param}
\end{figure}

\begin{figure}[H]
    \centering
     \includegraphics[width=\textwidth]{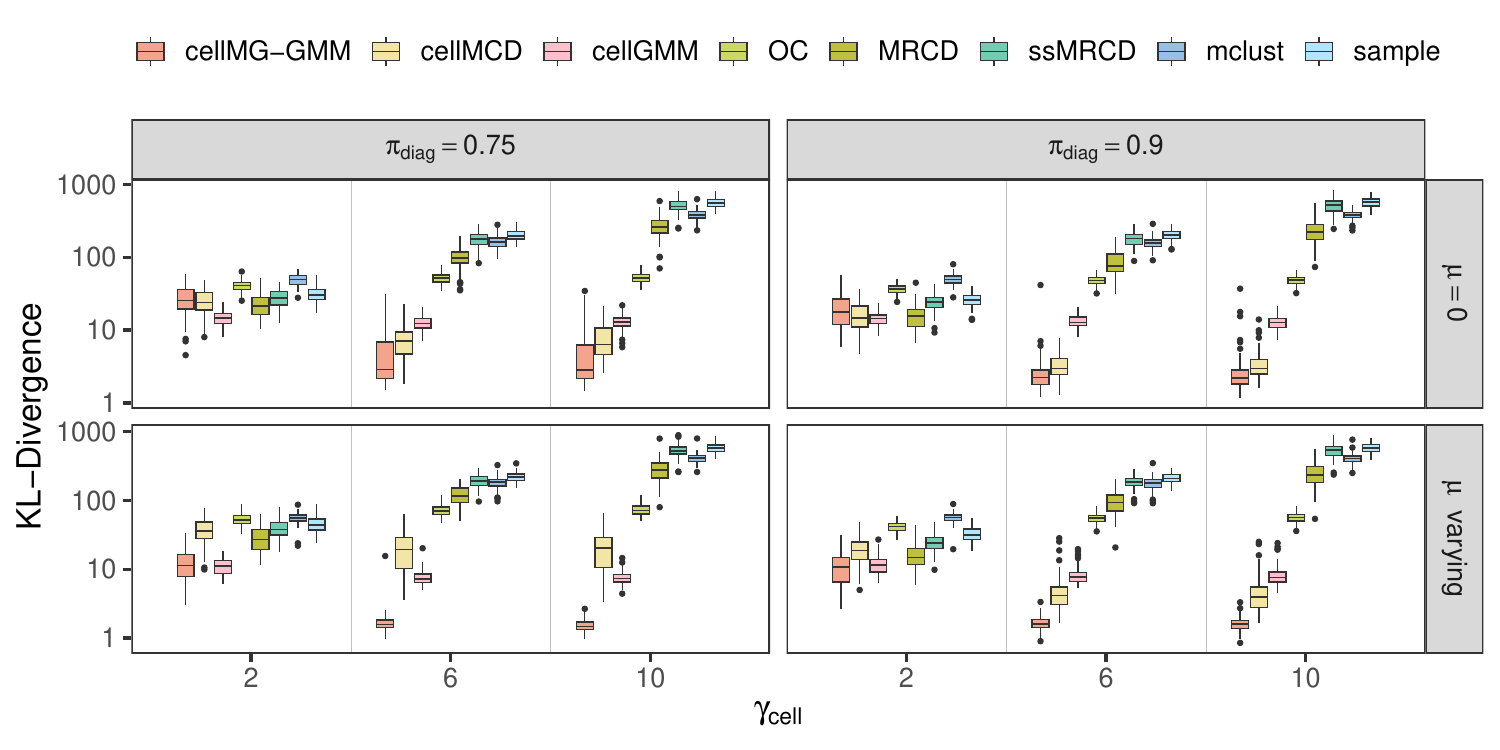}
    \includegraphics[width=\textwidth]{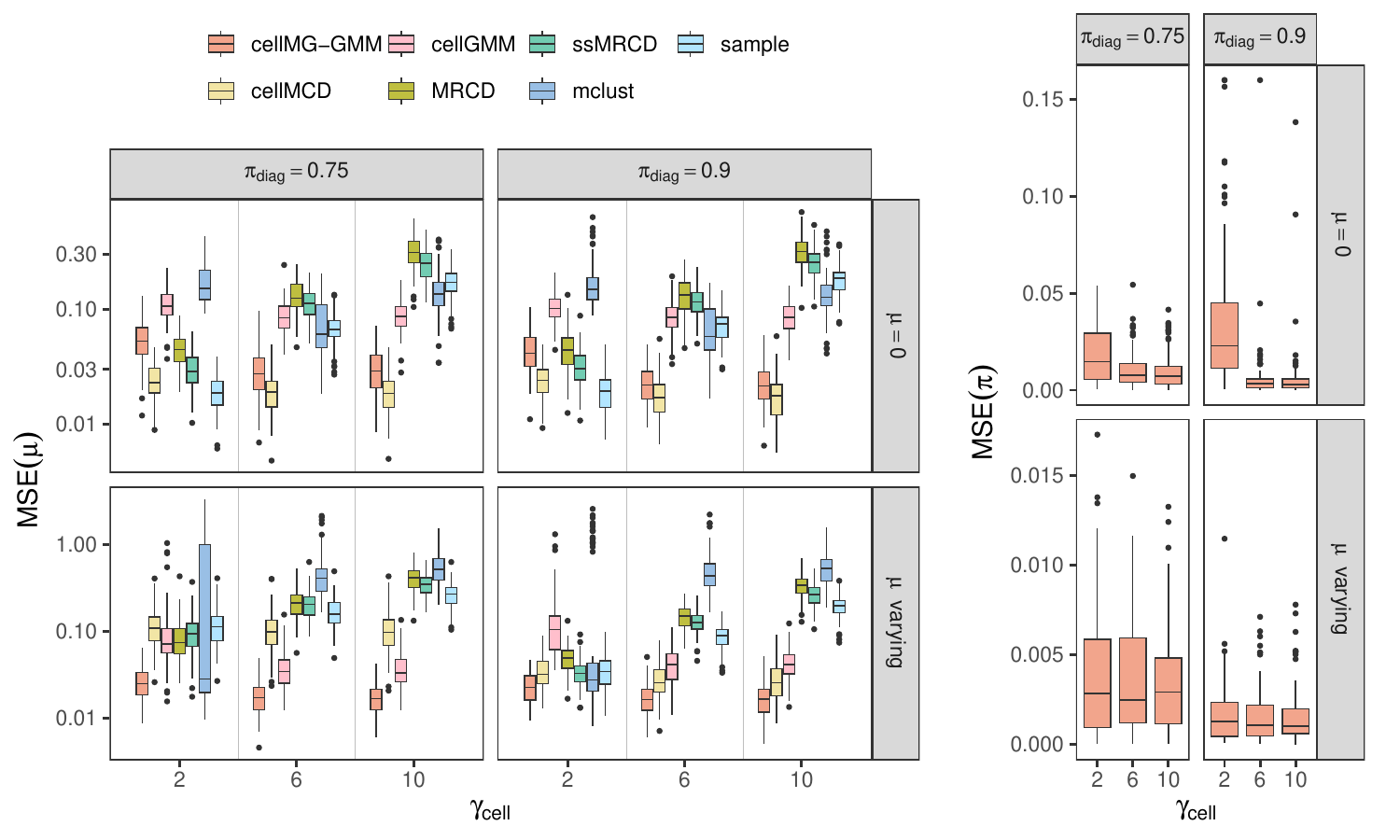}
    \caption{Parameter estimates for scenario 3, the unbalanced groups ($N = 2, p= 10, n_1 = 100, n_2 = 50$) with ALYZ structured covariance matrices. Top panel: KL-divergence of the covariance estimates. Bottom left panel: MSE of the mean estimation. Bottom right panel: MSE of the mixture probabilities $\bb \pi$.}
    \label{fig:sim_unbalanced_ALYZ_param}
\end{figure}

\begin{figure}[H]
    \centering
     \includegraphics[width=\textwidth]{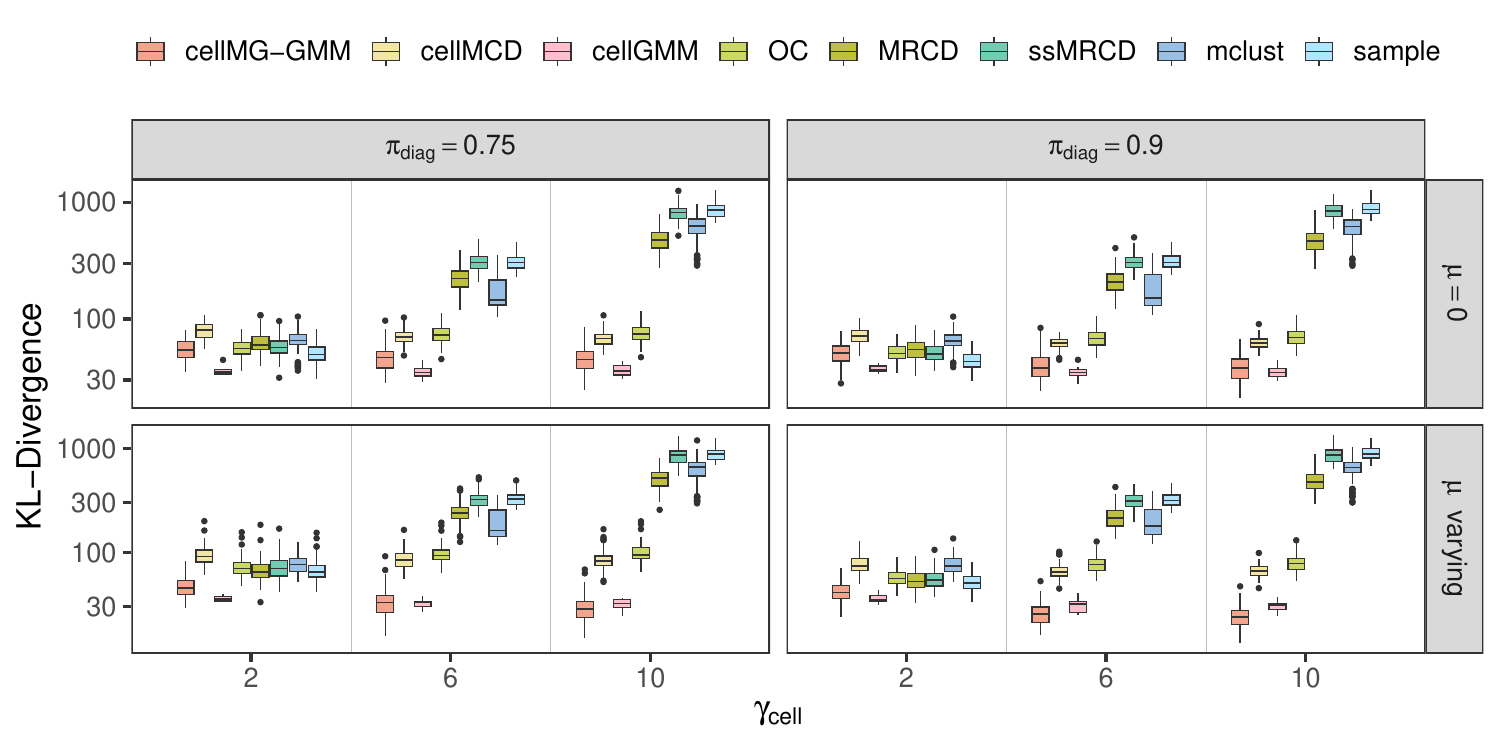}
    \includegraphics[width=\textwidth]{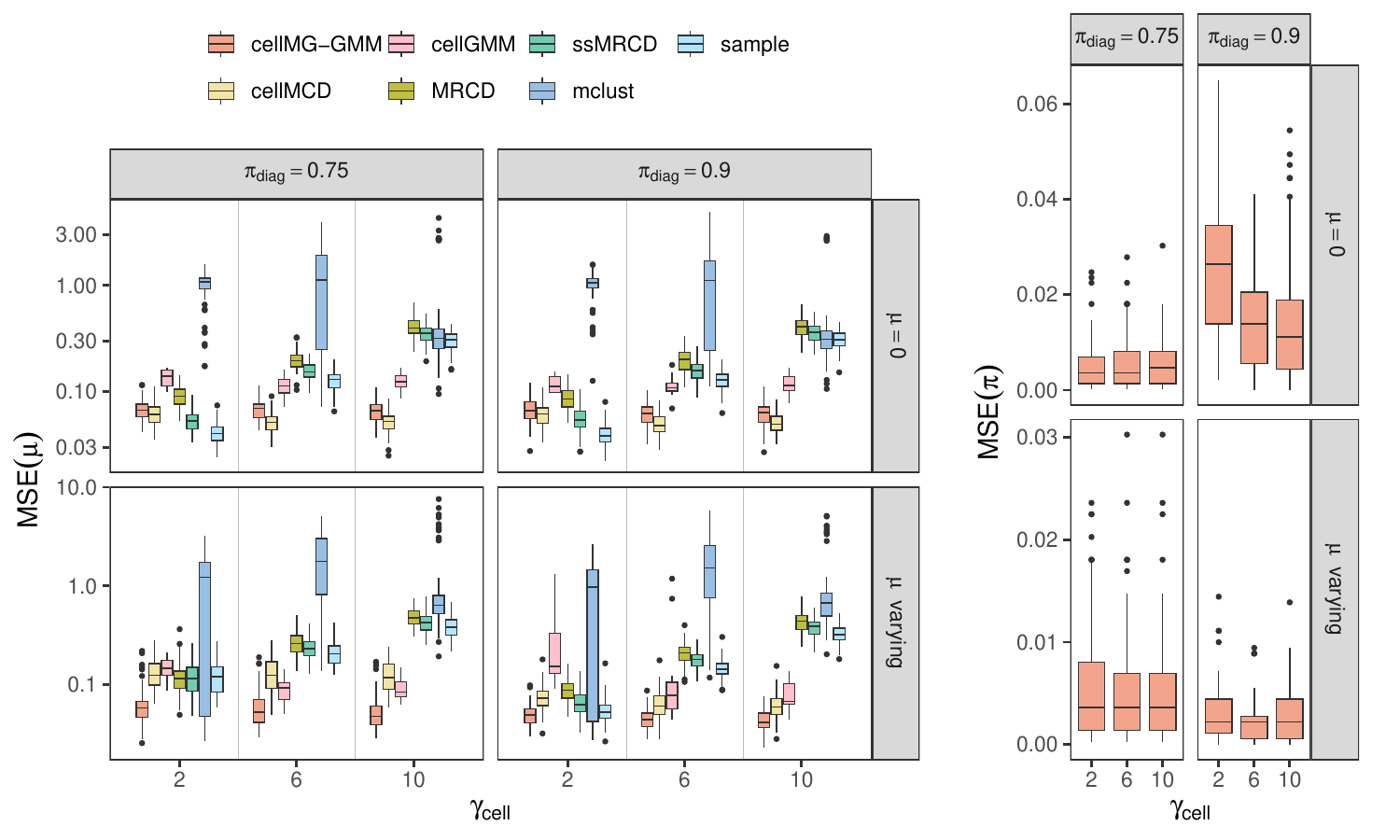}
    \caption{Parameter estimates for scenario 4 ($N = 2, p= 30, n_1 = n_2 = 40$) with ALYZ structured covariance matrices. Top panel: KL-divergence of the covariance estimates. Bottom left panel: MSE of the mean estimation. Bottom right panel: MSE of the mixture probabilities $\bb \pi$.}
    \label{fig:sim_mediumdim_ALYZ_param}
\end{figure}

\begin{figure}[H]
    \centering
     \includegraphics[width=\textwidth]{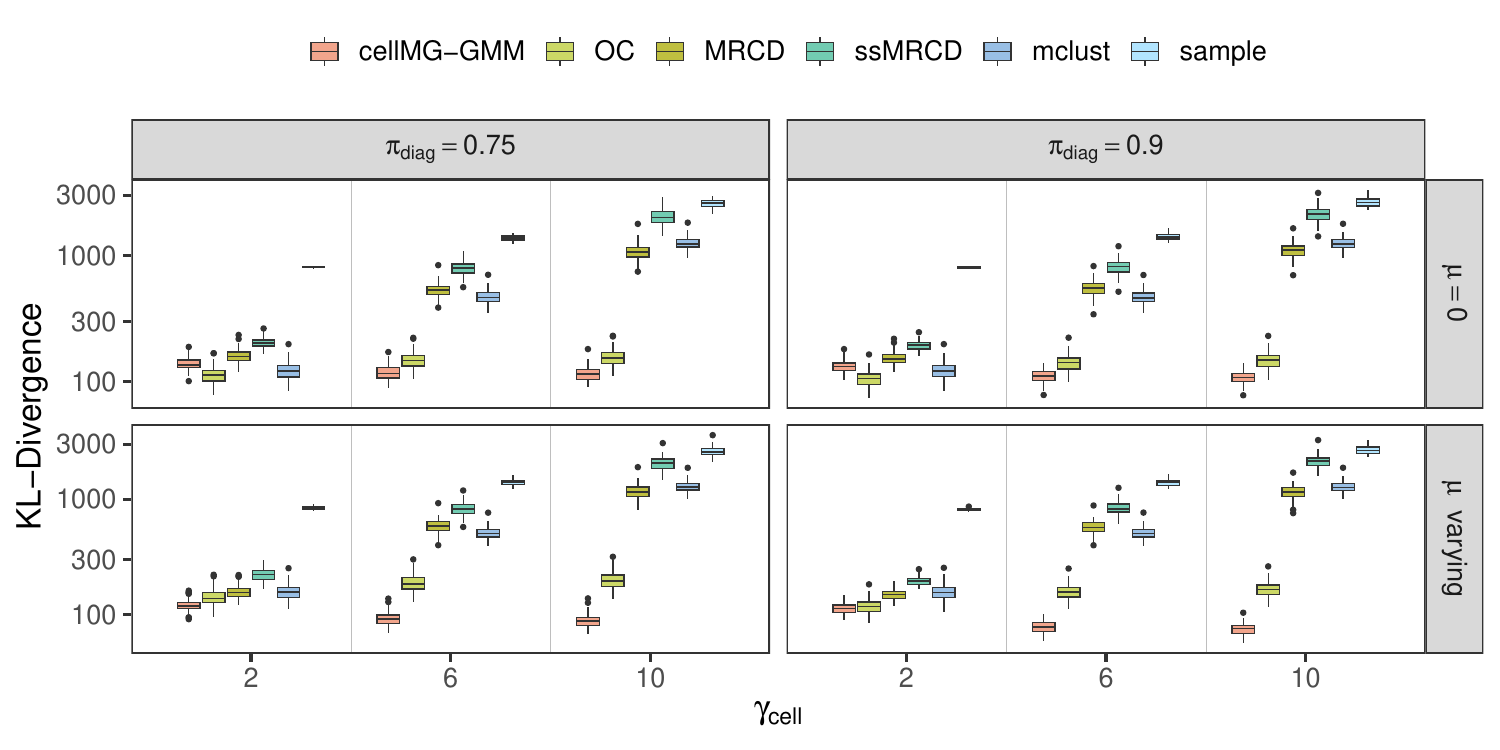}
    \includegraphics[width=\textwidth]{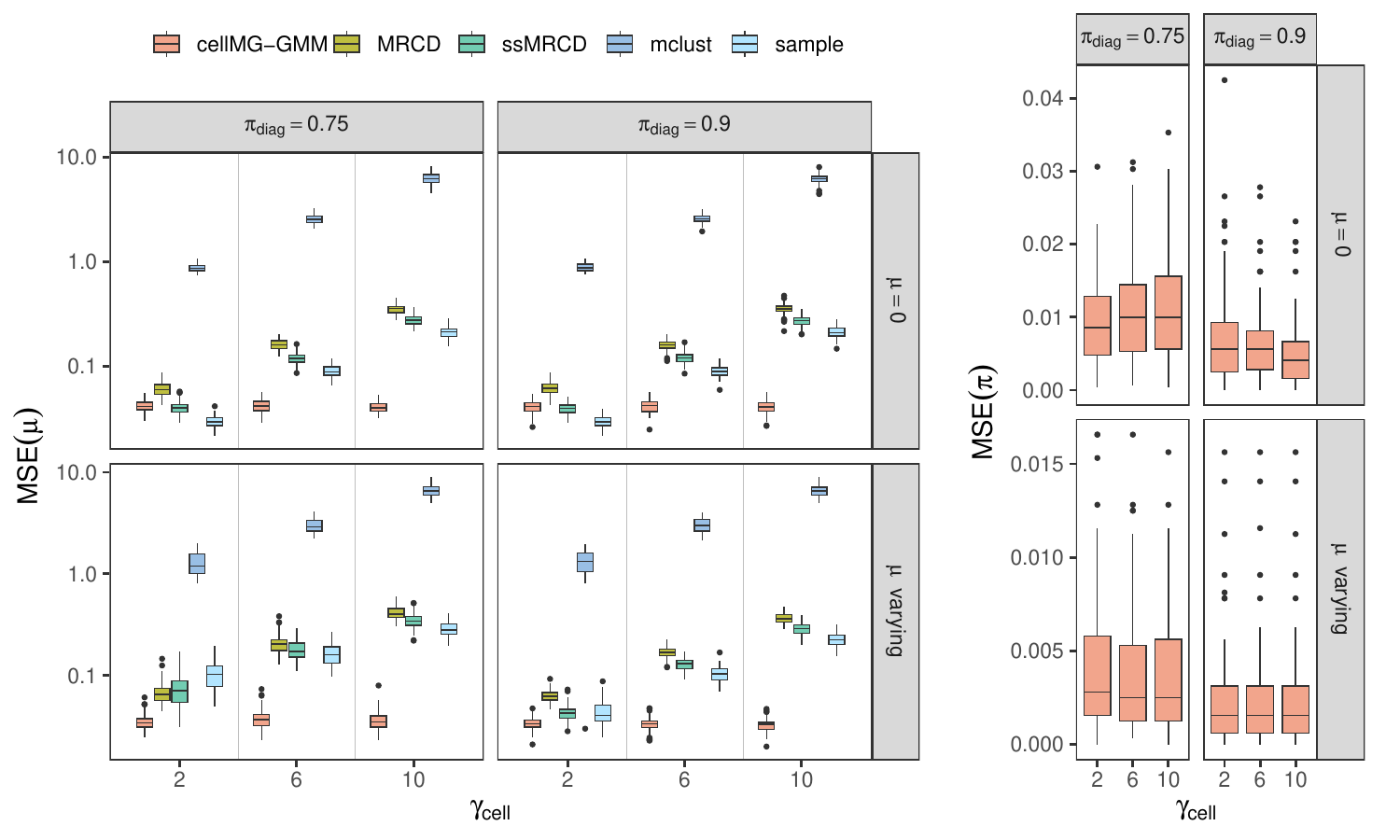}
    \caption{Parameter estimates for scenario 5 ($N = 2, p= 60, n_1 = n_2 = 40$) with ALYZ structured covariance matrices. Top panel: KL-divergence of the covariance estimates. Bottom left panel: MSE of the mean estimation. Bottom right panel: MSE of the mixture probabilities $\bb \pi$. No calculations of cellMCD and cellGMM were successful and are thus not included.}
    \label{fig:sim_highdim_ALYZ_param}
\end{figure}

\begin{figure}[H]
    \centering
    \includegraphics[width=\textwidth]{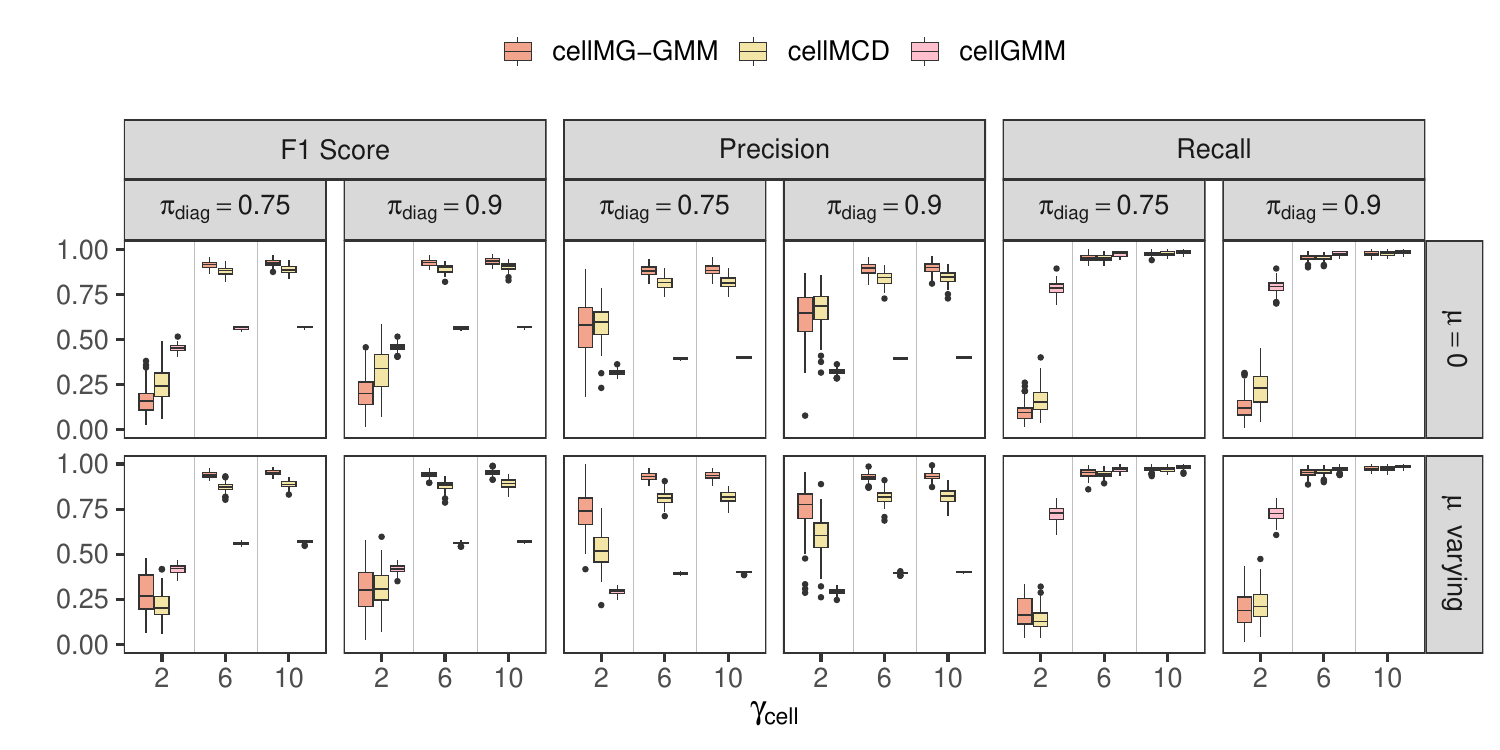}
    \caption{Performance of cellwise outlier detection in scenario 3, unbalanced groups ($N = 2, p= 10, n_1 = 100, n_2 = 50$) with ALYZ structured covariances evaluated by precision, recall and F1-score.}
    \label{fig:sim_unbalanced_ALYZ_W}
\end{figure}

\begin{figure}[H]
    \centering
    \includegraphics[width=\textwidth]{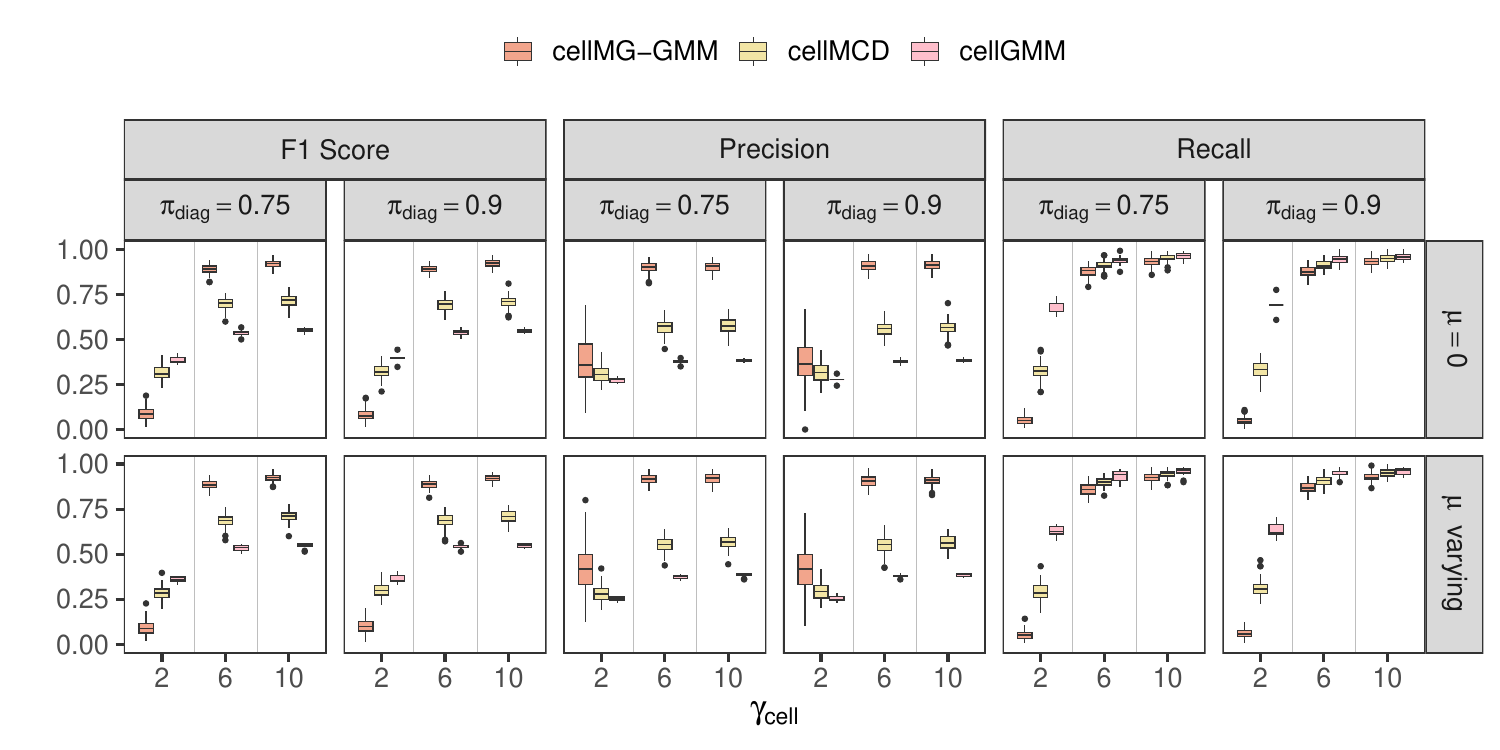}
    \caption{Performance of cellwise outlier detection in scenario 4 ($N = 2, p= 30, n_1 = n_2 = 40$) with ALYZ structured covariances evaluated by precision, recall and F1-score.}
    \label{fig:sim_mediumdim_ALYZ_W}
\end{figure}

\begin{figure}[H]
    \centering
    \includegraphics[width=\textwidth]{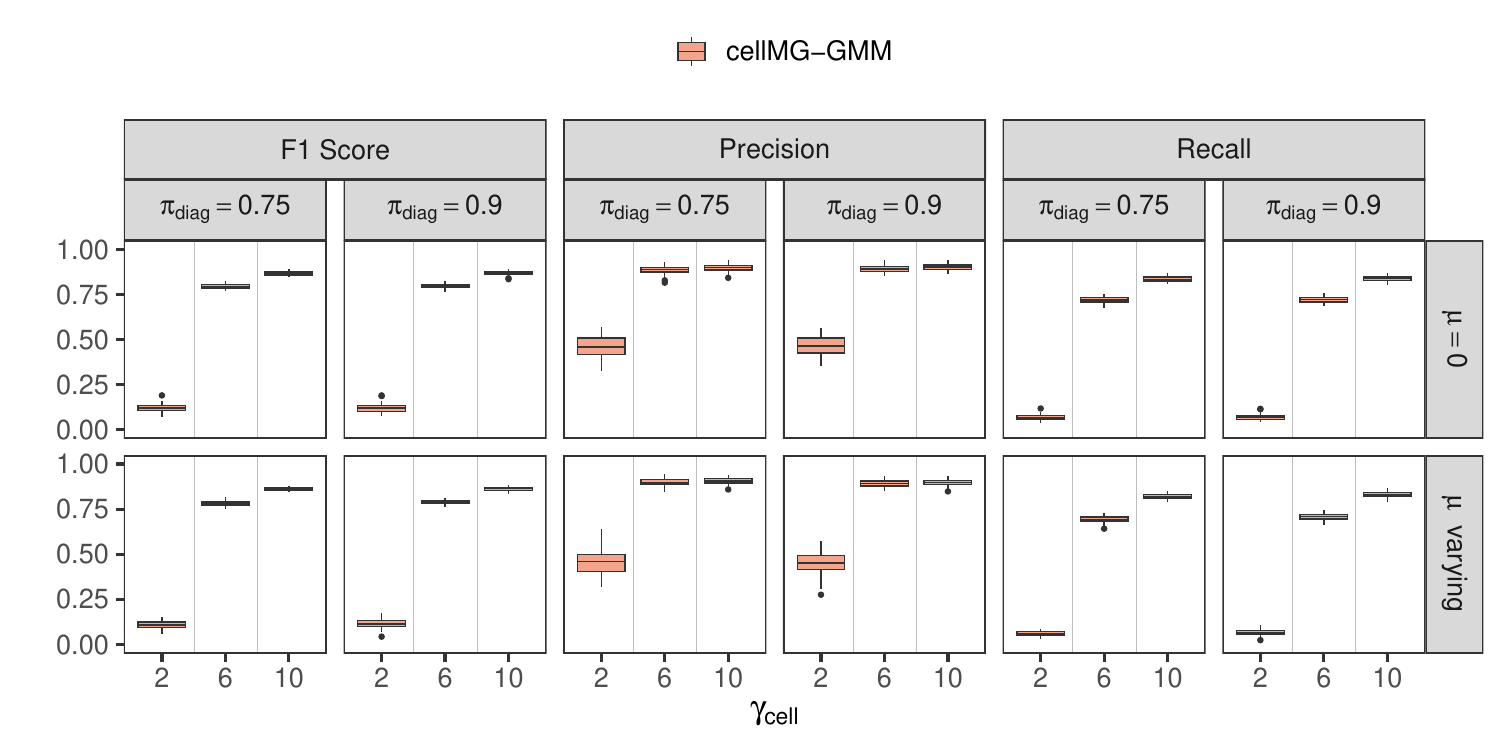}
    \caption{Performance of cellwise outlier detection in scenario 5 ($N = 2, p= 60, n_1 = n_2 = 40$) with ALYZ structured covariances evaluated by precision, recall and F1-score. No calculations of cellMCD and cellGMM were successful and are thus not included.}
    \label{fig:sim_highdim_ALYZ_W}
\end{figure}

\section{Application: Austrian Weather Data} 
\label{app:weather}

We  
use data 
of \citet{Geosphere2022}, with $p = 6$ monthly measured weather variables (averaged over the year 2021) at $n=183$ Austrian weather stations, including air pressure (p),  temperature (t), amount of rain (rsum), relative humidity (rel), hours of sunshine (s) and wind velocity (vv). 
The data set is available in the R-package \texttt{ssMRCD} \citep{ssMRCD_Cran} under the name \texttt{weatherAUT2021}. Figure~\ref{fig:weather_map_and_outliers}(a) shows the spatial locations and the underlying diverse geographical and 
meteorological structure of
the Alps.
We use this initial information to partition
the stations into $N =5$ more coherent groups, separated 
by the dashed lines
on the altitude map. The most western area (group 1, $n_1 = 31$) is characterized by  mountainous terrain, which extends to the east into group 2 ($n_2 = 80$) with high and low mountains. The most northern part (group 3, $n_3 = 35$) consists of low mountains and hills along the Danube river which flows through Vienna and the Vienna Basin (group 5, $n_5 = 21$). The  area to the East (group 4, $n_4 = 16$) 
is mainly flat.

\begin{figure}[t]
    \centering
    \begin{subfigure}[b]{0.55\textwidth}
        \includegraphics[height=0.33\textheight, trim = {1cm -1cm 1.1cm 0cm}, clip]{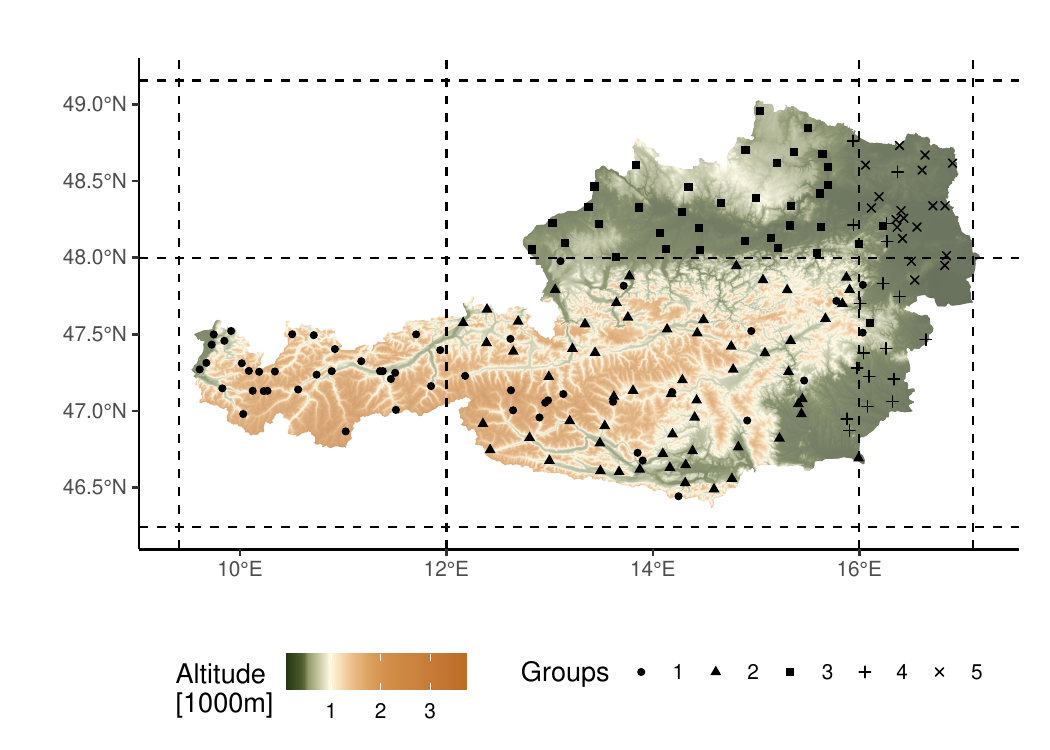}
        \caption{Altitude map.}
    \end{subfigure}
    \hfill
    \begin{subfigure}[b]{0.4\textwidth}
        \includegraphics[height=0.35\textheight, trim = {3.3cm 0cm 2.9cm 0cm}, clip]{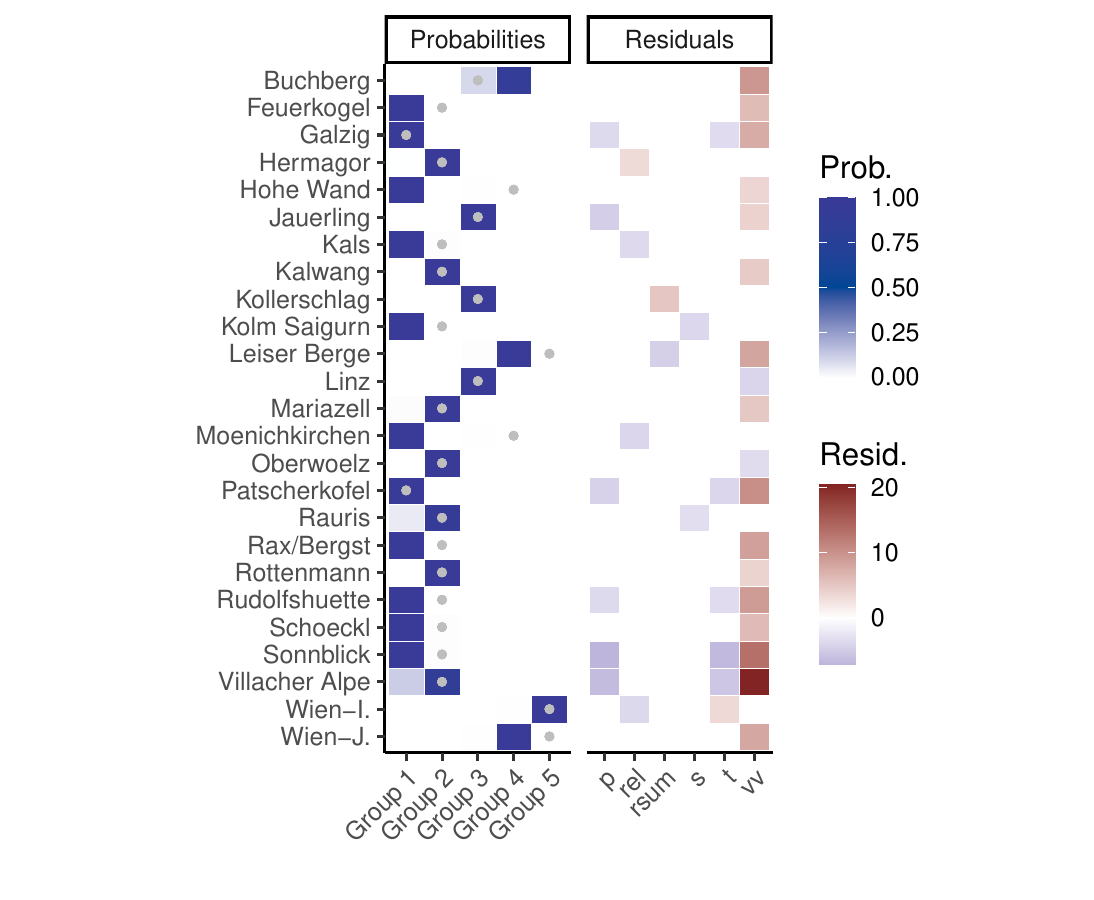}
     \caption{Group probabilities and residuals.}
    \end{subfigure}
    \caption{
    Left: Altitude map of Austria with $n=183$ weather stations separated into $N=5$ groups by the grid lines. Each station is (re-)assigned to a group, indicated by the different symbols, based on its maximal class probability.
    Right: Outlying weather stations (rows) with group probabilities $\hat{t}_{g,i,k}$ with dots at the initial groups in the left panel; and cell residuals in the right panel.} 
    \label{fig:weather_map_and_outliers}
\end{figure}

Our goal is to identify discrepancies between each station's predefined spatial label and its model-based grouping using the MG-GMM, to identify cellwise outliers and analyze why these occur.
We apply 
cellMG-GMM with $h_g = 0.75n_g$, allowing for up to $25\%$ of flagged cells per variable, and $\alpha = 0.5$, 
allowing for very flexible group re-assignments. The model-based grouping structure, based on each station's 
highest class probability $\max_k \hat{t}_{g,i,k}$, 
is shown on the altitude map of Figure~\ref{fig:weather_map_and_outliers}(a) through the 
different plotting symbols.
In Figure~\ref{fig:weather_map_and_outliers}(b), we  display
observations (in the rows) with at least one flagged cell. The color of each tile in the left panel shows the estimated class probabilities $\hat{t}_{g,i,k}$, while the initial group membership is marked by a dot. 
Here, cellMG-GMM identifies observations 
that are outlying 
in their initial group. Such stations can be observed from the left panel of Figure~\ref{fig:weather_map_and_outliers}(b), by their high probability of belonging to another group, and thus the mismatch between their initial group (dot) and dark blue tile (re-assigned group). For example, the weather station Hohe Wand is originally assigned to group 4 - a group of observations in a mostly flat area - but the weather station is located above 900m altitude and is actually very exposed. The model suggests that the group of high alpine weather stations (group 1) would be a better fit for Hohe Wand.

In the right panel of Figure~\ref{fig:weather_map_and_outliers}(b), outlying cells are colored according to their standardized residuals. Positive residual values indicate that the observed value is higher than what would be expected, vice versa for negative values.
cellMG-GMM can also identify 
observations that are outlying across all groups, as indicated by a high number of cellwise outliers (e.g.\ half of the cells being outlying).
Many outlying stations are connected to cell outliers in the variable wind velocity (vv), likely due to the diverse exposure of weather stations, even 
for stations in the same area. 
The five weather stations Villacher Alpe, Sonnblick, Rudolfsh\"utte, Patscherkofel, and Galzig have several outyling cells and display unexpected high values in wind velocity (vv) and low values in air pressure (p) and temperature (t). These are exactly the five highest weather stations with an altitude of more than $2000$ meters.

Finally, Figure~\ref{fig:weather_tempvswind} presents a more detailed analysis of the variables wind velocity and air temperature. The tolerance ellipses, based on the estimated locations and covariance matrices per group, show a smooth transition from groups connected to mountainous landscapes (group 1 and 2) that display higher variation in temperature to flatter landscapes (group 3 to 5) that display increased variation in wind velocity and generally higher temperatures. The weather station Wien-IS is the only cellwise outlier with unexpectedly high temperature, it is located in the city center of the capital Vienna. 

\begin{figure}[t]
    \centering
    \includegraphics[width=\textwidth, trim = {0cm 0.5cm 0cm 1.8cm}, clip]{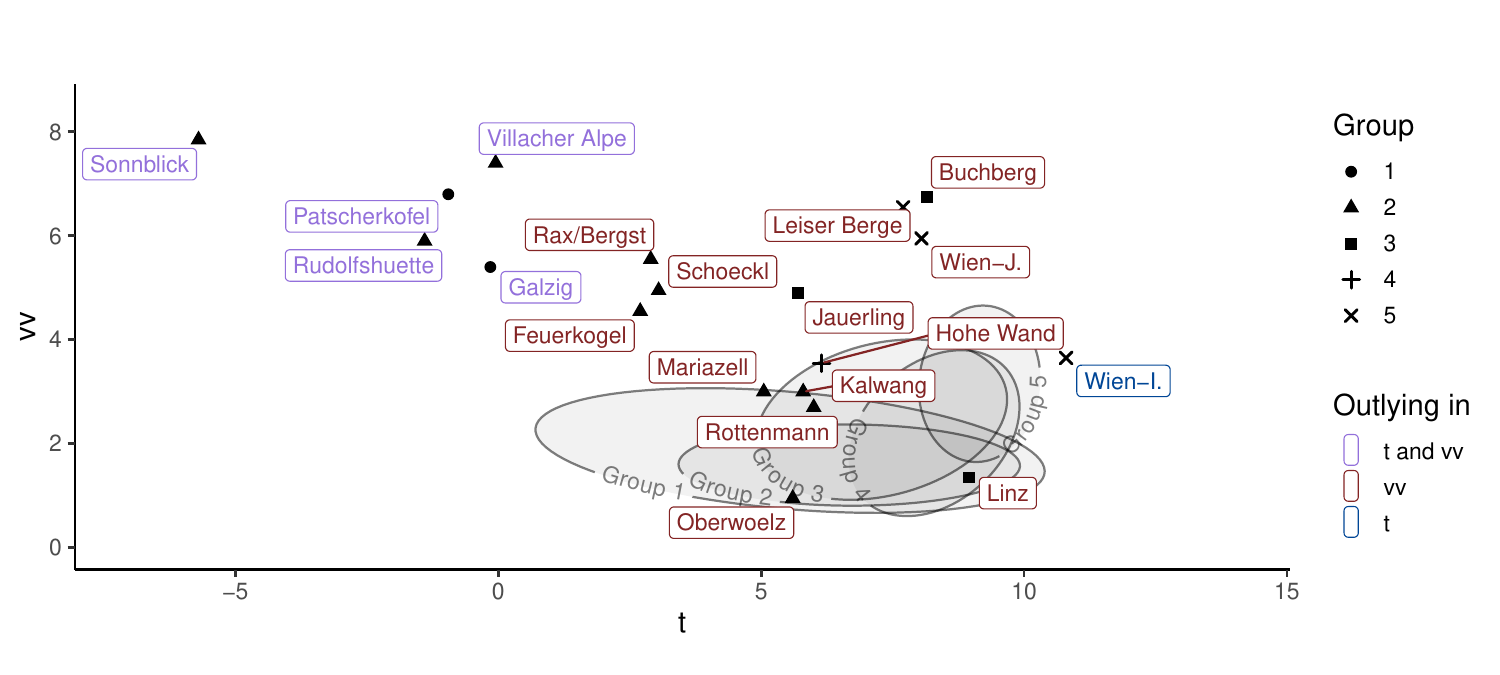}
    \caption{Bivariate feature space of wind velocity (vv) and air temperature (t). The $95\%$ tolerance ellipses are based on the estimated locations and covariance matrices per group. Stations outlying in at least one of the two variables are displayed. Shapes correspond to the initial group of each station, the color of the label indicates which cells are outlying.}
    \label{fig:weather_tempvswind}
\end{figure}

\end{document}